\documentclass[nolimits,final]{lmcs}
\pdfoutput=1

\usepackage{lastpage}
\lmcsdoi{15}{3}{1}
\lmcsheading{}{\pageref{LastPage}}{}{}%
{Dec.~29,~2017}{Jul.~04,~2019}{}

\usepackage{etex}
\newcommand{\tuple}[1]{( {#1} )}

\newcommand{\sgl}{\{\cdot\}}

\usepackage{bm}
\usepackage{bbold}
\newcommand{\Nat}{\mathbb{N}}  

\usepackage{microtype} 

\usepackage{proof}                              %
\usepackage{bussproofs}                         %
\usepackage{amssymb}                            %
\usepackage{wasysym}                            %

\usepackage{graphicx}
\usepackage{mathrsfs}
\usepackage{subcaption}
\usepackage{tikz}
\usepackage{tkz-euclide}
\usetkzobj{all}
\usepackage{tikz-cd}

\usepackage{relsize}

\usetikzlibrary{
  fit,%
  calc,%
  arrows,%
  arrows.meta,%
  intersections,%
  shapes.misc,%
  shapes.arrows,%
  patterns,%
  automata,%
  chains,%
  matrix,%
  positioning,%
  scopes,%
  decorations.markings,%
  decorations.pathmorphing,%
  external
}

\tikzset{
commutative diagrams/.cd,
arrow style=tikz,
diagrams={>=stealth},
row sep=large, 
column sep = huge
}

\tikzset{shiftarr/.style={
        rounded corners,%
        to path={--([#1]\tikztostart.center)
                     -- ([#1]\tikztotarget.center) \tikztonodes
                     -- (\tikztotarget)},
}}

\usepackage{optparams}                          %
\usepackage{xspace}

\usepackage{ifdraft}                            %

\usepackage{dsfont}

\usepackage{needspace}

  {\gdef\scalefactor{#1}\begin{center}\proofSkipAmount \leavevmode}%
  {\scalebox{\scalefactor}{\DisplayProof}\proofSkipAmount \end{center} }

\let\by\relax

\providecommand{\catname}{\mathbf} 
\providecommand{\clsname}{\mathcal}
\providecommand{\oname}[1]{\operatorname{\mathsf{#1}}}

\def\defcatname#1{\expandafter\def\csname B#1\endcsname{\catname{#1}}}
\def\defcatnames#1{\ifx#1\defcatnames\else\defcatname#1\expandafter\defcatnames\fi}
\defcatnames ABCDEFGHIJKLMNOPQRSTUVWXYZ\defcatnames

\def\defclsname#1{\expandafter\def\csname C#1\endcsname{\clsname{#1}}}
\def\defclsnames#1{\ifx#1\defclsnames\else\defclsname#1\expandafter\defclsnames\fi}
\defclsnames ABCDEFGHIJKLMNOPQRSTUVWXYZ\defclsnames

\def\defbbname#1{\expandafter\def\csname BB#1\endcsname{\mathbb{#1}}}
\def\defbbnames#1{\ifx#1\defbbnames\else\defbbname#1\expandafter\defbbnames\fi}
\defbbnames ABCDEFGHIJKLMNOPQRSTUVWXYZ\defbbnames

\def\Set{\catname{Set}}

\providecommand{\argument}{\operatorname{-\!-}}

\DeclareOldFontCommand{\bf}{\normalfont\bfseries}{\mathbf}
\providecommand{\FSet}{{\mathcal P}_{\omega}}		%
\providecommand{\CSet}{{\mathcal P}_{\omega_1}}		%
\providecommand{\Id}{\operatorname{Id}}

\providecommand{\Hom}{\mathsf{Hom}}
\providecommand{\id}{\mathsf{id}}

\providecommand{\comp}{\mathbin{\circ}}
\providecommand{\iso}{\mathbin{\cong}}

\providecommand{\bang}{\operatorname!}				%

\providecommand{\mto}{\mapsto}
\providecommand{\xto}[1]{\,\xrightarrow{#1}\,}

\providecommand{\dar}{\kern-1.2pt\operatorname{\downarrow}}	
\providecommand{\uar}{\kern-1.2pt\operatorname{\uparrow}}

\providecommand{\pr}{\oname{pr}}

\providecommand{\brks}[1]{\langle #1\rangle}

\providecommand{\inl}{\oname{inl}}
\providecommand{\inr}{\oname{inr}}
\providecommand{\inj}{\oname{in}}

\DeclareSymbolFont{Symbols}{OMS}{cmsy}{m}{n}
\DeclareMathSymbol{\iobj}{\mathord}{Symbols}{"3B}

\usepackage{stmaryrd}

\providecommand{\comma}{,\operatorname{}\linebreak[1]}		%
\providecommand{\dash}{\nobreakdash-\hspace{0pt}}		%

\providecommand{\by}[1]{\text{/\!\!/~#1}}			%
\providecommand{\pacman}[1]{}					%

\def\undef#1{\let #1\relax}

\providecommand{\noqed}{\def\qed{}}				%

\providecommand{\mone}{{\text{\kern.5pt\rmfamily-}\mathsf{\kern-.5pt1}}}

\makeatletter
\@ifpackageloaded{enumitem}{}{\usepackage[loadonly]{enumitem}}		%
\makeatother

\newlist{citemize}{itemize}{1}
\setlist[citemize]{label=\labelitemi,wide} %

\newlist{cenumerate}{enumerate}{1}
\setlist[cenumerate,1]{label=\arabic*.~,ref={\arabic*},wide} %

\makeatletter
\def\mfix#1{\oname{#1}\@ifnextchar\bgroup\@mfix{}}	%
\def\@mfix#1{#1\@ifnextchar\bgroup\mfix{}}			%
\makeatother

\providecommand{\case}[3]{\mfix{case}{\mathbin{}#1}{of}{#2}{\kern-1pt;}{\mathbin{}#3}}

\setlist[cenumerate,1]{label=(\arabic*),ref={(\arabic*)},wide} %

\newenvironment{theorem}{\begin{thm}}{\end{thm}}
\newenvironment{corollary}{\begin{cor}}{\end{cor}}
\newenvironment{lemma}{\begin{lem}}{\end{lem}}
\newenvironment{proposition}{\begin{prop}}{\end{prop}}
\newenvironment{assumption}{\begin{asm}}{\end{asm}}

\newenvironment{remark}{\begin{rem}}{\end{rem}}
\newenvironment{example}{\begin{exa}}{\end{exa}}
\newenvironment{definition}{\begin{defi}}{\end{defi}}

\usepackage[utf8]{inputenc}
\usepackage{savesym}

\savesymbol{degree}
\savesymbol{leftmoon}
\savesymbol{rightmoon}
\savesymbol{fullmoon}
\savesymbol{newmoon}
\savesymbol{diameter}
\savesymbol{langle}
\savesymbol{rangle}

\usepackage[matha,mathx]{mathabx}

\restoresymbol{varlangle}{langle}
\restoresymbol{varrangle}{rangle}

\usepackage{xcolor}
\newcommand{\gray}[1]{
#1}
\newcommand{\myparagraph}[1]{\smallskip\par\noindent\textbf{\textbf{#1}}\hspace{6pt}}

\newcommand{\klstar}{\star}  				%
\newcommand{\kklstar}{\text{\kreuz}}  	    %
\newcommand{\istar}{\dagger}  				%
\newcommand{\iistar}{\ddagger}  			%
\renewcommand{\comp}{\,}					%

\newcommand{\del}{\operatorname{\rhd}}

\newcommand{\bmu}{\bm\mu}
\newcommand{\bta}{\bm\eta}

\newcommand{\module}[1]{M} %
\newcommand{\moduleS}[1]{N} %
\newcommand{\modstar}{{\circ}}
\newcommand{\mediator}[2]{[{#2}, {#1}]}
\newcommand{\modstarS}{{\bullet}}

\renewcommand{\inl}{\inj_1}
\renewcommand{\inr}{\inj_2}

\providecommand{\out}{\operatorname{\mathsf{out}}}
\providecommand{\inm}{\operatorname{\mathsf{in}}}
\providecommand{\tuo}{\operatorname{\out^{\text{\kern.5pt\rmfamily-}\kern-.5pt1}\kern-1pt}}
\providecommand{\nim}{\operatorname{\inm^{\text{\kern.5pt\rmfamily-}\kern-.5pt1}\kern-1pt}}

\providecommand{\modm}{\xi}

\newcommand{\coit}{\operatorname{\mathsf{coit}}}
\newcommand{\iter}{\operatorname{\mathsf{iter}}}
\newcommand{\IB}[3][\hash]{#2 \mathbin{#1} #3}

\newcommand{\IBnu}{\digamma\kern-3pt{}_{\hash}}
\newcommand{\IBnup}{\digamma\kern-3pt{}_{\hash'}}

\newcommand{\IBnum}[1]{\digamma\kern-3pt{}_{#1}}

\newcommand{\SKIP}{\oname{skip}}
\newcommand{\TRUE}{\oname{true}}

\newcommand{\whileTermS}[2]{\mfix{while\;}{#1}{\;do\;}{#2}}

\newcommand{\cpto}{
  \mathrel{\raisebox{0.5ex}{\kern3pt\ensuremath{\mathrel{\tikz{ \draw [-stealth,line width=0.4] (0.6ex,1ex) -- (0,1ex) -- (0,0.4ex) -- (2.2ex,0.4ex); }}}\kern3pt}}
}

\ifdraft{
 \usepackage{todos}
 \AtEndPreamble{
 \hypersetup{
        pdftex,
        colorlinks,
        linkcolor={blue},
        citecolor={blue},
        urlcolor={blue},
        breaklinks=true,
        final
  }
  \usepackage[layout=footnote,draft]{fixme}
}
}{
 \newcommand{\todo}[1]{}
 \usepackage[layout=footnote,final]{fixme}
}

\newcommand{\SG}[1]{\sgnote{#1}}

\FXRegisterAuthor{ls}{als}{LS} %
\FXRegisterAuthor{sg}{asg}{SG} %
\FXRegisterAuthor{cr}{acr}{CR} %

\begin{document}
\allowdisplaybreaks

\title{Guarded and Unguarded Iteration\\ for Generalized Processes}
\date{}
\author[Sergey Goncharov]{Sergey Goncharov{\rsuper{a}}}
\author[Lutz Schr\"oder]{Lutz Schr\"oder\rsuper{b}}
\author[Christoph Rauch]{Christoph Rauch\rsuper{c}}
\address{\lsuper{a,b,c}Friedrich-Alexander-Universit\"at Erlangen-N{\"u}rnberg}
\thanks{Research supported by the DFG project \emph{A High Level Language for Monad-based Processes} (GO~2161/1\dash 2, SCHR 1118/8-2)}
\email{\{sergey.goncharov,lutz.schroeder,christoph.rauch\}@fau.de}

\author[Maciej Pir\'og]{Maciej Pir\'og\rsuper{d}}
\address{\lsuper{d}Uniwersytet Wroc\l{}awski}
\email{maciej.pirog@cs.uni.wroc.pl}


\maketitle

\begin{abstract} 
Models of iterated computation, such as (completely) iterative monads, often
depend on a notion of guardedness, which guarantees unique solvability of
recursive equations and requires roughly that recursive calls happen only
under certain guarding operations. On the other hand, many models of iteration
do admit unguarded iteration. Solutions are then no longer unique, and in
general not even determined as least or greatest fixpoints, being instead
governed by quasi-equational axioms. Monads that support unguarded iteration
in this sense are called (complete) Elgot monads. Here, we propose to equip
(Kleisli categories of) monads with an abstract notion of guardedness and
then require solvability of abstractly guarded recursive equations; examples
of such \emph{abstractly guarded pre-iterative monads} include both iterative
monads and Elgot monads, the latter by deeming any recursive definition to be
abstractly guarded. Our main result is then that Elgot monads are precisely the
iteration-congruent retracts of abstractly guarded \emph{iterative} monads, the
latter being defined as admitting \emph{unique} solutions of abstractly guarded
recursive equations; in other words, models of unguarded iteration come about by
quotienting models of guarded iteration.
\end{abstract}

\section{Introduction}

\noindent In recursion theory, notions of guardedness traditionally
play a central role. Guardedness typically means that recursive calls
must be in the scope of certain guarding operations, a
condition aimed, among other things, at ensuring progress. The
paradigmatic case are recursive definitions in process algebra, which
are usually called guarded if recursive calls occur only under action
prefixing~\cite{BergstraPonseEtAl01}. A more abstract example are completely
iterative theories~\cite{ElgotBloomEtAl78} and monads~\cite{Milius05}, where,
in the latter setting, a recursive definition is guarded if it factors
through a given ideal of the monad. Guarded recursive definitions
typically have unique solutions; e.g.\ the unique solution of the
guarded recursive definition 
\begin{equation*}
  x =a.\,x
\end{equation*}
is the process that keeps performing the action~$a$.

For unguarded recursive definitions, the picture is, of course,
different. For example, to obtain the denotational semantics of an
unproductive while loop $\whileTermS{\TRUE}{\SKIP}$ characterized by circular operational
behavior
\begin{equation*}
  \whileTermS{\TRUE}{\SKIP} \quad\to\quad \SKIP;\whileTermS{\TRUE}{\SKIP} \quad\to\quad \whileTermS{\TRUE}{\SKIP}
\end{equation*}
one will select one of many solutions of this trivial equation, e.g.\
the least solution in a domain-theoretic semantics. 

Sometimes, however, one has a selection among non-unique solutions of
unguarded recursive equations that is \emph{not} determined
order-theoretically, i.e.\ by picking least or greatest fixpoints. One
example arises from \emph{coinductive
  resumptions}~\cite{GoncharovSchroder13,PirogGibbons15,PirogGibbons14}. In
the paradigm of monad-based encapsulation of
side-effects~\cite{Moggi91}, coinductive resumptions over a base
effect encapsulated by a monad~$T$ form a monad~$T^\nu$, the
\emph{coinductive resumption transform}, given~by
\begin{equation}\label{eq:T-nu}
  T^\nu X= \nu\gamma.\,T(X+\gamma)
\end{equation}
-- that is, a computation over $X$ performs a step with effects
from~$T$, and then returns either a value from $X$ or a resumption
that, when resumed, proceeds similarly, possibly ad infinitum. We thus can view 
coinductive resumptions as processes whose atomic steps are programs over~$T$. 
We
generally restrict to monads $T$ for which~\eqref{eq:T-nu} exists for
all $X$ (although many of our results do not depend on this
assumption). Functors (or monads)~$T$ for which this holds are called
\emph{iteratable}~\cite{AczelAdamekEtAl03}.  Most computationally
relevant monads are iteratable (notable exceptions in the category of
sets are the powerset monad and the continuation monad).
The last occurrence of $\gamma$ in~\eqref{eq:T-nu} may be seen as
being wrapped in an implicit unary \emph{delay} operation that
represents the gap between returning a resumption and resuming~it.
One thus has a natural \emph{delay map} $T^\nu X\to T^\nu X$ that
converts a computation into a resumption, i.e.\ prefixes it with a
delay step. In fact, for $T=\id$, $T^\nu$ is precisely Capretta's
\emph{partiality monad}~\cite{Capretta05}, also called the \emph{delay
  monad}. It is not in general possible to equip $T^\nu X$ with an ordered
domain structure that would allow for selecting least (or greatest)
solutions of unguarded recursive definitions over $T^\nu$. However,
one \emph{can} select solutions in a coherent way, that is, such that
a range of natural quasi-equational axioms is satisfied, making
$T^\nu$ into a (complete) \emph{Elgot
  monad}~\cite{AdamekMiliusEtAl10,GoncharovEA18} whenever~$T$ is so.

More precisely, we closely follow the perspective advanced by Bloom and
Esik~\cite{BloomEsik93,Esik99}, who identify as \emph{iteration operators}
certain categorical operators with the profile $(f\colon X\to Y+X)\mto (f^\istar\colon X\to Y)$
(which are categorical duals of \emph{parametrized recursion operators} 
$(f\colon Y\times X\to X)\mto (f_\istar\colon Y\to X)$~\cite{SimpsonPlotkin00}). The above-mentioned Elgot 
monads support iteration operators in this sense, specifically as operators
on their Kleisli categories (with coproduct~$+$ inherited 
from the base category). We place total (unguarded) iteration and 
partial (guarded) iteration on the same footing and thus aim to unify the 
theories of guarded and unguarded iteration. To this end, we introduce a notion of
\emph{abstractly guarded monads}, that is, monads equipped with a
distinguished class of \emph{abstractly guarded} equation morphisms
satisfying natural closure properties (Section~\ref{sec:guarded}). The
notion of abstract guardedness can be instantiated in various ways,
e.g.\ with the class of immediately terminating `recursive'
definitions, with the class of guarded morphisms in a completely
iterative monad, or with the class of all equation morphisms. We call
an abstractly guarded monad \emph{pre-iterative} if all abstractly
guarded equation morphisms have a solution, and \emph{iterative} if
these solutions are unique. Then completely iterative monads are
iterative abstractly guarded in this sense, and (complete) Elgot
monads are pre-iterative, where we deem every equation morphism to be
abstractly guarded in the latter case.

The quasi-equational axioms of Elgot monads are easily seen to be
satisfied when fixpoints are unique, i.e.\ in iterative abstractly
guarded monads, and moreover stable under iteration-congruent
retractions in a fairly obvious sense. Our first main result
(Section~\ref{sec:congruence}, Theorem~\ref{thm:elgot_is_retract})
states that the converse holds as well, i.e.\ \emph{a monad $T$ is a
  complete Elgot monad iff~$T$ is an iteration-congruent retract of an
  iterative abstractly guarded monad} -- specifically of $T^\nu$ as
in~\eqref{eq:T-nu}. As a slogan,
\begin{quote}
  \emph{monad-based models of unguarded iteration arise by quotienting 
  models of guarded iteration.}
\end{quote}
Our second main result
(Theorem~\ref{thm:nu-algebras}) is an algebraic characterization of
complete Elgot monads: We show that the construction $(-)^\nu$ mapping
a monad $T$ to $T^\nu$ as in~\eqref{eq:T-nu} is a monad on the
category of monads (modulo existence of $T^\nu$), and \emph{complete
  Elgot monads are precisely those $(-)^\nu$-algebras $T$ that cancel
  the delay map on $T^\nu$}, i.e.\ interpret the delay operation as
identity.

As an illustration of these results we discuss various semantic domains of 
processes equipped with canonical solutions of systems of process definitions under
various notions of guardedness (Example~\ref{ex:transfer-example}) and show how
these domains can be related via iteration-preserving morphisms implementing 
a suitable coarsening of the underlying equivalence relation, e.g.\ from bisimilarity to finite trace equivalence (Example~\ref{expl:trace}).
Moreover, we show
(Section~\ref{sec:sandwich}) that sandwiching a complete Elgot monad
between a pair of adjoint functors again yields a complete Elgot
monad, in analogy to a corresponding result for completely iterative
monads~\cite{PirogGibbons15}. Specifically, we prove a sandwich
theorem for iterative abstractly guarded monads and transfer it to
complete Elgot monads using our first main result. For illustration,
we then relate iteration in ultrametric spaces using Escardó's metric
lifting monad~\cite{Escardo99} to iteration in pointed cpo's, by
noting that the corresponding monads on sets obtained using our
sandwich theorems are related by an iteration-congruent retraction in
the sense of our first main result.

The material is organized as follows. We discuss preliminaries on
monads and their Kleisli categories and on coalgebras in
Section~\ref{sec:prelim}. Our notion of abstractly guarded monad,
derived from a notion of guarded co-Cartesian category, is presented
in Section~\ref{sec:guarded}, and extended to parametrized monads in
the sense of Uustalu~\cite{Uustalu03} in
Section~\ref{sec:parametrized}. We prove our main results on the
relationship between Elgot monads and guarded iteration as discussed
above in Section~\ref{sec:congruence}, and present the mentioned
application to sandwiching in Section~\ref{sec:sandwich}. We discuss
related work in Section~\ref{sec:related}; Section~\ref{sec:concl}
concludes. The present paper extends an earlier conference
version~\cite{GoncharovSchroderEtAl17} by full proofs and additional
example material, mostly within Examples~\ref{ex:transfer-example}
and~\ref{expl:trace}.

\section{Preliminaries}\label{sec:prelim}
We work in a category $\BC$ with finite coproducts (including an initial 
object~$\iobj$) throughout. 
A pair
$\sigma=\brks{{\sigma_1\colon Y_1\to X}\comma \sigma_2\colon Y_2\to X}$ of
morphisms is a \emph{summand} of $X$, denoted ${\sigma\colon Y_1\cpto X}$, if
it forms a coproduct cospan, i.e.\ $X$ is a coproduct of $Y_1$
and~$Y_2$ with~$\sigma_1$ and~$\sigma_2$ as coproduct injections. Each
summand $\sigma=\brks{\sigma_1,\sigma_2}$ thus determines a
\emph{complement summand}
$\bar\sigma=\brks{\sigma_2,\sigma_1}\colon Y_2\cpto X$.  We often shorten a
summand $\brks{\sigma_1,\sigma_2}$ to its first component $\sigma_1$, in order to use  
$\sigma$ as a morphism $Y_1\to X$.  Summands of a given object~$X$ are naturally
preordered by taking $\brks{\sigma_1,\sigma_2}$ to be smaller than
$\brks{\theta_1,\theta_2}$ if $\sigma_1$ factors
through~$\theta_1$ and $\theta_2$ factors through $\sigma_2$. This preorder has a greatest
element $\brks{\id_X,\bang}$ and a least element
$\brks{\bang,\id_X}$. By writing $X+Y$ we designate the
latter as a coproduct of $X$ and $Y$ and assign the canonical names
$\inj_1\colon X\cpto X+Y$ and $\inj_2\colon Y\cpto X+Y$ to the corresponding
summands. Dually, we write $\pr_1\colon X\times Y\to X$ and $\pr_2\colon X\times Y\to Y$
for canonical \emph{projections} (without introducing a special arrow
notation).
We do not assume that $\BC$ is
\emph{extensive}~\cite{CarboniLackEtAl93}, in which case coproduct complements
would be uniquely determined.

A \emph{monad} $\BBT$ over $\BC$ can be given in the form of a
\emph{Kleisli triple} $(T,\eta,\argument^\klstar)$ where $T$ is an
endomap over the objects $|\BC|$ of $\BC$, the \emph{unit} $\eta$ is a family
of morphisms $(\eta_X\colon X\to TX)_{X\in|\BC|}$, \emph{Kleisli lifting}
$(\argument)^\klstar$ is a family of maps $\colon\Hom(X,TY)\to\Hom(TX,TY)$,
and the \emph{monad laws} are satisfied:
\begin{align*} 
\eta^{\klstar}=\id, && f^{\klstar}\comp\eta=f, && (f^{\klstar}\comp g)^{\klstar}=f^{\klstar}\comp g^{\klstar}.
\end{align*} 
These laws precisely ensure that taking morphisms of the form $X\to TY$ under 
$f^\klstar g$ as the composition and $\eta$ as identities yields a category, 
which is also called the \emph{Kleisli category} of $\BBT$, and denoted $\BC_{\BBT}$.
The standard (equivalent) categorical definition~\cite{MacLane71} of
$\BBT$ as an endofunctor with natural transformation \emph{unit}
$\eta\colon \Id\to T$ and \emph{multiplication} $\mu\colon TT\to T$ can be
recovered by taking $Tf = (\eta\comp f)^\klstar$,
$\mu=\id^\klstar$. (We adopt the convention that monads and their
functor parts are denoted by the same letter, with the former in
blackboard bold.) We call morphisms ${X\to TY}$ \emph{Kleisli
  morphisms} and view them as a high level abstraction of sequential
programs where~$\BBT$ encapsulates the underlying computational effect
as proposed by Moggi~\cite{Moggi91a}, with~$X$ representing the input
type and~$Y$ the output type. The Kleisli category inherits coproducts
from $\BC$, i.e.\ a coproduct $X+Y$ of objects $X$, $Y$ in~$\BC$
remains a coproduct in $\BC_{\BBT}$, with coproduct injections
$\eta\inl$ and $\eta\inr$.

A more traditional use of monads in semantics is due
to Lawvere~\cite{Lawvere63}, who identified finitary monads on $\Set$
with \emph{algebraic theories}, hence objects $TX$ can be viewed as
sets of terms of the theory over free variables from $X$, the unit as
the operation of casting a variable to a term, and Kleisli composition
as substitution.  We informally refer to this use of monads as
\emph{algebraic monads}. Regardless of this informal convention, for every monad $\BBT$ we 
have an associated category of \emph{(Eilenberg-Moore-)algebras} $\BC^\BBT$
whose objects are pairs $(A,a\colon TA\to A)$ satisfying $a \comp \eta = \id$ and 
$\mu\comp (Ta) = a\comp (Ta)$ and whose morphisms from $(A,a\colon TA\to A)$ to 
$(B, b\colon TB\to B)$ are maps $f\colon  A\to B$ such that $f\comp a = b\comp (Tf)$.

Given an adjunction $F\dashv G\colon \BD\to\BC$, we obtain a monad whose
functor part is the composite $Gf\colon \BC\to\BC$, and both the
Eilenberg-Moore construction and the Kleisli construction show that
every monad has this form. In consequence, we can \emph{sandwich} a
monad~$\BBT$ on $\BD$ between an adjunction $F\dashv G\colon \BD\to\BC$,
obtaining a monad on~$\BC$ with functor part~$GTF$.

A(n\/ \emph{$F$-)coalgebra} for an endofunctor $f\colon \BC\to\BC$ is a pair
$(X,f\colon X\to FX)$ where $X\in |\BC|$. Coalgebras form a category, with
morphisms $(X,f)\to (Y,g)$ being $\BC$-morphisms $h\colon X\to Y$ such that
$(Fh) f = g\, h$. A final object of this category is called a
\emph{final coalgebra}, and we denote it by
\begin{equation*}
(\nu F,\out\colon \nu F\to F\nu F)
\end{equation*}
if it exists. For readability,
\begin{quote}
  \emph{we will often be cavalier about existence of final coalgebras
    and silently assume they exist when we need them;}
\end{quote}
that is, we hide sanity conditions on the involved functors, such as
accessibility (we make an exception to this in parts of
Section~\ref{sec:congruence} where we characterize Elgot monads as
certain Eilenberg-Moore algebras for a monad on the category of
monads). By definition,~$\nu F$ comes with \emph{coiteration} as a
definition principle (dual to the iteration principle for algebras):
given a coalgebra $(X,f\colon X\to FX)$ there is a unique morphism
$(\coit f)\colon X\to\nu F$ such that
\begin{align*}
\out\comp(\coit f) = F(\coit f)\comp f.
\end{align*}
This implies that $\out$ is an isomorphism (\emph{Lambek's lemma}) and
that $\coit\out=\id$ (see~\cite{UustaluVene99} for more details about
coalgebras for coiteration). The category of \emph{$F$-algebras}, $F$-algebra
morphisms and the notion of \emph{initial $F$-algebra} $(\mu F,\inm\colon F\mu F\to\mu F)$ 
are obtained in a completely dual way. The characteristic properties of
final coalgebras and initial algebras can be summarized in the following diagrams:
\begin{equation*}
\begin{tikzcd}[column sep = large,row sep = 4ex]
F\mu F
	\ar[d,"\inm"']\ar[r,"F(\iter f)"] & 
FX\ar[d,"f"] 
\\
\mu F
	\ar[r,"\iter f"']& 
X
\end{tikzcd}
\hspace{10ex}
\begin{tikzcd}[column sep = large,row sep = 4ex]
X
	\ar[d,"f"']
	\ar[r,"\coit f"] & 
\nu F\ar[d,"\out"]\\
FX 
	\ar[r,"F(\coit f)"'] &
F\nu F
\end{tikzcd}
\end{equation*} 
Note that $F$-algebras should not be confused with Eilenberg-Moore algebras of 
monads (as we indicated above, those satisfy additional laws).

We generally drop sub- and superscripts, e.g.\ on natural
transformations, whenever this improves readability.

\section{Abstractly Guarded Categories and Monads}\label{sec:guarded}
The notion of guardedness is paramount in process algebra: typically
one considers systems of mutually recursive process definitions of the
form $x_i = t_i$, and a variable~$x_i$ is said to be guarded in $t_j$
if it occurs in $t_j$ only in subterms of the form $a.\,s$ where
$a.\,(\argument)$ is action prefixing. A standard categorical approach
is to replace the set of terms over variables~$X$ by an object~$TX$
where $\BBT$ is a monad. We then can model separate variables by
partitioning~$X$ into a sum $X_1 +\ldots + X_n$ and thus talk about
guardedness of a morphism $f\colon X\to T(X_1 + \ldots + X_n)$ in any $X_i$,
meaning that every variable from $X_i$ is guarded in $f$. %
One way to capture guardedness categorically is to identify the
operations of~$\BBT$ that serve as guards by distinguishing a suitable
subobject of $TX$; e.g.\ the definition of completely iterative
monad~\cite{Milius05} follows this approach. For our purposes, we
require a yet more general notion where we just distinguish some
Kleisli morphisms as being guarded in certain output variables. We
thus aim to work in a Kleisli category of a monad, but since our
formalization and initial results can already be stated in any
co-Cartesian category, we phrase them at this level of generality as
long as possible.
\begin{definition}[Abstractly guarded
  category/monad] \label{def:g-cat} A co-Cartesian category $\BC$ is
  \emph{abstractly guarded} if it is equipped with a notion of
  \emph{abstract guardedness}, i.e.\ with a relation between morphisms
  $f\colon X\to Y$ and summands $\sigma\colon Y'\cpto Y$ closed under the rules in
  Figure~\ref{fig:guard} where $f\colon X\to_\sigma Y$ denotes the fact that
  $f$ and $\sigma$ are in the relation in question.

  A monad is \emph{abstractly guarded} if its Kleisli category is
  abstractly guarded.  A monad morphism $\alpha\colon\BBT\to\BBS$ between
  abstractly guarded monads $\BBT$, $\BBS$ is \emph{abstractly guarded}
  if $f\colon X\to_\sigma TY$ implies $\alpha f\colon X\to_\sigma SY$.
\end{definition}

\begin{figure}[t!]
\begin{gather*}
\textbf{(trv)}\quad\frac{f\colon X\to Y}{~\inj_1\comp f\colon X\to_{\inj_2} Y+Z~}\qquad\qquad
\textbf{(par)}\quad\frac{~f\colon X\to_\sigma Z~\quad{}~g\colon Y\to_\sigma Z}
{~[f,g]\colon X+Y\to_\sigma Z}\\[2ex]
\textbf{(cmp)}\quad\frac{~f\colon X\to_{\inj_2} Y+Z\qquad g\colon Y\to_{\sigma} V\qquad h\colon Z\to V~}{[g,h]\comp f\colon X\to_{\sigma} V} %
\end{gather*}
\caption{Axioms of abstract guardedness.}
\label{fig:guard}
\end{figure}
\noindent The rules in Figure~\ref{fig:guard} are designed so as to
enable a reformulation of the classical laws of iteration w.r.t.\
abstract guardedness, as we shall see in Section~\ref{sec:congruence}. 
Intuitively, \textbf{(trv)} states that if a program does
not output anything via a summand of the output type then it is
guarded in that summand.  %
Rule~\textbf{(par)} states that putting two guarded equation systems
side by side again produces a guarded system.  Finally,
rule~\textbf{(cmp)} states that guardedness is preserved under
composition: if the unguarded part of the output of a program is
postcomposed with a $\sigma$-guarded program, then the result is
$\sigma$-guarded, no matter how the guarded part is
transformed. That is, guardedness, once introduced, cannot be ``undone''
through sequential composition, but it can be ``forgotten'', as 
the following \emph{weakening rule} indicates:
\begin{displaymath}
  \textbf{(wkn)}\quad\frac{~f\colon X\to_{\sigma} Y~}{~f\colon X\to_{\sigma\theta} Y~},	
\end{displaymath}
where $\sigma$ and $\theta$ are composable summands. This rule was originally 
part of our axiomatization~\cite{GoncharovSchroderEtAl17} but it was
later observed to be derivable from the other three~\cite{GoncharovSchroder18}:
\begin{proposition}
Rule \textbf{(wkn)} is derivable in the calculus of Figure~\ref{fig:guard}.
\end{proposition}
\begin{proof}
  Let $\bar\sigma\colon Z\to Y$ be the complement of $\sigma$, thus
  $Y=Z+Y'$, $\sigma=\inr$ and $\bar\sigma=\inl$. Analogously we
  present $Y'$ as $Z'+Y''$ with $\theta=\inr$.  In summary, $Y$ is a
  coproduct of $Z$, $Z'$ and $Y''$, $f\colon  X \to_{\inr}Z+(Z'+Y'')$, and
  we need to show that $f\colon  X \to_{\inr\inr} Z+(Z'+Y'')$. Since
  $f=[\inl,\inr]\comp f$, by~\textbf{(cmp)} we are left to check that
  $\inl\colon  Z \to_{\inr\inr}Z+(Z'+Y'')$. Now $Z+(Z'+Y'')$ is also a
  coproduct of $Z+Z'$ and~$Y''$, with evident injections; so
  $\inl\colon  Z \to_{\inr\inr}Z+(Z'+Y'')$ is equivalent to
  $\inl\inl\colon  Z \to_{\inr}(Z+Z')+Y''$, which is an instance
  of~\textbf{(trv)}.
\end{proof}
\noindent Rule \textbf{(wkn)} is a weakening principle: If a program
is guarded in some summand then it is guarded in any subsummand of
that summand. Analogously, we obtain stability of guardedness under isomorphisms:
\begin{proposition}
The rule 
{\upshape
\begin{displaymath}
  \textbf{(iso)}\quad\frac{~f\colon X\to_{\sigma} Y\qquad h\colon Y\iso Z}{~h\comp f\colon X\to_{h\sigma} Z~},	
\end{displaymath}
}
is derivable in the calculus of Figure~\ref{fig:guard}.
\end{proposition}
\begin{proof}
Let $\bar\sigma\colon W\to Y$ be the complement of $\sigma$, thus
$Y=W+Y'$, $\sigma=\inr$ and $\bar\sigma=\inl$. Now, $Z$ is a coproduct of $W$ and 
and $Y'$ with $h\inl\colon W\to Z$ and $h\inr\colon Y'\to Z$ as the coproduct injections, and 
$h$ is the copair of $h\inl$ and $h\inr$ w.r.t.\ this coproduct structure. The 
rule in question now follows from~\textbf{(cmp)}, using the fact that by~\textbf{(trv)}, 
$h\inl$ is $h\inr$-guarded. 
\end{proof}

\noindent We write $f\colon X\to_{i_1,\ldots,i_k} X_1+\ldots+X_n$ as a
shorthand for $f\colon X\to_\sigma X_1+\ldots+X_n$ with
$\sigma=[\inj_{i_1},\ldots,\inj_{i_k}]\colon X_{i_1}+\ldots+X_{i_k}\cpto
X_1+\ldots+X_n$.
More generally, we sometimes need to refer to components of some
$X_{i_j}$. This amounts to replacing the corresponding~$i_j$ with a
sequence of pairs $i_j n_{j,m}$, and $\inj_{i_j}$ with
$\inj_{i_j}[\inj_{n_{j,1}},\ldots,\inj_{n_{j,k_j}}]$, so, e.g.\ we
write $f\colon X\to_{12,2} (Y+Z)+Z$ to mean that $f$ is
$[\inj_1\inj_2,\inj_2]$-guarded. Where coproducts $Y+Z$ etc.\ appear
in the rules, we mean any coproduct, not just some selected coproduct.

Recall that we have defined the notion of guardedness as a certain
relation between morphisms and summands. Clearly, the \emph{greatest}
such relation is the one declaring all morphisms to be
$\sigma$-guarded for all~$\sigma$.  We call categories (or monads)
equipped with this notion of guardedness \emph{totally guarded}.  It
turns out we also always have a \emph{least} guardedness relation
(originally called \emph{trivial}~\cite{GoncharovSchroderEtAl17}):
\begin{definition}[Vacuous guardedness]
  A morphism $f\colon X\to Y$ is \emph{vacuously
    $\sigma$-guarded} for ${\sigma\colon Z\cpto Y}$ if $f$ factors through
  the coproduct complement $\bar\sigma$ of~$\sigma$.
\end{definition}
\noindent Intuitively, $f$ is vacuously guarded in ${\sigma\colon Z\cpto Y}$
if $f$ does not output anything via the summand~$Z$; observe that by
the \textbf{(trv)} rule, vacuous guardedness always implies
guardedness. Formally, we have:
\begin{prop}\label{prop:guard_triv}
  By taking the abstractly guarded morphisms to be the vacuously
  guarded morphisms, we obtain the least guardedness relation making
  the given category into a guarded category.
\end{prop}
\begin{proof}
  As indicated above, it is immediate from \textbf{(trv)} that every
  vacuously guarded morphism is guarded under any guardedness relation
  making the category into a guarded category. It remains to show that
  vacuous guardedness is closed under the rules in
  Figure~\ref{fig:guard}; in the following we write
  $f\colon X\to_{\sigma} Y$ to mean that $f$ is vacuously $\sigma$-guarded.
\begin{citemize}
\item\textbf{(trv)}:   Immediate from the definition of vacuous
  guardedness.
\item\textbf{(cmp)}:   Suppose $f\colon X\to_{2} Y+Z$, i.e.\
  $f = \inj_1\comp w$ for some $w\inl\colon X\to Y$. Now, for
  $g\colon Y\to_{\sigma} V$ and $h\colon Z\to V$,
  $[g,h]\comp f= g\comp w$.  Let $\bar\sigma\colon W\cpto V$ be the
  complement of $\sigma\colon V'\cpto V$. By assumption, $g$ factors through
  $\bar\sigma$, i.e.\ $w=\bar\sigma\comp u$ for some~$u$. Therefore
  $[g,h]\comp f=\bar\sigma\comp u\comp w$, which by definition
  means that $[g,h]\comp f$ is vacuously $\sigma$-guarded.
\item \textbf{(par)}:   Suppose that $f\colon X\to_{\sigma} Z$ and $g\colon Y\to_{\sigma} Z$, 
i.e.\ $f= \bar\sigma\comp f'$ and $g=\bar\sigma\comp g'$ for some $f'\colon X\to Z'$ and $g'\colon Y\to Z'$ 
where $\sigma\colon Z'\cpto Z$ and $\bar\sigma$ is the coproduct complement of~$\sigma$. Then, of course,
$[f,g] = \bar\sigma\comp [f',g']$, i.e.\ $[f,g]\colon X+Y\to_{\sigma} Z$. 
\qed
\end{citemize}  
\noqed\end{proof}

\noindent We call a guarded category (or monad) \emph{vacuously
  guarded} if its notion of abstract guardedness is given by vacuous
guardedness. We note briefly how vacuous guardedness instantiates to
Kleisli categories:
\begin{lem}
  Let $\BBT$ be a monad on a category~$\BC$. A morphism
  $f\colon X\to T(Y+Z)$ is vacuously $\inr$-guarded iff $f$ factors through
  $T\inl$ in~$\BC$.
\end{lem}
\begin{proof}
  Immediate from the fact that the left injection into the coproduct
  $Y+Z$ in the Kleisli category of~$\BBT$ is $\eta\inl$, and
  $(\eta\inl)^\klstar=T\inl$.
\end{proof}
\noindent The notion of abstract guardedness can thus vary on a large
spectrum from \emph{vacuous guardedness} to \emph{total guardedness},
possibly detaching it from the initial intuition on guardedness. It is
for this reason that we introduced the qualifier \emph{abstract} into
the terminology; for brevity, we will omit this qualifier in the
sequel in contexts where no confusion is likely, speaking only of
guarded monads, guarded morphisms etc.
\begin{rem}\label{rem:multi_guard}
  One subtle feature of our axiomatization is that it allows for
  seemingly counterintuitive situations when a morphism is
  individually guarded in two disjoint summands, but not in their
  union. This can be illustrated by the following example.  Let $\BBT$
  be the algebraic monad induced by the theory of abelian groups
  presented, in additive notation, by binary~$-$ alone. In this
  presentation, the zero element is presented by terms of the form
  $x-x$; the theory thus differs slightly from the more standard
  presentation in that there is no zero element in the absence of
  variables, i.e.\ $T\iobj=\iobj$. We equip~$\BBT$ with vacuous
  guardedness. Now let $z:1\to T(\{x\}+\{y\})$ be the map that picks
  out the zero element. This morphism is both $\inj_1$-guarded and
  $\inj_2$-guarded, i.e.\ it factors both through $T\inr$ and through
  $T\inl$, since we can write the zero element both as $y-y$ and as
  $x-x$. However,~$z$ fails to be $\id$-guarded, because it does not
  factor through $T\iobj=\iobj$.

Note that, conversely, collective guardedness does always imply
individual guardedness, for by~\textbf{(wkn)}, $f\colon X\to_{1,2} Y+Z$
implies both $f\colon X\to_{1} Y+Z$ and $f\colon X\to_{2} Y+Z$.

\end{rem}
\noindent As usual, guardedness serves to identify systems of
equations that admit solutions according to some global principle:
\begin{definition}[Guarded (pre-)iterative category/monad]
  Given $f\colon X\to_2 Y+X$, we say that $f^\istar\colon X\to Y$ is a
  \emph{solution} of $f$ if $f^\istar$ satisfies the \emph{fixpoint
    identity} $f^\istar = [\id,f^\istar]\comp f$. A guarded category
  is \emph{guarded pre-iterative} if it is equipped with an
  \emph{iteration operator} that assigns to every $\inr$-guarded
  morphism $f\colon X\to_2 Y+X$ a solution $f^\istar$ of $f$.  If every such
  $f$ has a unique solution, we call the category \emph{guarded
    iterative}.

  A guarded monad is \emph{guarded (pre-)iterative} if its Kleisli
  category is guarded \mbox{(pre-)}it\-er\-ative.  A guarded monad morphism
  $\alpha\colon \BBT\to\BBS$ between guarded pre-iterative monads
  $\BBT,\BBS$ is \emph{iteration-preserving} if
  $\alpha f^\istar = (\alpha f)^\istar$ for every $f\colon X\to_2 T(Y+X)$.
\end{definition}
\noindent We can readily check that the iteration operator preserves
guardedness:
\begin{prop}\label{prop:guard_preserve}
  Let $\BC$ be a guarded pre-iterative category, let $\sigma\colon Z\cpto Y$, and let $f\colon X\to_{\sigma+\id} Y+X$. Then
  $f^{\istar}\colon X\to_\sigma Y$.
\end{prop}
\begin{proof}
Let $\bar\sigma\colon Z'\to Y$ be the complement of $\sigma\colon Z\to Y$, so we proceed under 
the assumption that $Y = Z' + Z$, $\sigma=\inr$ and $\bar\sigma=\inl$. Then 
\begin{align*}
f^\istar =&\; [\id,f^\istar] f
=[[\inl,\inr],f^\istar]\comp f
=[\inl,[\inr,f^\istar]]\comp [[\inl,\inr\inl],\inr\inr]\comp f.
\end{align*}
By assumption, $f\colon X\to_{\inr+\id} (Z'+Z)+X$. The morphism 
$h=[[\inl,\inr\inl],\inr\inr]$ is simply an associativity isomorphism, for which
$h\comp (\inr+\id) = \inr$, hence by~\textbf{(iso)},
$[[\inl,\inr\inl],\inr\inr]\comp f\colon X\to_2 Z'+(Z + X)$. Since by~\textbf{(trv)},
$\inl$ is $\inr$-guarded, we are done by~\textbf{(cmp)}.
\end{proof}
\noindent We note that for guarded morphisms into guarded
\emph{iterative} monads, preservation of iteration is automatic:
\begin{lemma}\label{lem:iter-preserve}
  Let $\alpha\colon \BBT\to\BBS$ be a guarded morphism between guarded
  pre-iterative monads~$\BBT$, $\BBS$ with $\BBS$ being guarded
  iterative. Then $\alpha$ is iteration-preserving.
\end{lemma}
\begin{proof}
Indeed, given $f\colon X\to_2 T(Y+X)$, 
\begin{flalign*}
&& \alpha f^\istar 
&  \;= \alpha\comp [\eta, f^\istar]^\klstar f &\by{fixpoint identity for $f^\istar$}\\
&&&\;= [\eta, \alpha f^\istar]^\klstar\alpha f &\by{monad morphism}
\end{flalign*}
but this equation has $(\alpha f)^\istar$ as its unique solution, 
hence $(\alpha f)^\istar = \alpha f^\istar$.
\end{proof}
\noindent In vacuously guarded categories, there is effectively
nothing to iterate, so we have
\begin{prop}\label{prop:triv-iter}
Every vacuously guarded category is guarded iterative.
\end{prop} 
\begin{proof}
  Let $f\colon X\to_2 Y+X$, which by assumption means that $f=\inl g$ for
  some $g$. Then for any $f^\istar$ satisfying
  $f^\istar = [\id,f^\istar]\comp f$, we have
  $f^\istar = [\id,f^\istar]\comp f = [\id,f^\istar]\comp \inl g = g$,
  which proves uniquenes of solutions. Moreover,
  $[\id,g]f=[\id,g]\inl g=g$, which shows existence.
\end{proof}

\noindent We now revisit our motivating considerations on process
algebra from the beginning of this section.
\begin{example}[Generalized processes]\label{ex:gen-proc}
  A natural semantic domain for finitely branching possibly infinite
  processes under strong bisimilarity with final results in $X$ and
  atomic actions in $A$ is the final coalgebra
  $\nu\gamma.\,\FSet(X+A\times\gamma)$ in the category of sets and
  functions, where~$\FSet$ is the finite powerset
  monad. Alternatively, we can view inhabitants of this domain as
  equivalence classes of possibly non-well-founded terms over
  variables from $X$, which can also be thought of as process names,
  and over the operations $+$ of non-deterministic choice, deadlock
  $\iobj$ and action prefixing $a.\,(\argument)$. The latter view is
  useful for syntactic presentations of those processes that happen to
  be finite. Systems of recursive process definitions are naturally
  represented by morphisms
  $f\colon X\to\nu\gamma.\,\FSet((Y+X)+A\times\gamma)$ where $X$ contains
  process names being defined and $Y$ contains the remaining process
  names that can occur freely. For example, the system
\begin{align}\label{eq:spec-ex}
x = y + a.\, x
\end{align}
corresponds to the following data: 
$X = \{x\}$, $Y =\{y\}$, $A = \{a\}$, and
\begin{align*}
f(x) =&\, \tuo\{\inl y,\inr\brks{a,\tuo\{\inl x\}}\}
\end{align*}
(eliding the isomorphism $Y+X\cong\{x,y\}$). The generalization
arising from this example is as follows: Given an endofunctor $\Sigma$
on a co-Cartesian category $\BC$ and a monad $\BBT$ such that final
coalgebras $T_{\Sigma} X = \nu\gamma.\,T(X+\Sigma\gamma)$ exist, we
obtain a corresponding monad $\BBT_{\Sigma}$ called the
\emph{generalized coalgebraic resumption monad transform} of
$\BBT$. As above, we can view morphisms $f\colon X\to T_{\Sigma} (Y+X)$ as
systems of recursive equations for \emph{generalized processes} with
$\BBT$ capturing the relevant computational effect (such as
non-determinism) and $\Sigma$ capturing \emph{atomic steps} (such as
actions $\Sigma = A\times\argument$).
\end{example}
\noindent Abstract guardedness can be used to effectively distinguish
those systems $f\colon X\to T_{\Sigma} (Y+X)$ for which we can define
\emph{desirable} solutions $f^\istar\colon X\to T_{\Sigma} Y$. For the
moment, we proceed under the assumption that \emph{desirable} means
\emph{unique}, for instance~\eqref{eq:spec-ex} has the unique solution
$x = y + a.\, (y + a.\,(\ldots))$. Let us recall the existing approach
to defining guardedness in this context via completely iterative
monads~\cite{Milius05}, which are based on idealized
monads~\cite[Definition 5.5]{Milius05}. To make this precise, recall
some definitions.
\begin{definition}[Monad modules, idealized monads]\label{def:module}
A \emph{module} over a monad~$\BBT$ on~$\BC$ is a pair $(\module T,
\argument^\modstar)$, where~$\module T$ is an endomap over the objects of~$\BC$,
while the lifting~$(\argument)^\modstar$ is a map $\Hom(X,TY)\to\Hom(\module T
X,\module T Y)$ such that the following laws are satisfied:
\begin{align*}
\eta^\modstar = \id,
&&
g^\modstar f^\modstar = (g^\klstar f)^\modstar.
\end{align*}
Note that $\module T$ extends to an endofunctor by taking
$\module T f = (\eta f)^\modstar$.  A \emph{module-to-monad morphism}
is a natural transformation $\modm\colon \module T \to T$ that satisfies
$\modm f^\modstar = f^\klstar \modm$. We call the tuple
$(\BBT, \module T, \argument^\modstar, \modm)$ an \emph{idealized
  monad}; when no confusion is likely, we refer to these data just
as~$\BBT$.  %
An \emph{idealized monad morphism} between idealized monads
$((T,\eta^T,\argument^\klstar), \module T, \argument^\modstar, \modm)$
and
$((S,\eta^S,\argument^\kklstar), \moduleS S, \argument^\modstarS,
\modm')$
is a pair $(\alpha, \beta)$ where $\alpha\colon  T \to S$ is a monad
morphism while $\beta\colon  \module T \to \moduleS S$ is a natural
transformation satisfying $\alpha\modm = \modm'\beta$ and
$\beta f^\klstar = f^\kklstar \beta$.
\end{definition}
\begin{example}
  It follows from previous
  results~\cite[Corollary~3.13]{PirogGibbons14} that the monad
  $\BBT_\Sigma$ from Example~\ref{ex:gen-proc} is idealized when
  equipped with the module $T\Sigma T_{\Sigma}$. In the concrete case
  where $\BBT=\FSet$ and $\Sigma = A\times(\argument)$, i.e.\
  $T_\Sigma X=\nu\gamma.\,\FSet((Y+X)+A\times\gamma)$, the module
  $\FSet(A\times T_\Sigma)$ contains processes that consist of
  (finitely many) non-deterministic branches all of which begin with
  an action.
\end{example}
\noindent Milius~\cite{Milius05} defines guardedness only for
equation morphisms, i.e.\ morphisms of type $X\to T(Y+X)$. Extending
this notion in the obvious way to morphisms of type $X\to T(Y+Z)$ as
required in our framework, we obtain the following definition:
\begin{definition}[Completely iterative monads]\label{def:comp-iter}
  Given an idealized monad
  $(\BBT, \module T, \argument^\modstar, \modm)$, a morphism
  $f\colon  X \to T(Y + Z)$ is \emph{guarded} if it factors via
  $[\eta \inl, \modm]\colon  Y + \module T (Y+Z) \to T(Y+Z)$. The monad
  $\BBT$ is \emph{completely iterative} if every guarded
  $f\colon X\to T(Y+X)$ in this sense has a unique solution.
\end{definition}
\noindent It turns out that the above notion of guardedness is not an
instance of abstract guardedness; specifically, it does not satisfy
our~\textbf{(par)} rule. Equation~\eqref{eq:spec-ex} provides a good
illustration of what happens: although both terms $y$ and $a.\,x$ are
guarded in $x$, we cannot factor the corresponding term
$X \to T(Y + X)$ through any
$[\eta \inl, \modm]\colon  Y + \module T (Y+X) \to T(Y+X)$ due to the
top-level nondeterministic choice.

 Fortunately, we can fix this by
noticing that completely iterative monads actually support iteration
for a wider class of morphisms:
\begin{definition}
  Let $(\BBT, \module T, \argument^\modstar, \modm)$ be an idealized
  monad. Given $\sigma\colon Z\cpto Y$, we say that a morphism
  $f\colon  X \to T Y$ is \emph{weakly $\sigma$-guarded} if it factors
  through $[\eta \bar\sigma, \modm]^\klstar\colon  T(Y' + \module TY) \to TY$
  for a complement $\bar\sigma\colon Y'\cpto Y$ of $\sigma$.
\end{definition}
\noindent Since a morphism that factors as $[\eta \inl, \modm] f$ can be
rewritten as $[\eta \inl, \modm]^\klstar \eta f$, every guarded
morphism in an idealized monad is also weakly guarded.%
\begin{theorem}\label{thm:weak-guardedness}
  Let $(\BBT, \module T, \argument^\modstar, \modm)$ be an idealized
  monad. Then the following hold.
  \begin{enumerate}
  \item\label{item:weak-abstract} $\BBT$ becomes abstractly guarded
    when equipped with weak guardedness as the notion of abstract
    guardedness.
  \item\label{item:weak-solutions} If\/ $\BBT$ is completely iterative,
    then every weakly $\inr$-guarded morphism $f\colon X\to T(Y+X)$ has a
    unique solution. 
  \item\label{item:weak-monad-morphisms} If $(\alpha,\beta)$ is an idealized monad morphism, then $\alpha$ preserves weak guardedness. 
  \end{enumerate}
\end{theorem}
\noindent
That is, completely iterative monads are abstractly guarded iterative
monads w.r.t.\ weak guardedness.
\begin{proof}
\emph{\eqref{item:weak-abstract}: } We need to verify that weak
guardedness is closed under the rules from Definition~\ref{def:g-cat}.
\begin{citemize}
\item \textbf{(trv)}
Given a morphism $f\colon  X \to TY$, the following holds:
\begin{flalign*}
&& (T\inl)f
=&\; (\eta\inl)^\klstar f & \by{Kleisli}\\
&&=&\; ([\eta\inl,\modm]\inl)^\klstar f & \by{coproducts}\\
&&=&\; [\eta\inl,\modm]^\klstar (T\inl) f & \by{Kleisli}
\end{flalign*}
\item \textbf{(cmp)}
Given $f\colon  X \to_2 T(Y+Z)$, $g\colon  Y \to_\sigma TV $, and $h\colon  Y \to TV$, assume that $f$ factors as $[\eta\inl, \modm]^\klstar f'$, while $g$ factors as $[\eta\bar\sigma, \modm]^\klstar g'$. Then, the following holds:
\begin{flalign*}
&& [g, h]^\klstar f
=&\; [[\eta\bar\sigma, \modm]^\klstar g', h]^\klstar [\eta\inl, \modm]^\klstar f' & \\
&&=&\; ([[\eta\bar\sigma, \modm]^\klstar g', h]^\klstar [\eta\inl, \modm])^\klstar f' & \by{Kleisli}\\
&&=&\; [[[\eta\bar\sigma, \modm]^\klstar g', h]^\klstar \eta\inl, [[\eta\bar\sigma, \modm]^\klstar g', h]^\klstar \modm]^\klstar f' & \by{coproducts}\\
&&=&\; [[[\eta\bar\sigma, \modm]^\klstar g', h]\inl, [[\eta\bar\sigma, \modm]^\klstar g', h]^\klstar \modm]^\klstar f' & \by{Kleisli}\\
&&=&\; [[\eta\bar\sigma, \modm]^\klstar g', [[\eta\bar\sigma, \modm]^\klstar g', h]^\klstar \modm]^\klstar f' & \by{coproducts}\\
&&=&\; [[\eta\bar\sigma, \modm]^\klstar g', \modm [[\eta\bar\sigma, \modm]^\klstar g', h]^\modstar]^\klstar f' & \by{module-to-monad morphism}\\
&&=&\; [[\eta\bar\sigma, \modm]^\klstar g', [\eta\bar\sigma, \modm]\inr [[\eta\bar\sigma, \modm]^\klstar g', h]^\modstar]^\klstar f' & \by{coproducts}\\
&&=&\; [[\eta\bar\sigma, \modm]^\klstar g', [\eta\bar\sigma, \modm]^\klstar \eta \inr [[\eta\bar\sigma, \modm]^\klstar g', h]^\modstar]^\klstar f' & \by{Kleisli}\\
&&=&\; ([\eta\bar\sigma, \modm]^\klstar [g', \eta \inr [[\eta\bar\sigma, \modm]^\klstar g', h]^\modstar])^\klstar f' & \by{coproducts}\\
&&=&\; [\eta\bar\sigma, \modm]^\klstar [g', \eta \inr [[\eta\bar\sigma, \modm]^\klstar g', h]^\modstar]^\klstar f'. & \by{Kleisli}\\
\end{flalign*}
\item \textbf{(par)}
Given a morphism $f\colon  X \to_{\sigma} TZ$ and $Y \to_{\sigma} TZ$ assume that $f$ factors as $[\eta\inl, \modm]^\klstar f'$, and $g$ factors as $[\eta\inl, \modm]^\klstar g'$. Then, the following holds:
\begin{flalign*}
&& [f,g]
=&\; [[\eta\bar\sigma,\modm]^\klstar f', [\eta\bar\sigma,\modm]^\klstar g'] & \by{guardedness}\\
&&=&\; [\eta\bar\sigma,\modm]^\klstar [f', g']. & \by{coproducts}
\end{flalign*}
\end{citemize}

\emph{\eqref{item:weak-solutions}: } Let
$f = [\eta \inl, \modm]^\klstar j$ for a morphism
$j \colon X \to T(Y + \module T(Y + X))$. We define an auxiliary morphism
$g = [\eta\inl, j]^\klstar \modm \colon \module T(Y+X) \to T(Y + \module
T(Y+X))$.
Note that $g$ is guarded (in the sense of~\cite{Milius05}), since it
can be rewritten as follows:
\begin{flalign*}
&&\mediator{j}{ \eta\inl}^\klstar \modm
=&\; \modm \mediator{j}{ \eta\inl}^\modstar&\by{module-to-monad morphism}\\
&&=&\; \mediator{\modm}{ \eta\inl} \inr \mediator{j}{ \eta\inl}^\modstar.&\by{coproducts}
\end{flalign*}
Thus, $g$ has a unique solution $g^\istar\colon  \module T(Y+X) \to TY$. We
use it to define a solution to $f$, namely
$f^\iistar = [\eta, g^\istar]^\klstar j$. It is left to show that it
is indeed a solution and that it is unique:

\begin{itemize}
\item \textit{Solution:}
\begin{flalign*}
&&f^\iistar
=&\; \mediator{g^\istar}{\eta}^\klstar j\\
&&=&\; \mediator{\mediator{g^\istar}{\eta}^\klstar g }{\eta}^\klstar j &\by{solution}\\
&&=&\; \mediator{\mediator{g^\istar}{\eta}^\klstar \mediator{j}{ \eta\inl}^\klstar \modm }{\eta}^\klstar j\\
&&=&\; \mediator{(\mediator{g^\istar}{\eta}^\klstar \mediator{j}{ \eta\inl})^\klstar \modm }{\eta}^\klstar j &\by{Kleisli}\\
&&=&\; \mediator{\mediator{\mediator{g^\istar}{\eta}^\klstar j}{ \mediator{g^\istar}{\eta}^\klstar\eta\inl}^\klstar \modm }{\eta}^\klstar j &\by{coproducts}\\
&&=&\; \mediator{\mediator{\mediator{g^\istar}{\eta}^\klstar j}{ \mediator{g^\istar}{\eta}\inl}^\klstar \modm }{\eta}^\klstar j &\by{Kleisli}\\
&&=&\; \mediator{\mediator{\mediator{g^\istar}{\eta}^\klstar j}{ \eta}^\klstar \modm }{\eta}^\klstar j &\by{coproducts}\\
&&=&\; \mediator{\mediator{\mediator{g^\istar}{\eta}^\klstar j}{ \eta}^\klstar \modm}{ \mediator{\mediator{g^\istar}{\eta}^\klstar j}{ \eta} \inl}^\klstar j &\by{coproducts}\\
&&=&\; \mediator{\mediator{\mediator{g^\istar}{\eta}^\klstar j}{ \eta}^\klstar \modm}{ \mediator{\mediator{g^\istar}{\eta}^\klstar j}{ \eta}^\klstar \eta \inl}^\klstar j &\by{Kleisli}\\
&&=&\; (\mediator{\mediator{g^\istar}{\eta}^\klstar j}{ \eta}^\klstar \mediator{\modm}{ \eta \inl})^\klstar j &\by{coproducts}\\
&&=&\; \mediator{\mediator{g^\istar}{\eta}^\klstar j}{ \eta}^\klstar \mediator{\modm}{ \eta \inl}^\klstar j  &\by{Kleisli}\\
&&=&\; \mediator{f^\iistar}{ \eta}^\klstar f.
\end{flalign*}
\item \textit{Uniqueness:} Let $r \colon X \to TY$ be a solution of $f$,
  that is, $r = [\eta,r]^\klstar f$. First, we calculate:
\begin{flalign*}
&& \mediator{r}{\eta}^\klstar \modm
=&\; \mediator{\mediator{r}{\eta}^\klstar f}{\eta}^\klstar \modm & \\
&&=&\; \mediator{\mediator{r}{\eta}^\klstar \mediator{\modm}{ \eta \inl}^\klstar j}{\eta}^\klstar \modm \\
&&=&\; \mediator{(\mediator{r}{\eta}^\klstar \mediator{\modm}{ \eta \inl})^\klstar j}{\eta}^\klstar \modm &\by{Kleisli}\\
&&=&\; \mediator{\mediator{\mediator{r}{\eta}^\klstar \modm}{ \mediator{r}{\eta}^\klstar \eta \inl}^\klstar j}{\eta}^\klstar \modm &\by{coproducts}\\
&&=&\; \mediator{\mediator{\mediator{r}{\eta}^\klstar \modm}{ \mediator{r}{\eta} \inl}^\klstar j}{\eta}^\klstar \modm &\by{Kleisli}\\
&&=&\; \mediator{\mediator{\mediator{r}{\eta}^\klstar \modm}{ \eta}^\klstar j}{ \eta}^\klstar \modm &\by{coproducts}\\
&&=&\; \mediator{\mediator{\mediator{r}{\eta}^\klstar \modm}{ \eta}^\klstar j}{ \mediator{\mediator{r}{\eta}^\klstar \modm}{ \eta} \inl}^\klstar \modm &\by{coproducts}\\
&&=&\; \mediator{\mediator{\mediator{r}{\eta}^\klstar \modm}{ \eta}^\klstar j}{ \mediator{\mediator{r}{\eta}^\klstar \modm}{ \eta}^\klstar \eta\inl}^\klstar \modm &\by{Kleisli}\\
&&=&\; (\mediator{\mediator{r}{\eta}^\klstar \modm}{ \eta}^\klstar \mediator{j}{ \eta\inl})^\klstar \modm &\by{coproducts}\\
&&=&\; \mediator{\mediator{r}{\eta}^\klstar \modm}{ \eta}^\klstar \mediator{j}{ \eta\inl}^\klstar \modm &\by{Kleisli}\\
&&=&\; \mediator{\mediator{r}{\eta}^\klstar \modm}{ \eta}^\klstar g.
\end{flalign*}
Thus, $[\eta,r]^\klstar \modm$ is a solution of $g$. By uniqueness, we obtain that $g^\istar = [\eta,r]^\klstar \modm$. With this, we can check the uniqueness of $f^\iistar$:
\begin{flalign*}
&& r
=&\; \mediator{r}{\eta}^\klstar f & \\
&&=&\; \mediator{r}{\eta}^\klstar \mediator{\modm}{ \eta \inl}^\klstar j \\
&&=&\; (\mediator{r}{\eta}^\klstar \mediator{\modm}{ \eta \inl})^\klstar j &\by{Kleisli}\\
&&=&\; \mediator{\mediator{r}{\eta}^\klstar \modm}{ \mediator{r}{\eta}^\klstar \eta \inl}^\klstar j &\by{coproducts}\\
&&=&\; \mediator{\mediator{r}{\eta}^\klstar \modm}{ \mediator{r}{\eta}\inl}^\klstar j &\by{Kleisli}\\
&&=&\; \mediator{\mediator{r}{\eta}^\klstar \modm}{ \eta}^\klstar j &\by{coproducts}\\
&&=&\; \mediator{g^\istar}{\eta}^\klstar j &\by{the above}\\
&&=&\; f^\iistar.
\end{flalign*} 
\end{itemize}

\emph{\eqref{item:weak-monad-morphisms}: }
Let $(\alpha,\beta)$ be as in Definition~\ref{def:module}. Let $f\colon  X \to TY$ be weakly $\sigma$-guarded. This means that $f$ factors as $[\eta^T \bar\sigma, \modm]^\klstar f'$ for a morphism $f'\colon  X \to T(Y' + \module T Y)$. We need to show that $\alpha f\colon  X \to SY$ factors as $[\eta^S \bar\sigma, \modm']^\kklstar g$ for some $g\colon  X \to S(Y' + \moduleS S Y)$. We calculate:
\begin{flalign*}
&& \alpha f
=&\; \alpha [\eta^T \bar\sigma,\modm]^\klstar f' &\by{factorisation of $f$} \\
&&=&\; (\alpha [\eta^T \bar\sigma,\modm])^\kklstar \alpha f' &\by{monad morphism} \\
&&=&\; [\alpha \eta^T \bar\sigma,\alpha \modm]^\kklstar \alpha f' &\by{coproducts} \\
&&=&\; [\eta^S \bar\sigma,\alpha \modm]^\kklstar \alpha f' &\by{monad morphism} \\
&&=&\; [\eta^S \bar\sigma,\modm' \beta]^\kklstar \alpha f' &\by{idealized monad morphism} \\
&&=&\; ([\eta^S \bar\sigma,\modm'](\id+\beta))^\kklstar \alpha f' &\by{coproducts} \\
&&=&\; ([\eta^S \bar\sigma,\modm']^\kklstar \eta^S (\id+\beta))^\kklstar \alpha f' &\by{Kleisli} \\
&&=&\; [\eta^S \bar\sigma,\modm']^\kklstar (\eta^S (\id+\beta))^\kklstar \alpha f' &\by{Kleisli}
\end{flalign*}
\end{proof}

\section{Parametrizing Guardedness}\label{sec:parametrized}
Uustalu~\cite{Uustalu03} defines a \emph{parametrized monad} to be a
functor from a category $\BC$ to the category of monads over $\BC$. We
need a minor adaptation of this notion where we allow parameters from
a different category than~$\BC$, and simultaneously introduce a
guarded version of parametrized monads:
\begin{definition}[Parametrized guarded monad]
  A \emph{parametrized (guarded) monad} is a functor from a category
  $\BD$ to the category of (guarded) monads and (guarded) monad
  morphisms over $\BC$. Alternatively (by uncurrying), it is a
  bifunctor $\IB{}{}\colon \BC\times\BD\to\BC$ such that for any
  $X\in|\BD|$, $\IB{\argument}{X}\colon \BC\to\BC$ is a (guarded) monad, and
  for every $f\colon Z\to V$, $\IB{\id}{f}\colon \IB{X}{Z}\to\IB{X}{V}$ is the
  $X$-component of a (guarded) monad morphism
  $\IB{\argument}{f}\colon \IB{\argument}{Z}\to\IB{\argument}{V}$,
explicitly,
\begin{align}\label{eq:par-moph}
(\IB{\id}{f})\comp\eta = \eta &&  (\IB{\id}{f})\comp g^\klstar =  ((\IB{\id}{f})\comp g)^\klstar (\IB{\id}{f})\comp 
\end{align}
for any $g\colon X\to Y$ and, in the guarded case,
\begin{equation*}
g\colon Z\to_{\sigma}\IB{V}{X}\quad\text{implies}\quad
(\IB{\id}{f})\comp g\colon Z\to_\sigma\IB{V}{Y}.
\end{equation*}
A \emph{parametrized (guarded) monad morphism} between parametrized
(guarded) monads qua functors into the category of (guarded) monads
over $\BC$ is a natural transformation that is componentwise a
(guarded) monad morphism. In uncurried notation, given parametrized
monads $\IB{}{},\IB[\hat\hash]{}{}\colon \BC\times\BD\to\BC$ a natural
transformation $\alpha\colon \IB{}{}\to\IB[\hat\hash]{}{}$ is a parametrized
(guarded) monad morphism if for each $X\in|\BD|$,
$\alpha_{\argument,X}\colon \IB{\argument}{X}\to\IB[\hat\hash]{\argument}{X}$
is a (guarded) monad morphism.

A parametrized guarded monad $\IB{}{}$ is \emph{guarded
  (pre-)iterative} if each monad $\IB{\argument}{X}$ is guarded
(pre-)iterative and the monad morphisms $\IB{\argument}{f}$ are
iteration-preserving, i.e.\
\begin{align}\label{eq:par-iter-preserve}
(\IB{\id}{f})\comp g^\istar = ((\IB{\id}{f})\comp g)^\istar. 
\end{align}
\end{definition}
\noindent Note that by Lemma~\ref{lem:iter-preserve}, condition~\eqref{eq:par-iter-preserve}
is automatic for guarded iterative parametrized monads.

In the sequel, we tend to use the same notation for parametrized
monads as for the non-parametrized case, assuming that omitted
information is understood from the context. For example, the monad unit
$\eta_{X,Y}\colon X\to\IB{X}{Y}$ is additionally parametrized by $Y$, and
both parameters will be occasionally omitted unless confusion
arises. Kleisli lifting assigns $f^\klstar\colon \IB{X}{Z}\to\IB{Y}{Z}$ to
$f\colon X\to\IB{Y}{Z}$, and for fixed $Z$ all monad laws can be used
for parametrized monads as stated for non-parametrized monads. The
connection between Kleisli lifting and the functor part of the monad
can now be restated as follows:
$(\IB{f}{\id_Z}) = (\eta_{Y,Z} f)^\klstar$ where $f\colon X\to\IB{Y}{Z}$.
\begin{example}\label{ex:param-monad}
  For purposes of the present work, the most important example
  (taken from~\cite{Uustalu03}) is
  $\IB{}{}=T(\argument+\Sigma\argument)\colon\BC\times\BC\to\BC$ where
  $\BBT$ is a (non-parametrized) monad on~$\BC$ and~$\Sigma$ is an
  endofunctor on $\BC$. Informally, $\BBT$ captures a computational
  effect, e.g.\ nondeterminism for $T$ being a (bounded) powerset monad, and $\Sigma$
  captures a signature of actions, e.g.\ $\Sigma X=A\times X$, as in
  Example~\ref{ex:gen-proc}. Specifically, taking $A=1$ we obtain
  $\IB{X}{Y}=T(X+Y)$; in this case, we have only one guard, which can
  be interpreted as a delay. The second argument of $\IB{}{}$ can thus
  be thought of as designated for guarded recursion.
\end{example}

\begin{theorem}\label{thm:ext_guard}
  Let $\IB{}{}\colon \BC\times(\BC\times\BD)\to\BC$ be a parametrized monad,
  with unit $\eta$ and Kleisli lifting $(-)^\klstar$. Then 
  \begin{align*}
   \IB[\hash^\nu]{X}{Y} = \nu\gamma.\, \IB{X}{(\gamma,Y)}
  \end{align*} 
  defines a
  parametrized monad $\IB[\hash^\nu]{}{}\colon \BC\times\BD\to\BC$, whose unit and Kleisli lifting we denote 
$\eta^\nu$ and $(\argument)^\kklstar$, respectively.
 Moreover,
\begin{enumerate}
  \item\label{item:par-guard} If $\IB{}{}$ is guarded, then so is $\IB[\hash^\nu]{}{}$, with guardedness defined as follows: given
  $\sigma\colon Y'\cpto Y$, $f\colon X\to \IB[\hash^\nu]{Y}{Z}$ is
  $\sigma$-guarded if $\out f\colon X\to\IB{Y}{(\IB[\hash^\nu]{Y}{Z},Z)}$
  is $\sigma$-guarded; the correspondence $\IB{}{}\mto\IB[\hash^\nu]{}{}$ extends
to a functor between the respective categories of guarded parametrized monads. 
  \item\label{item:par-pre-iter} If $\IB{}{}$ is guarded pre-iterative, with an iteration operator $(\argument)^\istar$,
then so is $\IB[\hash^\nu]{}{}$, with the iteration
operator $(\argument)^\iistar$ sending  $f\colon X\to_2 \IB[\hash^\nu]{(Y+X)}{Z}$ to $f^\iistar\colon X\to \IB[\hash^\nu]{Y}{Z}$ as follows:
\begin{align*}
f^\iistar = \coit\Bigl([\eta_Y,(\out f)^\istar]^\klstar\out\colon   \IB[\hash^\nu]{(Y+X)}{Z} \to\IB{Y}{((\IB[\hash^\nu]{(Y+X)}{Z}),Z)} \Bigr)\comp\eta^\nu_{Y+X,Z}\inr \kern-6ex
\end{align*}
\item\label{item:par-iter} If $\IB{}{}$ is guarded iterative, then so
  is $\IB[\hash^\nu]{}{}$, with solutions described as in the previous
  clause.
\end{enumerate}
\end{theorem}
\noindent To better understand the typing in the second clause above, note
that
\begin{itemize}
\item $\out f\colon X\to_2\IB{(Y+X)}{(\IB[\hash^\nu]{(Y+X)}{Z},Z)}$, so
\item $(\out f)^\istar\colon X\to\IB{Y}{(\IB[\hash^\nu]{(Y+X)}{Z},Z)}$;
\item the right-most occurrence of $\out$ has type
  \begin{equation*}
    \out\colon \IB[\hash^\nu]{(Y+X)}{Z} \to\IB{(Y+X)}{(\IB[\hash^\nu]{(Y+X)}{Z},Z)};
  \end{equation*}
\item the $\coit(\dots)$ subterm has type $\IB[\hash^\nu]{(Y+X)}{Z} \to\IB[\hash^\nu]{Y}{Z}$.
\end{itemize}
In case there is only one parameter of type $\BC$, i.e.\
$\IB{}{}\colon \BC\times\BC\to\BC$, the typing simplifies slightly: Now
$\hash^\nu$ is just a monad on~$\BC$, which we denote by $\IBnu$
(i.e.\ $\IBnu X= X\hash^\nu()$). We write $\eta^\nu$,
$(\argument)^\kklstar$ for the corresponding monad structure. Then
given $f\colon X\to_2\IBnu(Y+X)$,
\begin{itemize}
\item $\out f\colon X\to_2\IB{(Y+X)}{\IBnu(Y+X)}$;
\item $(\out f)^\istar\colon X\to_2\IB{Y}{\IBnu(Y+X)}$;
\item the right-most occurrence of $\out$ has type
  \begin{equation*}
    \out\colon \IBnu(Y+X) \to\IB{(Y+X)}{\IBnu(Y+X)};
  \end{equation*}
\item the $\coit(\dots)$ subterm has type $\IBnu(Y+X) \to\IBnu Y$;
\item and, of course, $f^\iistar\colon X\to\IBnu Y$.
\end{itemize}
\begin{proof}[Proof (Theorem~\ref{thm:ext_guard})]
\eqref{item:par-guard}:  By currying we equivalently view $\IB{}{}$ as a functor from
  $\BD$ to the category of parametrized guarded monads of type
  $\BC\times\BC\to\BC$, and the transformation $\IB{}{}\mto \IB[\hash^\nu]{}{}$
  as given pointwise.
 It therefore suffices to show that the
  assignment
  \begin{equation*}
    \IB{}{}\mapsto \IBnu\qquad\text{where}\qquad \IBnu{X}=\nu\gamma.\,\IB{X}{\gamma}
  \end{equation*}
  extends to a functor 
  from parametrized guarded monads of type $\BC\times\BC\to\BC$ to
  guarded monads over $\BC$ where guardedness for $\IBnu{}$ is defined as follows:
$f\colon X\to \IBnu Y$ is $\sigma$-guarded iff
  $\out f\colon X\to \IB{Y}{\IBnu Y}$ is $\sigma$-guarded w.r.t.\ $\IB{}{}$.
 Uustalu~\cite{Uustalu03} already proves
  that $\IBnu$ is a monad; we proceed to check that his
  construction is in fact functorial.

  As indicated above, we denote the monad structure on $\IBnu$ by
  $\eta^\nu$, $(\argument)^\kklstar$. These data are uniquely
  determined by commutation of
  \begin{equation}\label{eq:IBnu-monad}
    \begin{tikzcd}
      X\ar{r}{\eta^\nu_X}\ar{d}[swap]{\eta_{X,\IBnu{X}}}&[2ex] \IBnu{X}\ar{d}{\out}\\
      X\hash \IBnu{X}\ar[r,equal] &X\hash \IBnu{X}
    \end{tikzcd}\hspace{30pt}%
    \begin{tikzcd}[column sep = large]
      \IBnu{X} +\IBnu{Y} \ar{r}{[f^\kklstar,\,\id]}\ar{d}[swap]{[\hat f,\,(Y\hash\inr)\out]}
      &[2ex] \IBnu{Y} \ar{d}{\out}\\
      Y\hash(\IBnu{X}+\IBnu{Y})\ar{r}{Y\hash[f^\kklstar,\,\id]}
        & Y\hash(\IBnu{Y})
    \end{tikzcd}
  \end{equation}
  where 
  \begin{align*}
    \hat f =&\;\Bigl(\IBnu{X} \xrightarrow{\out} X \hash \IBnu{X} \xrightarrow{X\hash\inl}
    X\hash(\IBnu{X}{+}\IBnu{Y})\xrightarrow{\;\;\bar f^\klstar}
     Y\hash(\IBnu{X}+\IBnu{Y})\Bigr)\\
    \bar f =&\;\Bigl(X\xrightarrow{\;f} \IBnu{Y}
    \xrightarrow{\out}X\hash \IBnu{Y} \xrightarrow{Y\hash\inr} 
    Y\hash(\IBnu{X}+\IBnu{Y})\Bigr).
  \end{align*}
  That is, $\eta^\nu_X$ is the unique $(\IB{X}{\argument})$-coalgebra
  morphism $(X,\eta_{X,\IBnu{X}})\to(\IBnu{X},\out)$, and $[f^\kklstar,\id]$ is
  the unique $(\IB{Y}{\argument})$-coalgebra morphism
  \begin{equation*}
    (\IBnu{X}+\IBnu{Y},\;[\hat f,(Y\hash\inr)\out])\to(\IBnu{Y},\out),
  \end{equation*}
  the latter being essentially a definition of $f^\kklstar$ by primitive
  corecursion.  In the sequel, we will omit the object part of
  coalgebras when convenient, saying, e.g., that $\eta^\nu_X$ is a
  coalgebra morphism $\eta_{X,\IBnu{X}}\to\out$.

  We need to define the action of $\IBnum{}$ on morphisms: Let $\IB[\hash']{}{}$ be a
  further parametrized monad, with all data of $\hash'$ and
  $\IBnup$ indicated by primes, and let
  \begin{equation*}
    \alpha\colon \IB{}{}\to\IB[\hash']{}{}
  \end{equation*}
  be a parametrized monad morphism. We then define a monad morphism 
  $\IBnum{\alpha}\colon \IBnu\to\IBnup$ by commutation of
  \begin{equation*}
    \begin{tikzcd}[column sep = 12ex]
      \IBnu{X} \ar{r}{(\IBnum{\alpha})_X}\ar{d}[swap]{\alpha_{X,\IBnu{X}}\out} &
      \IBnup{X}\ar{d}{\out'}\\
      \IB[\hash']{X}{\IBnu{X}} \ar{r}{\IB[\hash]{X}{(\IBnum{\alpha})_X}} & X\hash'\IBnup{X},
    \end{tikzcd}
  \end{equation*}
  i.e.\ $(\IBnum{\alpha})_X$ is the unique $(\IB[\hash']{X}{\argument})$-coalgebra
  morphism $(\IBnu{X},\alpha\out)\to(\IBnup{X},\out')$.

  We first check functoriality of $\IBnum{}$. For preservation of identities,
  just note that $\id\colon (\IBnu{X},\id\out)\to(\IBnu{X},\out)$ is a
  coalgebra morphism. For preservation of composition, we have that if
  $\beta\colon \hash'\to\hash''$ is a further parametrized monad morphism
  then by naturality of $\beta$, the $(\IB[\hash']{X}{\argument})$-coalgebra
  morphism $(\IBnum{\alpha})_X\colon \alpha\out\to\out'$ is also an
  $(\IB[\hash'']{X}{\argument})$-coalgebra morphism
  $\beta\alpha\out\to\beta\out'$; so $(\IBnum{\beta})_X(\IBnum{\alpha})_X$ is a
  coalgebra morphism $\beta\alpha\out\to\out''$, and hence equals
  $(\IBnum{\beta\alpha})_X$.
  
  It remains to verify that $\IBnum{\alpha}$ is indeed a monad morphism. 
  First, we show compatibility with the unit, i.e.\ 
  \begin{equation*}
    (\IBnum{\alpha})_X\,\eta^\nu_X = {\eta'}^\nu_X\colon X\to \IBnup{X}.
  \end{equation*}
  We note that by naturality of $\alpha$, the
  $(\IB{X}{\argument})$-coalgebra morphism
  $\eta^\nu_X\colon \eta_{X,\IBnu{X}}\to\out$ is also an
  $(\IB[\hash']{X}{\argument})$-coalgebra morphism
  $\eta'_{X,\IBnu{X}}=\alpha\comp\eta_{X,\IBnu{X}}\to\alpha\out$, so that
  $(\IBnum{\alpha})_X\comp\eta^\nu_X$ is a coalgebra morphism
  $\eta'_{X,\IBnu{X}}\to\out'$ and hence equals ${\eta'}^\nu_X$.

  For compatibility of $\IBnum{\alpha}$ with Kleisli lifting, we have to
  show that for $f\colon X\to \IBnu{Y}$, 
  \begin{equation*}
    (\IBnum{\alpha}\comp f)^{\kklstar'}\IBnum{\alpha}=\IBnum{\alpha}\comp f^\kklstar.
  \end{equation*}
  We strengthen this goal to one concerning $[f^\kklstar,\id]$, specifically
  we show that
  \begin{equation*}
    \begin{tikzcd}[column sep = 12ex]
      \IBnu{X}+\IBnu{Y} \ar{r}{[f^\kklstar\!,\,\id]}\ar{d}[swap]{(\IBnum{\alpha})_X+(\IBnum{\alpha})_Y}
        & \IBnu{Y}\ar{d}{(\IBnum{\alpha})_Y}\\
        \IBnup{X}+ \IBnup{Y} \ar{r}{[((\IBnum{\alpha})_Yf)^{\kklstar'}\!,\,\id]} &
        \IBnup{Y}        
    \end{tikzcd}
  \end{equation*}
  commutes. By definition, the bottom arrow is a
  $(Y\hash'\argument)$-coalgebra morphism 
  \begin{equation*}
    [\widehat{(\IBnum{\alpha})_Yf},(Y\hash'\inr)\out']\to\out',
  \end{equation*}
  and by now-familiar arguments, the top and right-hand arrows compose
  to yield a $(Y\hash'\argument)$-coalgebra morphism 
  \begin{equation*}
    [\alpha\hat f,\alpha(Y\hash\inr)\out]\to\out'.
  \end{equation*}
  It therefore suffices to show that 
  \begin{equation*}
    (\IBnum{\alpha})_X+(\IBnum{\alpha})_Y\colon [\alpha\hat f,\alpha(Y\hash\inr)\out]
    \to[\widehat{(\IBnum{\alpha})_Yf},(Y\hash'\inr)\out']    
  \end{equation*}
  is a $(Y\hash'\argument)$-coalgebra morphism. We first check
  commutation of the corresponding square on the right-hand summand
  $\IBnu{Y}$: 
  \begin{flalign*}
    &&& (Y\hash'((\IBnum{\alpha})_X+(\IBnum{\alpha})_Y))\alpha(Y\hash\inr)\out \\
    &&& =(Y\hash'((\IBnum{\alpha})_X+(\IBnum{\alpha})_Y))(Y\hash'\inr)\alpha\out 
      &\by{naturality of $\alpha$}\\
    &&& = (Y\hash'\inr)(Y\hash' (\IBnum{\alpha})_Y)\alpha\out \\
    &&& = (Y\hash'\inr)\out'(\IBnum{\alpha})_Y. 
   	  & \by{definition of $\IBnum{\alpha}$}
  \end{flalign*}
  For commutation on the left-hand summand we have to show that
  \begin{equation}\label{eq:N-lhs-mor}
    (Y\hash'((\IBnum{\alpha})_X+(\IBnum{\alpha})_Y))\comp\alpha\hat f =
    \widehat{(\IBnum{\alpha})_Yf}(\IBnum{\alpha})_X.    
  \end{equation}
  We rewrite the left-hand side of~\eqref{eq:N-lhs-mor}:
  \begin{flalign*}
    &&& (Y\hash'((\IBnum{\alpha})_X+(\IBnum{\alpha})_Y))\comp\alpha\hat f \\
    &&& = (Y\hash'((\IBnum{\alpha})_X+(\IBnum{\alpha})_Y))\comp\alpha\bar f^\klstar(X\hash\inl)\out
    & \by{definition of $\hat f$}\\
    &&& = (Y\hash'((\IBnum{\alpha})_X+(\IBnum{\alpha})_Y))\comp(\alpha\bar f)^{\klstar'}\alpha(X\hash\inl)\out
    & \by{$\alpha$ a monad morphism}\\
    &&& = (Y\hash'((\IBnum{\alpha})_X+(\IBnum{\alpha})_Y))\comp (\alpha\bar f)^{\klstar'}(X\hash'\inl)\alpha\out.
    & \by{naturality of $\alpha$} 
  \end{flalign*}
  We next rewrite the right-hand side of~\eqref{eq:N-lhs-mor}:
  \begin{flalign*}
    &&& \widehat{(\IBnum{\alpha})_Yf}(\IBnum{\alpha})_X \\
    &&& = \overline{(\IBnum{\alpha})_Yf}^{\klstar'}(X\hash'\inl)\out'(\IBnum{\alpha})_X
    & \by{definition of $\widehat{(\IBnum{\alpha})_Yf}$}\\
    &&& = \overline{(\IBnum{\alpha})_Yf}^{\klstar'}(X\hash'\inl)(X\hash'(\IBnum{\alpha})_X)
      \comp\alpha\out
    & \by{definition of $\widehat{(\IBnum{\alpha})_Yf}$}\\
    &&& = \overline{(\IBnum{\alpha})_Yf}^{\klstar'}(X\hash'((\IBnum{\alpha})_X+(\IBnum{\alpha})_Y))(\IB[\hash']{X}{\IBnum{\inl}})
      \comp\alpha\out
  \end{flalign*}
  It thus suffices to show that
  \begin{equation}
    \label{eq:N-lhs-mor-reduced}
    (Y\hash'((\IBnum{\alpha})_X+(\IBnum{\alpha})_Y))(\alpha\bar f)^{\klstar'}
    = \overline{(\IBnum{\alpha})_Yf}^{\klstar'}(X\hash'((\IBnum{\alpha})_X+(\IBnum{\alpha})_Y)).
  \end{equation}
  We further rewrite the right-hand side
  of~\eqref{eq:N-lhs-mor-reduced}: 
  \begin{flalign*}
    &&& \overline{(\IBnum{\alpha})_Yf}^{\klstar'}(X\hash'((\IBnum{\alpha})_X+(\IBnum{\alpha})_Y))\\
    &&& = ((Y\hash'\inr)\out'(\IBnum{\alpha})_Yf)^{\klstar'}\\
    &&&\qquad\quad  (X\hash'((\IBnum{\alpha})_X+(\IBnum{\alpha})_Y))
      & \by{definition of $\overline{(\IBnum{\alpha})_Yf}$}\\
    &&& = ((Y\hash'\inr)(Y\hash'(\IBnum{\alpha})_Y)\alpha\out f)^{\klstar'}\\
    &&&\qquad\quad (X\hash'((\IBnum{\alpha})_X+(\IBnum{\alpha})_Y))
      & \by{definition of $(\IBnum{\alpha})_Y$}\\
    &&& = ((Y\hash'((\IBnum{\alpha})_X+(\IBnum{\alpha})_Y))(Y\hash'\inr)\comp\alpha\out f)^{\klstar'}\\
    &&&\qquad\quad (X\hash'((\IBnum{\alpha})_X+(\IBnum{\alpha})_Y))\\
    &&& = (Y\hash'((\IBnum{\alpha})_X+(\IBnum{\alpha})_Y))((Y\hash'\inr)\comp\alpha\out f)^{\klstar'}
  \end{flalign*}
  where we use in the last step that
  $\IB[\hash']{Y}{((\IBnum{\alpha})_X+(\IBnum{\alpha})_Y)}$ is a monad morphism. We
  have thus reduced~\eqref{eq:N-lhs-mor-reduced} to showing that
  \begin{equation*}
    \alpha\bar f=(Y\hash'\inr)\alpha\out f.
  \end{equation*}
  But this is straightforward:
  \begin{flalign*}
    &&\alpha\bar f 
    & = \alpha\comp (Y\hash\inr)\out f & \by{definition of $\bar f$}\\
    &&& = (Y\hash'\inr)\comp\alpha\out f. & \by{naturality of $\alpha$}
  \end{flalign*}
\noindent  
  Next, we need to check the axioms of guarded monads for
  $\IBnum{\hash}$. 
\begin{citemize}
 \item\textbf{(trv)} Let $f\colon X\to \IBnum{\hash}Y$. Then 
\begin{align*}
 \out (\IBnum{\hash}\inj_1) f = (\IB{\inj_1}{(\IBnum{\hash}\inj_1)})\comp\out f. 
\end{align*}
By \textbf{(trv)} for $\IB{}{}$, $\out (\IBnum{\hash}\inj_1) f$ is $\inj_2$-guarded,
and thus, by definition so is $(\IBnum{\hash}{\inj_1}) f$.
 \item\textbf{(cmp)} Let $f\colon X\to_{2} \IBnum{\hash}(Y+Z)$, 
$g\colon Y\to_{\sigma} \IBnum{\hash}V$, $h\colon Z\to \IBnum{\hash}V$. Then 
we obtain
\begin{align*}
\out\comp [g,h]^\klstar f
=&\; [\out\comp g,\out\comp h]^\klstar\comp(\IB{\id}{[g,h]^\klstar})\comp\out f.
\end{align*}
By assumption $\out f$ is $\inj_2$-guarded, and therefore, since $\IB{}{}$ is a
parametrized guarded monad, so is $(\IB{\id}{[g,h]^\klstar})\comp\out f$.
Also, by assumption, $\out h$ is $\sigma$-guarded.
By~\textbf{(cmp)} for $\IB{}{}$, this implies that the composite
$[\out\comp g,\out\comp h]^\klstar\comp(\IB{\id}{[g,h]^\klstar})\comp\out f$
is $\sigma$-guarded and thus so is $[g,h]^\klstar f$.
\item\textbf{(par)} Let $f_i\colon X_i\to_\sigma\IBnum{\hash} Y$ for
  $i=1,2$, which by definition means that
  $\out f_i\colon X_i\to_\sigma Y\hash (\IBnum{\hash} Y)$. By \textbf{(par)}
  for $\hash$,
  $\out[f_1,f_2]=[\out f_1,\out f_2]\colon X_i\to_\sigma Y\hash
  (\IBnum{\hash} Y)$,
  so that $[f_1,f_2]\colon X_1+X_2\to_\sigma \IBnum{\hash} Y$ as required.
\end{citemize}

This shows that $\IBnum{\hash}$ is indeed a guarded monad; it remains
to show that given a parametrized guarded monad morphism
$\alpha\colon \hash\to\hash'$ as above, the monad morphism $\IBnum{\alpha}$
preserves guardedness. That is, for $f\colon Z\to_{\sigma}\IBnum{\hash}V$ we
have to show that $\IBnum{\alpha}f\colon Z\to_\sigma\IBnum{\hash'}V$, i.e.\
that $\out(\IBnum{\alpha}f)$ is $\sigma$-guarded. Indeed, by definition
of~$\IBnum{\alpha}$,
\begin{align*}
\out\comp (\IBnum{\alpha}\comp f) = (\IB{\id}{\IBnum{\alpha}})\comp\alpha\out f.
\end{align*}
By assumption, $\out f$ is $\sigma$-guarded and therefore, since $\IB{}{}$
is a parametrized guarded monad and $\alpha$ is a parametrized guarded
monad morphism, so is
$(\IB{\id}{\IBnum{\alpha}})\comp\alpha\comp\out f$.

\eqref{item:par-pre-iter}:  Let $f\colon X\to_2 \IB[\hash^\nu]{(Y+X)}{Z}$, and let
$g=\tuo\comp(\IB{\inl}{\id})\comp (\out f)^\istar\colon X\to_2 \IB[\hash^\nu]{(Y+X)}{Z}$.
Again, using the results of Uustalu~\cite[Theorem 3.11]{Uustalu03}, 
$h = \coit([\eta,(\out f)^\istar]^\klstar\out)\colon \IB[\hash^\nu]{(Y+X)}{Z}\to_2 \IB[\hash^\nu]{Y}{Z}$
is the unique solution of equation
\begin{displaymath}
	h = [\eta^\nu, h\comp g]^\kklstar,
\end{displaymath}
which implies that $f^\iistar = h\comp \eta^\nu\inr$ is a fixpoint of $g$. Indeed,
$f^\iistar = h\comp \eta^\nu\inr = [\eta^\nu, h\comp g]^\kklstar\eta^\nu\inr = h\comp g$,
and thus, $[\eta^\nu, f^\iistar]^\kklstar\comp g = [\eta^\nu, h\comp g]^\kklstar\comp g = h\comp g = f^\iistar$.
We are left to check that $f^\iistar$ is also a fixpoint of $f$. First, we record the
auxiliary equation
\begin{align}\label{eq:iter-trans-aux}
\out [\eta^\nu,f^\iistar]^\kklstar = [\eta, \out f^\iistar]^\klstar (\IB{\id}{[\eta^\nu,f^\iistar]^\kklstar})\out,
\end{align}
which entails the goal as follows (using the fact that $\out$ is an isomorphism):
\begin{flalign*}
&&\out [\eta^\nu,f^\iistar]^\kklstar\comp f
&= [\eta, \out f^\iistar]^\klstar (\IB{\id}{[\eta^\nu,f^\iistar]^\kklstar})\out f&\by{\eqref{eq:iter-trans-aux}}\\
&&&\;= [\eta, \out[\eta^\nu,f^\iistar]^\kklstar\comp g]^\klstar (\IB{\id}{[\eta^\nu,f^\iistar]^\kklstar})\out f&\by{definition of~$f^\iistar$}\\
&&&\;=[\eta, [\eta, \out f^\iistar]^\klstar (\IB{\id}{[\eta^\nu,f^\iistar]^\kklstar})\out g]^\klstar\\
&&&\qquad\quad (\IB{\id}{[\eta^\nu,f^\iistar]^\kklstar})\out f&\by{\eqref{eq:iter-trans-aux}}\\
&&&\;=[\eta, [\eta, \out f^\iistar]^\klstar (\IB{\inl}{[\eta^\nu,f^\iistar]^\kklstar})(\out f)^\istar]^\klstar \\
&&&\qquad\quad (\IB{\id}{[\eta^\nu,f^\iistar]^\kklstar})\out f&\by{definition of~$g$}\\
&&&\;=[\eta, \out f^\iistar]^\klstar [\eta\inl, (\IB{\inl}{[\eta^\nu,f^\iistar]^\kklstar})(\out f)^\istar]^\klstar\\
&&&\qquad\quad (\IB{\id}{[\eta^\nu,f^\iistar]^\kklstar})\out f\\
&&&\;=[\eta, \out f^\iistar]^\klstar (\IB{\inl}{[\eta^\nu,f^\iistar]^\kklstar}) [\eta,(\out f)^\istar]^\klstar\out f&\by{\eqref{eq:par-moph}}\\
&&&\;=[\eta, \out f^\iistar]^\klstar (\IB{\inl}{[\eta^\nu,f^\iistar]^\kklstar})(\out f)^\istar&\by{definition of~$(\argument)^\istar$}\\
&&&\;=[\eta, \out f^\iistar]^\klstar (\IB{\id}{[\eta^\nu,f^\iistar]^\kklstar})\out g&\by{definition of~$g$}\\
&&&\;=\out [\eta^\nu,f^\iistar]^\kklstar\comp g&\by{\eqref{eq:iter-trans-aux}}\\
&&&\;= \out f^\iistar.&\by{definition of~$f^\iistar$}
\intertext{
Equation~\eqref{eq:iter-trans-aux} is derived as follows:
}
&&\out[\eta^\nu,f^\iistar]^\kklstar 
&= \bigl(\IB{\id}{[[\eta^\nu,f^\iistar]^\kklstar,\id]}\bigr) \bigl((\IB{\id}{\inr})\out\comp [\eta^\nu,f^\iistar]\bigr)^\klstar\\
&&&\qquad\quad  (\IB{\id}{\inl})\out&\by{definition of~$(\argument)^\kklstar$}\\
&&&\;=  \bigl((\IB{\id}{[[\eta^\nu,f^\iistar]^\kklstar,\id]})(\IB{\id}{\inr})\out\comp [\eta^\nu,f^\iistar]\bigr)^\klstar\\
&&&\qquad\quad (\IB{\id}{[[\eta^\nu,f^\iistar]^\kklstar,\id]})\comp(\IB{\id}{\inl})\out&\by{\eqref{eq:par-moph}}\\
&&&\;=  (\out\comp [\eta^\nu,f^\iistar])^\klstar (\IB{\id}{[\eta^\nu,f^\iistar]^\kklstar})\out\\
&&&\;= [\eta, \out f^\iistar]^\klstar (\IB{\id}{[\eta^\nu,f^\iistar]^\kklstar})\out.&\by{definition of~$\eta^\nu$}
\end{flalign*}
Property~\eqref{eq:par-iter-preserve} transfers routinely along $\IB{}{}\mto\IB[\hash^\nu]{}{}$. 

\eqref{item:par-iter}:  We have to show that, given
$f\colon X\to_2 \IB[\hash^\nu]{(Y+X)}{Z}$ and
$\hat f\colon X\to \IB[\hash^\nu]{Y}{Z}$ such that
$\hat f = [\eta^\nu, \hat f]^\kklstar f$, we have
$\hat f = f^\iistar$, with $f^\iistar$ defined as in
Claim~(\ref{item:par-pre-iter}).  Again, let
$g=\tuo\comp(\IB{\inl}{\id})\comp (\out f)^\istar\colon X\to_2
\IB[\hash^\nu]{(Y+X)}{Z}$.
As we indicated above, $f^\iistar$ is the unique solution of the
equation $[\eta^\nu, f^\iistar]^\kklstar\comp g = f^\iistar$, and thus
to obtain the desired identity $\hat f = f^\iistar$, it suffices to
prove the same equation for $\hat f$. Note
that~\eqref{eq:iter-trans-aux} remains valid for $\hat f$ instead of
$f^\iistar$ and therefore we obtain
\begin{displaymath}
	\out\hat f = \out [\eta^\nu,\hat f]^\kklstar f = [\eta, \out \hat f]^\klstar (\IB{\id}{[\eta^\nu,f^\iistar]^\kklstar})\out f,
\end{displaymath}
which implies $\out\hat f = ((\IB{\id}{[\eta^\nu,f^\iistar]^\kklstar})\out f)^\istar$, for 
$(\IB{\id}{[\eta^\nu,f^\iistar]^\kklstar})\out f$ is $\inr$-guarded, and therefore
has a unique fixpoint. Now, since
\begin{flalign*}
&& \out\hat f
&  \;= ((\IB{\id}{[\eta^\nu,f^\iistar]^\kklstar})\out f)^\istar\\
&&&\;= (\IB{\id}{[\eta^\nu,f^\iistar]^\kklstar})\comp (\out f)^\istar&\by{\eqref{eq:par-iter-preserve}} \\
&&&\;= [\eta, \out f^\iistar]^\klstar (\IB{\id}{[\eta^\nu,f^\iistar]^\kklstar})\out\tuo \comp(\IB{\inl}{\id}) \comp (\out f)^\istar\\
&&&\;= \out [\eta^\nu,f^\iistar]^\kklstar g&\by{\eqref{eq:iter-trans-aux}, definition of~$g$}\\
&&&\;= \out f^\iistar,&\by{definition of~$f^\iistar$}
\end{flalign*}
we obtain $f^\iistar = \hat f$ using the fact that $\out$ is an
isomorphism.
\end{proof}
\begin{remark}\label{rem:t-nu}
  The definitions figuring in Theorem~\ref{thm:ext_guard} specialize
  to two generic cases occurring in previous literature:
\begin{cenumerate}
\item With $\BD=1$, $\IB{}{} = T(\argument+\Sigma\argument)$ for an
  endofunctor $\Sigma$ and a totally guarded pre-iterative monad
  $\BBT=(T,\eta,\argument^\klstar,\argument^\istar)$, we obtain the
  setting studied by Goncharov et al.~\cite{GoncharovEA18}: 
  $\IBnu{}$ is isomorphically a monad $\BBT_{\Sigma}$ on
  $\BC$ with $T_{\Sigma} X=\nu\gamma.\,T(X+\Sigma\gamma)$, unit
  $\eta^\nu = \tuo\eta\inl$, with Kleisli lifting
  $(f\colon X\to T_\Sigma Y)^\kklstar$ uniquely determined by the
  equation
\begin{align*}
\out f^\kklstar = [\out f\comma\eta\inr \Sigma f^\kklstar]^\klstar\out,
\end{align*}
and with the total iteration operator
\begin{align*}
\bigl(f\colon X\to T_\Sigma(Y+X)\bigr)^\iistar=\coit\bigl(\bigl[[\eta\inl,(T[\inl+\id,\inl\inr]\out f)^\istar],\eta\inr\bigr]^\klstar\out\bigr)\comp\eta^\nu\inr.
\end{align*}
\item With $\BD=1$, and any vacuously guarded $\IB{}{}\colon \BC\times\BC\to\BC$, we obtain the setting 
of Uustalu~\cite{Uustalu03}, with the guarded iterative monad $\IBnu{}=\nu\gamma.\,\IB{\argument}{\gamma}$
defined as follows: The monad structure is specified by~\eqref{eq:IBnu-monad}, and the
iteration operator $(\argument)^\iistar$ is uniquely determined by the equation 
$[\eta^\nu, f^\iistar]^\kklstar = f^\iistar$ for every $\inr$-guarded $f\colon X\to \IBnu{(Y+X)}$.
According to Theorem~\ref{thm:ext_guard}~(2),~$f$ is $\inr$-guarded
iff $\out f\colon X\to \IB{(Y+X)}{\IBnu{(Y+X)}}$ factors through $\IB{\inl}{\id}\colon \IB{Y}{\IBnu{(Y+X)}}\to\IB{(Y+X)}{\IBnu{(Y+X)}}$,
which is precisely the notion of guardedness in~\cite{Uustalu03}.
\end{cenumerate}
\end{remark}
\begin{example}\label{ex:transfer-example}
We proceed to illustrate the use of Theorem~\ref{thm:ext_guard} by various instances 
of Example~\ref{ex:gen-proc}.
\begin{cenumerate}
  \item By equipping the finite powerset monad $\FSet$ on $\Set$ with vacuous  
guardedness, we obtain by Theorem~\ref{thm:ext_guard} a notion of guardedness for $\nu\gamma.\,\FSet(X+A\times\gamma)$,
which allows for systems consisting of equations of the form
\begin{align}\label{eq:geq-example}
 x = y_1 + \ldots + y_n + a_1.\, t_1 + \ldots + a_m.\,t_m 
\end{align}
where the variables $y_i$ are not allowed to occur on the left-hand side
and the terms $t_i$ represent elements of $\nu\gamma.\,\FSet(X+A\times\gamma)$,
as previously explained in Example~\ref{ex:gen-proc}.  By Proposition~\ref{prop:triv-iter} and
Theorem~\ref{thm:ext_guard}, we conclude that these systems have unique
solutions, which is of course a known fact in process algebra. This and the 
following examples are intentionally chosen to be simple for illustrative 
purposes, but we emphasize that the same principles apply to the examples obtained 
by replacing $\FSet$ with a more general $T$ and $A\times\argument$ with a more general $\Sigma$.
For example, replacing $T$ with a \emph{subdistribution monad}~\cite{HasuoJacobsEtAl07},
to model probability instead on nondeterminism, would require changing the format 
of~\eqref{eq:geq-example}~to
\begin{align*}
 x = p_1\cdot y_1 + \ldots + p_n\cdot y_n + p'_1\cdot a_1.\, t_1 + \ldots + p'_m\cdot a_m.\,t_m 
\end{align*}
where the non-negative real coefficients $p_1,\ldots,p_n,p_1',\ldots,p_m'$, subject 
to the condition $p_1+\ldots+p_n+p_1'+\ldots+p_m'\leq 1$, represent the probabilities 
of choosing the corresponding alternative.
\item By replacing $\FSet$ with countable powerset $\CSet$ in the
  previous clause, we can relax the format of equation systems that
  can be solved, at the price of losing uniqueness of
  solutions. Specifically, let $\CSet$ be totally guarded
  pre-iterative with solutions of $f\colon X\to \CSet(Y+X)$ calculated via
  least fixpoints. The derived notion of guardedness for
  $\nu\gamma.\,\CSet(X+A\times\gamma)$ according to
  Theorem~\ref{thm:ext_guard} is again total, i.e.\ allows solving
  arbitrary systems of equations (we discuss an application of such
  unguarded recursive process definitions
  in~\cite[Section~3]{GoncharovEA18}; specifically, they allow
  defining countably branching systems in basic process algebra). 
  The canonical derived iteration operator makes use of both least
  fixpoints and unique coalgebraic fixpoints. For example, the
  canonical solution of
\begin{align}\label{eq:a-omega}
  x = x + a.\,x
\end{align} is
  the infinite sequence $x = a^\omega$, seen as an element of the
  final countably branching labelled transition system
  $\nu\gamma.\,\CSet(A\times\gamma)$ -- intuitively, the original
  system~\eqref{eq:a-omega} is first collapsed to $x = a.\,x$, iterating away the first
  $x$ in the sum by taking a least fixpoint, and the resulting system
  is solved uniquely. In detail, the definitions in
  Theorem~\ref{thm:ext_guard} unfold as follows. We have the case
  mentioned in Example~\ref{ex:param-monad} where
  $\IB{X}{Y}=\CSet(X+A\times Y)$ (so
  $\IBnu{X}=\nu\gamma.\CSet(X+A\times\gamma)$). Our example equation~\eqref{eq:a-omega} 
  corresponds to the map $f\colon X\to\IBnu(\iobj+X)$ where $X=\{x\}$
  and
  $\out f(x)=\{\inl\inr x,\inr\brks{a,\tuo(\{x\})}\}\in\CSet((\iobj+X)+A\times\IBnu(\iobj+X))\iso\CSet(A\times\IBnu{X}+X)$. The
  definition of $(-)^\iistar$ according to Theorem~\ref{thm:ext_guard}
  now tells us to first iterate $\out f$ in $\CSet$, by taking a
  least fixpoint with~$x$ seen as a variable, obtaining
  $(\out f)^\istar\colon X\to\CSet(A\times\IBnu{X})$ where
  \begin{equation*}
    (\out f)^\istar(x)=\{\brks{a,\tuo(\{\inl x\})}\}.
  \end{equation*}
  We next form the map
  $g=(\CSet\inr)[(\out f)^\istar,\eta]^\klstar\out\colon \IBnu{X}\to
  \CSet(\iobj+A\times\IBnu{X})$,
  where~$\eta$ and $(\argument)^\klstar$ are the unit and the Kleisli lifting of $\CSet$, so for
  $t\in\IBnu{X}=\nu\gamma.\,\CSet(X+A\times\gamma)$,
  \begin{align*}
     g(t) & =\{\inr(\out f)^\istar(x)\mid \inl x\in\out t\} \cup\{\inr\brks{b,s}\in A\times\IBnu{X}\mid \inr\brks{b,s}\in\out t\}\\
     & =\{\inr\brks{a,\tuo(\{x\}})\mid \inl x\in\out t\} \cup\{\inr\brks{b,s}\in A\times\IBnu{X}\mid \inr\brks{b,s}\in\out t\}.
  \end{align*}
  We then obtain a final coalgebra morphism
  $\coit g\colon \IBnu{X}\to\IBnu\iobj=\nu\gamma.\, \CSet(A\times\gamma)$.
  The solution $f^\iistar(x)$ is obtained by applying $\coit g$ to
  $\eta^\nu_{X}(x)=\tuo(\eta_{X,\IBnu X}(x))=\tuo(\{\inl x\})$, using
  the description of $\eta^\nu$ recalled in the proof of
  Theorem~\ref{thm:ext_guard}. Since
  $g(\tuo(\{\inl x\}))=\{\inr\brks{a,\tuo(\{\inl x\})}\}$, we obtain
  that $f^\iistar(x)$ is $a^\omega$, as expected.
\item Consider a further variation of the same example obtained by
  replacing $A$ in the previous example by $1+A$, where the adjoined
  element is supposed to capture the invisible action $\tau$ in the
  usual sense of process algebra~\cite{Milner89}.  Applying
  Theorem~\ref{thm:ext_guard} to $\IB{X}{Y}=\CSet(X+(1+A)\times Y)$ as
  in the previous example, we would derive a notion of guardedness
  that identifies as guarded any recursive call preceded by an action,
  visible or not.  We can refine this view by allowing only visible
  actions as guards, which is in fact standard for
  CCS~\cite{Milner89}. To this end, consider the obvious isomorphism
\begin{align*}
\nu\gamma.\,\CSet(X+(1+A)\times\gamma)\cong \nu\gamma'.\,\nu\gamma.\,\CSet(X+\gamma + A\times\gamma'),  
\end{align*}
which involves two more parametrized monads:
$\CSet(\argument+\argument + A\times\argument)\colon\Set\times
(\Set\times\Set)\to\Set$
and
$\nu\gamma.\,\CSet(\argument+\gamma +
A\times\argument)\colon\Set\times\Set\to\Set$.
The latter parametrized monad is formed on top of the former. We equip
$\nu\gamma.\,\CSet(\argument+\gamma + A\times\argument)$ with the
vacuous notion of guardedness. By furthermore forming the fixpoint
$\nu\gamma'.\,\nu\gamma.\,\CSet(X+\gamma + A\times\gamma')$, we obtain
precisely the notion of guardedness we aimed at for the isomorphic
monad $\IB[\hash^\nu]{}{}$.
 \item Consider $TX = \CSet(\mu\gamma.\, X + 1 + A\times\gamma)$, which can be understood as a semantic
domain for processes with results in $X$ as before, but now modulo \emph{finite trace 
equivalence} instead of strong bisimilarity as the underlying equivalence relation: the elements of $TX$ are sets of traces from 
$\mu\gamma.\, X + 1 + A\times\gamma\cong A^\star + A^\star\times X$ consisting 
of terminating traces (from $A^\star\times X$) and non-terminating traces (from $A^\star$).
In order to apply our theory to this example, we make use of Hasuo et al.'s results on
\emph{coalgebraic finite trace semantics}~\cite{HasuoJacobsEtAl07}. Specifically,
we make use of the fact that due to presence of a canonical distributive law
\begin{displaymath}
	X + 1 + A\times\CSet \to \CSet(X + 1 + A\times\argument) 
\end{displaymath}
and a suitable order-enrichment of $\CSet$, the object
$\mu\gamma.\, X + 1 + A\times\gamma$ computed in $\Set$ 
carries a final coalgebra $\nu\gamma.\, X + 1 + A\times\gamma$ in the Kleisli category
of $\CSet$. In this category we equip the parametrized monad
$\IB{}{} = \argument + 1 + A\times\argument$ with the vacuous notion of
guardedness and thus derive the notion of guardedness for
$\IB[\hash^\nu]{}{}$, allowing exactly for recursive calls preceded by
actions from $A$.  Again, by Proposition~\ref{prop:triv-iter} and by
Theorem~\ref{thm:ext_guard}~(3), the obtained monad is guarded
iterative.

Note that the monad $\CSet(\mu\gamma.\, \argument + 1 + A\times\gamma)$ is arguably too large, as it contains sets
of traces not realized by any process from
$\nu\gamma.\,\CSet(X+A\times\gamma)$. This can easily be fixed by
cutting down to the submonad of $\CSet(\mu\gamma.\, \argument + 1 + A\times\gamma)$ consisting of
the \emph{prefix-closed} sets of traces, i.e.\ such sets $S$ that
$st\in S$ implies $s\in S$ and $\brks{st,x}\in S$ implies $s\in S$. It
is easy to see that this is a guarded pre-iterative submonad of
$\BBT$, and therefore guarded iterative.
\end{cenumerate}
\end{example}

\section{Complete Elgot Monads and Iteration Congruences}\label{sec:congruence}
\begin{figure}[t!]
  \tikzset{
    font=\tiny,
    nonterminal/.style={
      rectangle,
      minimum size=6mm,
      very thick,
      draw=orange!50!black!50,         %
      fill=orange!50!white,
      font=\itshape
    },
    terminal/.style={
      scale=.5,
      circle,
      inner sep=0pt,
      thin,draw=black!50,
      top color=white,bottom color=black!20,
      font=\ttfamily
    },
    iterated/.style={
      fill=green!20,
      thick,
      draw=green!50
    },
    natural/.style={
      circle,
      minimum size=4mm,
      inner sep=2pt,
      thin,draw=black!50,
      top color=white,bottom color=black!20,
      font=\ttfamily},
    skip loop/.style={to path={-- ++(0,#1) -| (\tikztotarget)}},
    o/.style={
      shorten >=#1,
      decoration={
        markings,
        mark={
          at position 1
          with {
            \fill[black!55] circle [radius=#1];
          }
        }
      },
      postaction=decorate
    },
    o/.default=2.5pt,
    p/.style={
      shorten <=#1,
      decoration={
        markings,
        mark={
          at position 0
          with {
            \fill[black!55] circle [radius=#1];
          }
        }
      },
      postaction=decorate
    },
    p/.default=2.5pt
  }

  {
    \tikzset{nonterminal/.append style={text height=1.5ex,text depth=.25ex}}
    \tikzset{natural/.append style={text height=1.5ex,text depth=.25ex}}
  }
  \captionsetup[subfigure]{labelformat=empty,justification=justified,singlelinecheck=false}
  \pgfdeclarelayer{background}
  \pgfdeclarelayer{foreground}
  \pgfsetlayers{background,main,foreground}
    \begin{subfigure}{\textwidth}
    \centering
    \caption{Fixpoint:}
    \vspace{-2ex}
    \raisebox{-.5\height}{
    \begin{tikzpicture}[
      point/.style={coordinate},>=stealth',thick,draw=black!50,
      tip/.style={->,shorten >=0.007pt},every join/.style={rounded corners},
      hv path/.style={to path={-| (\tikztotarget)}},
      vh path/.style={to path={|- (\tikztotarget)}},
      text height=1.5ex,text depth=.25ex %
      ]
      \node [nonterminal] (f) {$f$};
      \draw [<-] (f.west) -- +(-1,0) node [midway,above] {$X$};
      \path [o,<-,draw] ($(f.west)+(-0.5,0)$) -- +(0,-0.8) -| ($(f.east)+(0.5,-0.15)$) node [pos=0.8,right] {$X$} -- +(-0.5,0);
      \draw [->] (f.east)++(0,0.15) -- +(1,0) node [midway,above] {$Y$};
      
      \begin{pgfonlayer}{background}
        \draw [iterated] ($(f.north west)+(-0.25,0.25)$) rectangle ($(f.south east)+(0.25,-0.25)$);
      \end{pgfonlayer}
    \end{tikzpicture}
    }
    ~~=~~
    \raisebox{-.5\height}{
    \begin{tikzpicture}[
      point/.style={coordinate},>=stealth',thick,draw=black!50,
      tip/.style={->,shorten >=0.007pt},every join/.style={rounded corners},
      hv path/.style={to path={-| (\tikztotarget)}},
      vh path/.style={to path={|- (\tikztotarget)}},
      text height=1.5ex,text depth=.25ex %
      ]
      \node [nonterminal] (f) {$f$};
      \node [nonterminal] (f2) at ($(f.east)+(1.5,-0.15)$) {$f$};
      \draw [<-] (f.west) -- +(-1,0) node [midway,above] {$X$};
      \draw [p,->] (f.east)++(0,-0.15) -- (f2.west) node [pos=0.25,below] {$X$};
      \path [o,<-,draw] (f2.west)++(-0.5,0) -- ++(0,-0.8) -| ($(f2.east)+(0.5,-0.15)$) node [near end,right] {$X$} -- +(-0.5,0);
      \draw [->] (f.east)++(0,0.15) -- node [midway,above] {$Y$} ++(0.5,0) -- ++(0,0.5) -- ++(3,0);
      \draw [->] (f2.east)++(0,0.15) -- ++(0.5,0) -- node [midway,right] {$Y$} ++(0,0.65);

      \begin{pgfonlayer}{background}
        \draw [iterated] ($(f2.north west)+(-0.25,0.25)$) rectangle ($(f2.south east)+(0.25,-0.25)$);
      \end{pgfonlayer}
    \end{tikzpicture}
    }
  \end{subfigure}
  \par%
  \begin{subfigure}{\textwidth}
    \centering
    \caption{Naturality:}
    \raisebox{-.5\height}{
    \begin{tikzpicture}[
      point/.style={coordinate},>=stealth',thick,draw=black!50,
      tip/.style={->,shorten >=0.007pt},every join/.style={rounded corners},
      hv path/.style={to path={-| (\tikztotarget)}},
      vh path/.style={to path={|- (\tikztotarget)}},
      text height=1.5ex,text depth=.25ex %
      ]
      \node [nonterminal] (f) {$f$};
      \node [nonterminal] (g) at ($(f.east)+(1.5,0.15)$) {$g$};
      \draw [<-] (f.west) -- +(-1,0) node [midway,above] {$X$};
      \path [o,<-,draw] (f.west)++(-0.5,0) -- ++(0,-0.8) -| ($(f.east)+(0.5,-0.15)$) node [near end,right] {$X$}  -- ++(-0.5,0);
      \draw [->] ($(f.east)+(0,0.15)$) -- (g) node [midway,above] {$Y$};
      \draw [->] (g.east) -- +(1,0) node [midway,above] {$Z$};
      \begin{pgfonlayer}{background}
        \draw [iterated] ($(f.north west)+(-0.25,0.25)$) rectangle ($(f.south east)+(0.25,-0.25)$);
      \end{pgfonlayer}
    \end{tikzpicture}
    }
    ~~=~~
    \raisebox{-.5\height}{
    \begin{tikzpicture}[
      point/.style={coordinate},>=stealth',thick,draw=black!50,
      tip/.style={->,shorten >=0.007pt},every join/.style={rounded corners},
      hv path/.style={to path={-| (\tikztotarget)}},
      vh path/.style={to path={|- (\tikztotarget)}},
      text height=1.5ex,text depth=.25ex %
      ]
      \node [nonterminal] (f) {$f$};
      \node [nonterminal] (g) at ($(f.east)+(1.5,0.15)$) {$g$};
      \draw [<-] (f.west) -- +(-1,0) node [midway,above] {$X$};
      \path [o,<-,draw] (f.west)++(-0.5,0) 
        -- ++(0,-0.75) 
        -| ($(g.east)+(0.5,-0.5)$) node [near end,right] {$X$} 
        -| ($(f.east)+(0.6,-0.15)$) 
        -- ++(-0.6,0);
      \draw [->] ($(f.east)+(0,0.15)$) -- (g) node [midway,above] {$Y$};
      \draw [->] (g.east) -- +(1,0) node [midway,above] {$Z$};
      \begin{pgfonlayer}{background}
        \draw [iterated] ($(f.north west)+(-0.25,0.25)$) rectangle ($(g.south east)+(0.25,-0.35)$);
      \end{pgfonlayer}
    \end{tikzpicture}
    }
  \end{subfigure}
  \par\medskip
  \begin{subfigure}{\textwidth}
    \centering
    \vspace{1ex}
    \caption{Codiagonal:}
    \vspace{-1ex}
    \raisebox{-.5\height}{
    \begin{tikzpicture}[
      point/.style={coordinate},>=stealth',thick,draw=black!50,
      tip/.style={->,shorten >=0.007pt},every join/.style={rounded corners},
      hv path/.style={to path={-| (\tikztotarget)}},
      vh path/.style={to path={|- (\tikztotarget)}},
      text height=1.5ex,text depth=.25ex %
      ]
      \node [nonterminal,minimum height=1.2cm] (f) {$g$};
      \draw [<-] (f.west) -- +(-1,0) node [midway,above] {$X$};
      \draw [->] (f.east)++(0,0.4) -- ++(1.5,0) node [pos=0.3,above] {$Y$};
      \draw [p,->] (f.east) -- ++(1,0) --node [midway,right] {$X$}  ++(0,-1.1) -| ($(f.west)+(-0.5,0)$);
      \draw [p,->] (f.east)++(0,-0.4) -- node [midway,below] {$X$} ++(0.5,0) -- ++(0,0.4);
      \begin{pgfonlayer}{background}
        \draw [iterated] ($(f.north west)+(-0.25,0.25)$) rectangle ($(f.south east)+(0.75,-0.25)$);
      \end{pgfonlayer}
    \end{tikzpicture}
    }
    ~~=~~
    \raisebox{-.5\height}{
    \begin{tikzpicture}[
      point/.style={coordinate},>=stealth',thick,draw=black!50,
      tip/.style={->,shorten >=0.007pt},every join/.style={rounded corners},
      hv path/.style={to path={-| (\tikztotarget)}},
      vh path/.style={to path={|- (\tikztotarget)}},
      text height=1.5ex,text depth=.25ex %
      ]
      \node [nonterminal,minimum height=1cm] (f) {$g$};
      \draw [<-] (f.west) -- +(-1.6,0) node [pos=0.7,above] {$X$};
      \draw [->] (f.east)++(0,0.3) -- ++(1.5,0) node [pos=0.3,above] {$Y$};
      \draw [p,->] (f.east) -- ++(1.1,0) -- node [midway,right] {$X$} ++(0,-1.35) -| ($(f.west)+(-0.95,0)$);
      \draw [p,->] (f.east)++(0,-0.3) -- ++(0.5,0) -- node [midway,right] {$X$} ++(0,-0.65) -| ($(f.west)+(-0.5,0)$);
      \begin{pgfonlayer}{background}
        \draw [iterated] ($(f.north west)+(-0.75,0.35)$) rectangle ($(f.south east)+(0.85,-0.6)$);
        \draw [iterated,fill=green!50] ($(f.north west)+(-0.25,0.2)$) rectangle ($(f.south east)+(0.25,-0.25)$);
      \end{pgfonlayer}
    \end{tikzpicture}
    }
  \end{subfigure}
  \par%
  \begin{subfigure}{\textwidth}
    \caption{Uniformity:}
    \centering
\begin{tabular}{rcl}
   \raisebox{-.5\height}{
    \begin{tikzpicture}[
      point/.style={coordinate},>=stealth',thick,draw=black!50,
      tip/.style={->,shorten >=0.007pt},every join/.style={rounded corners},
      hv path/.style={to path={-| (\tikztotarget)}},
      vh path/.style={to path={|- (\tikztotarget)}},
      text height=1.5ex,text depth=.25ex %
      ]
      \node [nonterminal,fill=blue!20,draw=blue!50] (h) {$h$};
      \node [nonterminal] (f) at ($(h.east)+(1.5,0)$) {$f$};
      \draw [<-] (h.west) -- +(-1,0) node [midway,above] {$Z$};
      \draw [->] (h.east) -- (f.west) node [midway,above] {$X$};
      \draw [->] (f.east)++(0,0.15) -- +(1,0) node [midway,above] {$Y$};
      \draw [p,->] (f.east)++(0,-0.15) -- +(1,0) node [midway,below] {$X$};
    \end{tikzpicture}
    }
    &$~~=~~$&
    \raisebox{-.5\height}{
    \begin{tikzpicture}[
      point/.style={coordinate},>=stealth',thick,draw=black!50,
      tip/.style={->,shorten >=0.007pt},every join/.style={rounded corners},
      hv path/.style={to path={-| (\tikztotarget)}},
      vh path/.style={to path={|- (\tikztotarget)}},
      text height=1.5ex,text depth=.25ex %
      ]
      \node [nonterminal] (f) {$g$};
      \node [nonterminal,fill=blue!20,draw=blue!50] (h) at ($(f.east)+(1.5,-0.15)$) {$h$};
      \draw [<-] (f.west) -- +(-1,0) node [midway,above] {$Z$};
      \draw [p,->] (f.east)++(0,-0.15) -- (h.west) node [midway,below] {$Z$};
      \draw [->] (f.east)++(0,0.15) -- ++(0.65,0) -- ++(0,0.4) node [midway,left] {$Y$} -- ++(2.35,0);
      \draw [->] (h.east) -- +(1,0) node [midway,above] {$X$};
    \end{tikzpicture}
    }\\
    &\Large{$\Downarrow$}&\\
    \raisebox{-.5\height}{
    \begin{tikzpicture}[
      point/.style={coordinate},>=stealth',thick,draw=black!50,
      tip/.style={->,shorten >=0.007pt},every join/.style={rounded corners},
      hv path/.style={to path={-| (\tikztotarget)}},
      vh path/.style={to path={|- (\tikztotarget)}},
      text height=1.5ex,text depth=.25ex %
      ]
      \node [nonterminal,fill=blue!20,draw=blue!50] (h) {$h$};
      \node [nonterminal] (f) at ($(h.east)+(1.5,0)$) {$f$};
      \draw [<-] (h.west) -- +(-1,0) node [midway,above] {$Z$};
      \draw [->] (h.east) -- (f.west) node [midway,above] {$X$};
      \draw [->] (f.east)++(0,0.15) -- +(1,0) node [midway,above] {$Y$};
      \path [o,<-,draw] (f.west)++(-0.5,0) -- ++(0,-0.8) -|
      ($(f.east)+(0.5,-0.15)$) node [pos=0.8,right] {$X$} -- ++(-0.5,0);
      \begin{pgfonlayer}{background}
        \draw [iterated] ($(f.north west)+(-0.25,0.25)$) rectangle ($(f.south east)+(0.25,-0.25)$);
      \end{pgfonlayer}
    \end{tikzpicture}
    }
    &$~~=~~$&
    \raisebox{-.5\height}{
    \begin{tikzpicture}[
      point/.style={coordinate},>=stealth',thick,draw=black!50,
      tip/.style={->,shorten >=0.007pt},every join/.style={rounded corners},
      hv path/.style={to path={-| (\tikztotarget)}},
      vh path/.style={to path={|- (\tikztotarget)}},
      text height=1.5ex,text depth=.25ex %
      ]
      \node [nonterminal] (f) {$g$};
      \draw [<-] (f.west) -- +(-1,0) node [midway,above] {$Z$};
      \path [o,<-,draw] ($(f.west)+(-0.5,0)$) -- +(0,-0.8) -| ($(f.east)+(0.5,-0.15)$) node [pos=0.8,right] {$Z$} -- +(-0.5,0);
      \draw [->] (f.east)++(0,0.15) -- +(1,0) node [midway,above] {$Y$};
      
      \begin{pgfonlayer}{background}
        \draw [iterated] ($(f.north west)+(-0.25,0.25)$) rectangle ($(f.south east)+(0.25,-0.25)$);
      \end{pgfonlayer}
    \end{tikzpicture}
    }
  \end{tabular}
  \end{subfigure}
\vspace{2ex}
\caption{Axioms of guarded iteration.}
\label{fig:ax}
\end{figure}

%
%
%
%

\begin{figure}[t!]
  \tikzset{
    font=\tiny,
    nonterminal/.style={
      rectangle,
      minimum size=6mm,
      very thick,
      draw=orange!50!black!50,         %
      fill=orange!50!white,
      font=\itshape
    },
    terminal/.style={
      scale=.5,
      circle,
      inner sep=0pt,
      thin,draw=black!50,
      top color=white,bottom color=black!20,
      font=\ttfamily
    },
    iterated/.style={
      fill=green!20,
      thick,
      draw=green!50
    },
    natural/.style={
      circle,
      minimum size=4mm,
      inner sep=2pt,
      thin,draw=black!50,
      top color=white,bottom color=black!20,
      font=\ttfamily},
    skip loop/.style={to path={-- ++(0,#1) -| (\tikztotarget)}},
    o/.style={
      shorten >=#1,
      decoration={
        markings,
        mark={
          at position 1
          with {
            \fill[black!55] circle [radius=#1];
          }
        }
      },
      postaction=decorate
    },
    o/.default=2.5pt,
    p/.style={
      shorten <=#1,
      decoration={
        markings,
        mark={
          at position 0
          with {
            \fill[black!55] circle [radius=#1];
          }
        }
      },
      postaction=decorate
    },
    p/.default=2.5pt
  }

  {
    \tikzset{nonterminal/.append style={text height=1.5ex,text depth=.25ex}}
    \tikzset{natural/.append style={text height=1.5ex,text depth=.25ex}}
  }
  \captionsetup[subfigure]{labelformat=empty,justification=justified,singlelinecheck=false}
  \pgfdeclarelayer{background}
  \pgfdeclarelayer{foreground}
  \pgfsetlayers{background,main,foreground}
 \begin{subfigure}{\textwidth}
   \centering
   \vspace{2.5ex}
   \caption{Dinaturality 1:}
   \vspace{.1ex}
   \raisebox{-.5\height}{
   \begin{tikzpicture}[
     point/.style={coordinate},>=stealth',thick,draw=black!50,
     tip/.style={->,shorten >=0.007pt},every join/.style={rounded corners},
     hv path/.style={to path={-| (\tikztotarget)}},
     vh path/.style={to path={|- (\tikztotarget)}},
     text height=1.5ex,text depth=.25ex %
     ]
     \node [nonterminal] (f) {$g$};
     \node [nonterminal] (g) at ($(f.east)+(1.5,-0.15)$) {$h$};
     \draw [<-] (f.west) -- +(-1,0) node [midway,above] {$X$};
     \path [<-,draw] (f.west)++(-0.5,0) -- ++(0,-0.95) -| ($(g.east)+(0.85,-0.15)$) node [near end,right] {$X$} -- ++(-0.85,0);
     \draw [->] (f.east)++(0,0.15) -- ++(0.65,0) node [midway,above] {$Y$} -- ++(0,0.4) -- ++(2.5,0);
     \draw [p,->] (f.east)++(0,-0.15) -- (g.west) node [midway,below] {$Z$};
     \draw [->] (g.east)++(0,0.15) -- ++(0.25,0) -- ++(0,0.55) node [midway,right] {$Y$};

     \begin{pgfonlayer}{background}
       \draw [iterated] ($(f.north west)+(-0.25,0.5)$) rectangle ($(g.south east)+(0.6,-0.25)$);
     \end{pgfonlayer}
   \end{tikzpicture}
   }
   ~~=~~
   \raisebox{-.5\height}{
   \begin{tikzpicture}[
     point/.style={coordinate},>=stealth',thick,draw=black!50,
     tip/.style={->,shorten >=0.007pt},every join/.style={rounded corners},
     hv path/.style={to path={-| (\tikztotarget)}},
     vh path/.style={to path={|- (\tikztotarget)}},
     text height=1.5ex,text depth=.25ex %
     ]
     \node [nonterminal] (h) at ($(f.west)+(-1.2,0.15)$) {$g$};
     \node [nonterminal] (f) {$h$};
     \node [nonterminal] (g) at ($(f.east)+(1.5,-0.15)$) {$g$};
     \draw [<-] (h.west) -- ++(-0.5,0) node [midway,above] {$X$};
     \draw [->] (h.east)++(0,0.15) -- ++(0.4,0) -- ++(0,0.75) node [midway,left] {$Y$} -- ++(3.75,0) -- ++(0,-0.5);
     \draw [o,<-] (f.west) -- ($(h.east)+(0,-0.15)$) node [pos=0.8,below] {$Z$};
     \path [o,<-,draw] (f.west)++(-0.5,0) -- ++(0,-0.95) -| ($(g.east)+(0.85,-0.15)$) node [near end,right] {$Z$} -- ++(-0.85,0);
     \draw [->] (f.east)++(0,0.15) -- ++(0.65,0) node [midway,above] {$Y$} -- ++(0,0.4) -- ++(2.5,0);
     \draw [->] (f.east)++(0,-0.15) -- (g.west) node [midway,below] {$X$};
     \draw [->] (g.east)++(0,0.15) -- ++(0.25,0) -- ++(0,0.55) node [midway,right] {$Y$};

     \begin{pgfonlayer}{background}
       \draw [iterated] ($(f.north west)+(-0.25,0.5)$) rectangle ($(g.south east)+(0.6,-0.25)$);
     \end{pgfonlayer}
   \end{tikzpicture}
   }
 \end{subfigure}
 \begin{subfigure}{\textwidth}
   \centering
   \vspace{2.5ex}
   \caption{Dinaturality 2:}
   \vspace{.1ex}
   \raisebox{-.5\height}{
   \begin{tikzpicture}[
     point/.style={coordinate},>=stealth',thick,draw=black!50,
     tip/.style={->,shorten >=0.007pt},every join/.style={rounded corners},
     hv path/.style={to path={-| (\tikztotarget)}},
     vh path/.style={to path={|- (\tikztotarget)}},
     text height=1.5ex,text depth=.25ex %
     ]
     \node [nonterminal] (f) {$g$};
     \node [nonterminal] (g) at ($(f.east)+(1.5,-0.15)$) {$h$};
     \draw [<-] (f.west) -- +(-1,0) node [midway,above] {$X$};
     \path [o,<-,draw] (f.west)++(-0.5,0) -- ++(0,-0.95) -| ($(g.east)+(0.85,-0.15)$) node [near end,right] {$X$} -- ++(-0.85,0);
     \draw [->] (f.east)++(0,0.15) -- ++(0.65,0) node [midway,above] {$Y$} -- ++(0,0.4) -- ++(2.5,0);
     \draw [->] (f.east)++(0,-0.15) -- (g.west) node [midway,below] {$Z$};
     \draw [->] (g.east)++(0,0.15) -- ++(0.25,0) -- ++(0,0.55) node [midway,right] {$Y$};

     \begin{pgfonlayer}{background}
       \draw [iterated] ($(f.north west)+(-0.25,0.5)$) rectangle ($(g.south east)+(0.6,-0.25)$);
     \end{pgfonlayer}
   \end{tikzpicture}
   }
   ~~=~~
   \raisebox{-.5\height}{
   \begin{tikzpicture}[
     point/.style={coordinate},>=stealth',thick,draw=black!50,
     tip/.style={->,shorten >=0.007pt},every join/.style={rounded corners},
     hv path/.style={to path={-| (\tikztotarget)}},
     vh path/.style={to path={|- (\tikztotarget)}},
     text height=1.5ex,text depth=.25ex %
     ]
     \node [nonterminal] (h) at ($(f.west)+(-1.2,0.15)$) {$g$};
     \node [nonterminal] (f) {$h$};
     \node [nonterminal] (g) at ($(f.east)+(1.5,-0.15)$) {$g$};
     \draw [<-] (h.west) -- ++(-0.5,0) node [midway,above] {$X$};
     \draw [->] (h.east)++(0,0.15) -- ++(0.4,0) -- ++(0,0.75) node [midway,left] {$Y$} -- ++(3.75,0) -- ++(0,-0.5);
     \draw [<-] (f.west) -- ($(h.east)+(0,-0.15)$) node [pos=0.8,below] {$Z$};
     \path [<-,draw] (f.west)++(-0.5,0) -- ++(0,-0.95) -| ($(g.east)+(0.85,-0.15)$) node [near end,right] {$Z$} -- ++(-0.85,0);
     \draw [->] (f.east)++(0,0.15) -- ++(0.65,0) node [midway,above] {$Y$} -- ++(0,0.4) -- ++(2.5,0);
     \draw [p,->] (f.east)++(0,-0.15) -- (g.west) node [midway,below] {$X$};
     \draw [->] (g.east)++(0,0.15) -- ++(0.25,0) -- ++(0,0.55) node [midway,right] {$Y$};

     \begin{pgfonlayer}{background}
       \draw [iterated] ($(f.north west)+(-0.25,0.5)$) rectangle ($(g.south east)+(0.6,-0.25)$);
     \end{pgfonlayer}
   \end{tikzpicture}
   }
 \end{subfigure}
 \begin{subfigure}{\textwidth}
   \centering
   \vspace{2.5ex}
   \caption{Bekić identity:}
   \vspace{.1ex}
  \raisebox{-.4\height}{
  \begin{tikzpicture}[
    point/.style={coordinate},>=stealth',thick,draw=black!50,
    tip/.style={->,shorten >=0.007pt},every join/.style={rounded corners},
    hv path/.style={to path={-| (\tikztotarget)}},
    vh path/.style={to path={|- (\tikztotarget)}},
    text height=1.5ex,text depth=.25ex %
    ]
    \node [nonterminal,minimum height=8mm] (g) {$g$};
    \node [nonterminal,minimum height=8mm] (f) at ($(g.east)+(1.75,0)$) {$f$};
    \node [nonterminal,minimum height=8mm] (g2) at ($(f.east)+(1.5,-0.85)$) {$g$};

    \path [draw,<-] (g.west) -- +(-1.25,0) node [pos=0.35,above] {$X$};
    \path [name path=ygdirect,->,draw] (g.west)++(-1.25,0.85) -| ($(g.east)+(0.5,0)$) node [pos=0.18,above] {$Y$};

    \path [p,draw,->] (g.east) -- (f.west);  
    \path [name path=ygexit,draw] ($(g.east)+(0,0.25)$) -- ($(f.west)+(-0.5,0.25)$);
    \path [draw,->] ($(f.west)+(-0.5,0.25)$) -- ($(f.west)+(-0.5,0.85)$) -| ($(g2.east)+(1.15,0.25)$);
    \path [p,draw,->] ($(g.east)+(0,-0.25)$) -| ($(g.east)+(0.5,-0.85)$) -| ($(g.west)+(-0.5,0)$);

    \path [name path=zg2exit,draw,->] ($(g2.east)+(0,0.25)$) -- ++(1.85,0) node [pos=0.85,above] {$Z$};
    \path [p,draw,->] ($(g2.east)+(0,0)$) -| ($(g2.east)+(1.15,-1.05)$) -| ($(f.west)+(-0.5,0)$);
    \path [p,draw,->] ($(g2.east)+(0,-0.25)$) -| ($(g2.east)+(0.5,-0.77)$) -| ($(g2.west)+(-0.5,0)$);

    \path [draw,->] ($(f.east)+(0,0.25)$) -| ($(g2.east)+(0.75,0.25)$);
    \path [p,name path=yfdirect,draw,->] ($(f.east)+(0,0)$) -| ($(g2.east)+(0.5,0)$);
    \path [p,draw,->] ($(f.east)+(0,-0.25)$) -| ($(g2.west)+(-0.85,0)$) -- ($(g2.west)$);

    \begin{pgfonlayer}{background}
      \draw [iterated] ($(g.north west)+(-0.25,.25)$) rectangle ($(g.south east)+(0.25,-0.25)$);
      \draw [iterated] ($(f.north west)+(-0.25,.25)$) rectangle ($(g2.south east)+(0.95,-0.45)$);
      \draw [iterated,fill=green!50] ($(g2.north west)+(-0.25,.25)$) rectangle ($(g2.south east)+(0.25,-0.25)$);
    \end{pgfonlayer}
  \end{tikzpicture}}
  ~~=~~
  \raisebox{-.4\height}{
  \begin{tikzpicture}[
    point/.style={coordinate},>=stealth',thick,draw=black!50,
    tip/.style={->,shorten >=0.007pt},every join/.style={rounded corners},
    hv path/.style={to path={-| (\tikztotarget)}},
    vh path/.style={to path={|- (\tikztotarget)}},
    text height=1.5ex,text depth=.25ex %
    ]
    \node [nonterminal,minimum height=8mm] (f) {$f$};
    \node [nonterminal,minimum height=8mm] (g) at ($(f.south)+(0,-1)$) {$g$};

    \path [draw,->] ($(f.west)+(-1.25,0)$) -- (f.west) node [pos=0.25,above] {$Y$};
    \path [draw,->] ($(g.west)+(-1.25,0)$) -- (g.west) node [pos=0.25,below] {$X$};

    \path [draw,->] ($(f.east)+(0,0.25)$) -| ($(g.east)+(0.25,0.25)$);

    \path [draw,->] ($(g.east)+(0,0.25)$) -- ++(1.65,0) node [above,pos=0.85] {$Z$};
    \path [p,draw,->] (f.east) -| ($(g.east)+(0.45,0)$);
    \path [p,draw,->] ($(f.east)+(0,-0.25)$) -| ($(g.east)+(0.65,-0.25)$);
    \path [p,draw,->] (g.east) -- ++(1,0) -- ++(0,2.3) -| ($(f.west)+(-0.5,0)$);
    \path [p,draw,->] ($(g.east)+(0,-0.25)$) -- ++(1,0) -- ++(0,-0.6) -| ($(g.west)+(-0.5,0)$);
    \begin{pgfonlayer}{background}
      \draw [iterated] ($(f.north west)+(-0.25,.25)$) rectangle ($(g.south east)+(0.80,-0.25)$);
    \end{pgfonlayer}
  \end{tikzpicture}}
\end{subfigure}  
 \caption{Derivable laws of iteration.}
\label{fig:dinat-bekic}
\end{figure}

%
%
%
%

 Besides the fixpoint identity we are interested in
natural guarded versions of the classical properties of the iteration
operator, which we refer to as the \emph{iteration
  laws}~\cite{Elgot75,BloomEsik93,SimpsonPlotkin00}:
\begin{itemize}
  \item\emph{naturality:} $g^{\klstar} f^{\istar} = ([(T\inl) \comp g, \eta\inr]^{\klstar} \comp f)^{\istar}$ for $f\colon X\to_2 T(Y+X)$, $g\colon  Y \to TZ$;
  \item\emph{codiagonal:} $(T[\id,\inr] \comp f)^{\istar} = f^{\istar\istar}$ for  $f\colon  X \to_{12,2} T((Y + X) + X)$;
  \item\emph{uniformity:} $f \comp h = T(\id+ h) \comp g$ implies
	$f^{\istar} \comp h = g^{\istar}$ for $f\colon  X \to_2 T(Y + X)$, $g\colon  Z \to_2 T(Y + Z)$ and
	$h\colon  Z \to X$.
\end{itemize}
Remarkably, this list does not include the 
well-known \emph{dinaturality law}, as is turns out to be derivable 
(cf.~\cite{GoncharovEA18,EsikGoncharov16}). We prove this further below.
The above axioms are summarized in graphical form in Figure~\ref{fig:ax},
and then become quite intuitive. 
We indicate the scope of the
iteration operator by a shaded box and guardedness by bullets at the
outputs of a morphism. Blue boxes indicate morphisms of the base category $\BC$,
to contrast orange boxes referring to Kleisli morphisms.

A guarded pre-iterative monad is called a \emph{complete Elgot monad}
if it is totally guarded and satisfies all iteration laws. In the
sequel we shorten `complete Elgot monads' to `Elgot monads' (to be
distinguished from Elgot monads in the sense
of~\cite{AdamekMiliusEtAl10}, which have solutions only for morphisms
with finitely presentable domain).

In general, the fact that the iteration laws are correctly formulated
relies on the axioms of guardedness.
For example, in 
the codiagonal axiom, this follows
by~\textbf{(cmp)} from the assumption
$f\colon  X \to_{12,2} T((Y + X) + X)$ that $T[\id,\inr] \comp f$ is
$\inj_2$-guarded, and by Proposition~\ref{prop:guard_preserve} that
$f^\istar$ is $\inj_2$-guarded. Indeed, the axioms for guarded monads
are designed precisely to enable the formulation of the iteration
laws.

We show next that for guarded \emph{iterative} monads, all iteration
laws are automatic. In preparation, we prove the aforementioned fact that 
dinaturality follows from the other axioms (thus generalizing corresponding recent 
observations on iteration theories~\cite{GoncharovEA18,EsikGoncharov16}).
Additionally, we show that the well-known Beki\'{c} identity is derivable 
too.
\begin{prop}\label{prop:deriv_din}
  Any guarded pre-iterative monad satisfying naturality, codiagonal
  and uniformity also satisfies 
\begin{itemize}
  \item\emph{dinaturality:} $([\eta \inl, h]^{\klstar} \comp g)^{\istar} = [\eta, ([\eta \inl,
	g]^{\klstar} \comp h)^{\istar}]^{\klstar} \comp g$ for  $g\colon  X \to_2 T(Y + Z)$ and $h\colon Z\to T(Y+X)$ or 
$g\colon  X \to T(Y + Z)$ and $h\colon Z\to_2 T(Y+X)$;
  \item\emph{Beki\'{c} identity:} $\left(T[\id+\inl,\inr\comp\inr]\comp [f, g]\right)^\istar = [h^\istar, [\eta,h^\istar]^\klstar g^\istar] $ with $f\colon X\to_{12,2} T((Y+X)+Z)$, $g\colon Z\to_{12,2} T((Y+X)+Z)$, and $h=
[\eta,g^\istar]^\klstar f\colon X\to_2 T(Y+X)$.
\end{itemize}
\end{prop}
\noindent In axiomatizations of total iteration, the Beki\'{c}
identity is sometimes taken to replace codiagonal and
dinaturality~\cite[Section~6.8]{BloomEsik93}~\cite{AdamekMiliusEtAl10,GoncharovEA18}. Both
dinaturality and the Beki\'{c} identity are again depicted graphically
in Figure~\ref{fig:dinat-bekic}.
The two versions of the dinaturality axiom
correspond to the alternative sets of guardedness assumptions in its formulation;
 basically, we need to distinguish cases on whether
the loop over~$g$ and~$h$ is guarded at~$g$ or at~$h$.

\begin{proof}
Following~\cite{EsikGoncharov16}, we consider a specific instance of
uniformity:
\begin{align}\label{eq:unif_inst}
f^\istar = ([T(\id+\inl) f, h]\gray{: X+Z\to_2 T(Y+(X+Z))})^\istar\inl
\end{align}
where $f\colon X\to_2 T(Y+X)$ and $h\colon Z\to_2 T(Y+(X+Z))$, and prove the following 
instance of the Beki\'{c} identity:
\begin{align}\label{eq:bekic1}
(T(\id + \inl)\comp [f, g]\gray{\colon X+Z\to T(Y+(X+Z))})^\istar = [f^\istar,  [\eta,f^\istar]^\klstar g] 
\end{align}
where $f\colon X\to_2 T(Y+X)$ and $g\colon Z\to_2 T(Y+X)$. Indeed, on the one hand, by~\eqref{eq:unif_inst},
\begin{align}\label{eq:bekic2}
\left(T(\id + \inl)\comp [f, g]\right)^\istar\inl = f^\istar, 
\end{align}
and on the other hand
\begin{flalign*}
&&(T(\id& + \inl)\comp [f, g] )^\istar\inr\\ 
&&=&\; [\eta,\left(T(\id + \inl)\comp[f, g]\right)^\istar]^\klstar \comp T(\id+\inl) [f, g]\inr&\by{fixpoint}\\
&&=&\; [\eta,\left(T(\id + \inl)\comp[f, g]\right)^\istar]^\klstar \comp T(\id+\inl) g\\
&&=&\; [\eta,\left(T(\id + \inl)\comp[f, g]\right)^\istar\inl]^\klstar\comp g\\*
&&=&\; [\eta, f^\istar]^\klstar\comp g.&\by{\eqref{eq:bekic2}}
\end{flalign*}
As the result we obtain~\eqref{eq:bekic1}. Analogously, we prove another 
instance of the Beki\'{c} identity, namely
\begin{align}\label{eq:bekic3}
(T(\id + \inr)\comp [f,g]\gray{\colon X+Z\to T(Y+(X+Z))})^\istar = [[\eta,g^\istar]^\klstar f,g^\istar] 
\end{align}
with $f\colon X\to_2 T(Y+Z)$ and $g\colon Z\to_2 T(Y+Z)$. We proceed to show that 
under the other axioms, these two instances imply the full Beki\'{c} identity
\begin{align}
(T[\id+\inl,\inr\comp\inr]\comp[f, g]\gray{\colon X+Z\to T(Y+(X+Z))})^\istar= [h^\istar, [\eta,h^\istar]^\klstar g^\istar] \label{eq:bekic4}
\end{align}
where $f\colon X\to_{12,2} T((Y+X)+Z)$ and $g\colon Z\to_{12,2} T((Y+X)+Z)$ and 
$h= [\eta,g^\istar]^\klstar f\colon X\to T(Y+X)$. Let us argue briefly that~$h$ is 
$\inr$-guarded. Note that the assumption for~$f$ by~\textbf{(iso)} implies that $f'=T[\id+\inl,\inr]\comp f\colon X\to T(Y+(X+Z))$
is $\inr$-guarded and therefore $h= [\eta\inl, [\eta\inl, g^\istar]]^\klstar f'$
is $\inr$-guarded by \textbf{(cmp)}.

Now, the proof of~\eqref{eq:bekic4} runs as follows:
\begin{flalign*}
&&(T[\id+&\inl,\inr\comp\inr]\comp[f, g])^\istar\\
&&=&\;\bigl(T[\id,\inr]\comp T((\id+\inl)+\inr)\comp [f,g]\bigr)^\istar\\
&&=&\;(T((\id+\inl)+\inr)\comp [f, g])^{\istar\istar}&\by{codiagonal}\\
&&=&\;(T(\id+\inr)\comp [T((\id+\inl)+\id)\comp f, T((\id+\inl)+\id)\comp g])^{\istar\istar}\\
&&=&\;[[\eta,(T((\id+\inl)+\id)\comp g)^\istar]^\klstar T((\id+\inl)+\id)\comp  f,\\
&&&\qquad (T((\id+\inl)+\id)\comp g)^\istar]^\istar&\by{\eqref{eq:bekic3}}\\
&&=&\;[[\eta(\id+\inl),(T((\id+\inl)+\id)\comp g)^\istar]^\klstar  f,\\
&&&\qquad(T((\id+\inl)+\id)\comp g)^\istar]^\istar&\by{naturality}\\
&&=&\;[T(\id+\inl)\comp [\eta,g^\istar]^\klstar\comp f, T(\id+\inl)\comp g^\istar]^\istar\\
&&=&\;(T(\id+\inl)\comp [h, g^\istar])^\istar\\
&&=&\; [h^\istar,  [\eta,h^\istar]^\klstar g^\istar].&\by{\eqref{eq:bekic1}}  
\end{flalign*}
Finally, let us derive dinaturality from~\eqref{eq:bekic4}. Suppose that
$g\colon X \to_2 T(Y + Z)$ and $h\colon Z\to T(Y+X)$ satisfy either guardedness premise of the 
dinaturality axiom and consider the following instance of~\eqref{eq:bekic4} with
$f$ replaced by $T(\inl+\id)\comp g$ and $g$ replaced by $(T\inl)\comp h$ 
(note that by the fixpoint identity, $((T\inl)\comp h)^\istar = h$):
\begin{align}\label{eq:bekic-din}
[T(\id+\inr)\comp g, T(\id+\inl)\comp h]^\istar = [([\eta\inl,h]^\klstar g)^\istar, [\eta,([\eta\inl,h]^\klstar g)^\istar]^\klstar h].   
\end{align}
Let $\gamma_{X,Y}\colon Y+X\to Y+X$ be the obvious symmetry transformation and note the 
following simple consequence of uniformity:
\begin{align}\label{eq:unif-din}
[T(\id+\inr)\comp g, T(\id+\inl)\comp h]^\istar\gamma = [T(\id+\inr)\comp h, T(\id+\inl)\comp g]^\istar.   
\end{align}
By combining~\eqref{eq:unif-din},~\eqref{eq:bekic-din} and the symmetric form of
the latter (with $h$ and $g$ switched), we obtain:
\begin{align*}
[([\eta\inl,h]^\klstar g)^\istar, [\eta,([\eta\inl,h]^\klstar g)^\istar]^\klstar h] =
[[\eta,([\eta\inl,g]^\klstar h)^\istar]^\klstar g,([\eta\inl,g]^\klstar h)^\istar].   
\end{align*}
Dinaturality is now obtained by composing both sides with  
$\inl\colon X\to X+Z$.
\end{proof}

\noindent The proof of the following result runs in accordance with
the original ideas of Elgot for iterative theories~\cite{Elgot75},
except that, by Proposition~\ref{prop:deriv_din}, dinaturality is now
replaced with uniformity.
\begin{theorem}\label{thm:compl_ax}
Every guarded iterative monad validates naturality, dinaturality, codiagonal, and uniformity.
\end{theorem}
\begin{proof}
By Proposition~\ref{prop:deriv_din} we only need to verify naturality, codiagonal and uniformity.
\begin{citemize}
  \item\emph{Naturality.} Let $f\colon X\to_2 T(Y+X)$ and $g\colon  Y \to TZ$. Then 
\begin{align*}
g^{\klstar} f^{\istar} 
=&\; g^{\klstar} [\eta,f^{\istar}]^\klstar f\\
=&\; [g^{\klstar}\eta,g^{\klstar} f^{\istar}]^\klstar f\\
=&\; [g,g^{\klstar} f^{\istar}]^\klstar f\\
=&\; [\eta,g^{\klstar} f^{\istar}]^\klstar [(T\inl) \comp g, \eta\inr]^{\klstar} f.
\end{align*}
Since the same equation uniquely characterizes $([(T\inl) \comp g, \eta\inr]^{\klstar} \comp f)^{\istar}$,
the latter is equal to~$g^{\klstar} f^{\istar} $.
  \item\emph{Codiagonal.} Let $f\colon  X \to_{12,2} T((Y + X) + X)$. Then 
\begin{align*}
f^{\istar\istar} 
=&\; [\eta,f^{\istar\istar}]^\klstar f^{\istar}\\
=&\; [\eta,f^{\istar\istar}]^\klstar [\eta,f^{\istar}]^\klstar f\\
=&\; [[\eta,f^{\istar\istar}], [\eta,f^{\istar\istar}]^\klstar f^{\istar}]^\klstar f\\
=&\; [[\eta,f^{\istar\istar}], f^{\istar\istar}]^\klstar f\\
=&\; [\eta, f^{\istar\istar}]^\klstar T[\inr,\id]\comp f.
\end{align*}
Therefore $f^{\istar\istar}$ satisfies the fixpoint identity for $(T[\inr,\id]\comp f)^\istar$,
and thus $f^{\istar\istar}  = (T[\inr,\id]\comp f)^\istar$.
  \item\emph{Uniformity.} Suppose that $f \comp h = T(\id+ h) \comp g$ for some 
$f\colon  X \to_2 T(Y + X)$, $g\colon  Z \to_2 T(Y + Z)$ and $h\colon  Z \to X$. Then 
\begin{align*}
f^\istar h = [\eta,f^\istar]^\klstar f\comp h = [\eta,f^\istar]^\klstar T(\id+ h) \comp g = [\eta, f^\istar h]^\klstar g,
\end{align*}
that is, $f^\istar h$ satisfies the fixpoint equation for $g^\istar$. Hence $g^\istar = f^\istar h$.\qed
\end{citemize}
\noqed\end{proof}
\noindent We now proceed to introduce key properties of morphisms of
guarded monads that allow for transferring pre-iterativity and
the iteration laws, respectively.
\begin{definition}[Guarded retraction]
Let $\BBT$ and $\BBS$ be guarded monads. We call a monad morphism
 $\rho\colon\BBT\to\BBS$ a \emph{guarded retraction} if there is a family of morphisms 
$(\upsilon_X\colon SX\to TX)_{X\in|\BC|}$  (not necessarily natural in $X$!) 
such that 
\begin{enumerate}
 \item for every $f\colon X\to_{\sigma} SY$, we have $\upsilon_Y f\colon X\to_{\sigma} TY$, and
 \item $\rho_X\upsilon_X=\id$ for all $X\in |\BC|$.
\end{enumerate}
\end{definition}

\begin{theorem}\label{thm:ext_iter}
  Let $\rho\colon\BBT\to\BBS$ be a guarded retraction, witnessed by
  $\upsilon\colon\BBS\to\BBT$, and suppose that $(\BBT,\argument^\istar)$
  is guarded pre-iterative. Then $\BBS$ is guarded pre-iterative with
  the iteration operator $(\argument)^\iistar$ given by
  $f^\iistar = \rho\comp (\upsilon f)^\istar$.
\end{theorem}
\begin{proof}
Since $\BBT$ satisfies the fixpoint identity, $[\eta,(\upsilon f)^\istar]^\klstar\upsilon f  = (\upsilon f)^\istar$ and therefore,
\begin{flalign*}
&& f^\iistar   =&\; \rho\comp (\upsilon f)^\istar\\
&& =&\;\rho\comp [\eta,(\upsilon f)^\istar]^\klstar\upsilon f\\
&& =&\; [\rho\eta, \rho(\upsilon f)^\istar]^\klstar\rho\upsilon f&\by{$\rho$ is a monad morphism}\\
&& =&\; [\eta, f^\iistar]^\klstar f. &\by{$\rho\upsilon=\id$}
\end{flalign*}
\end{proof}
\begin{definition}[Iteration congruence]
Let $\BBT$ be a guarded pre-iterative monad and let $\BBS$ be a monad. We call a 
monad morphism $\rho\colon\BBT\to\BBS$ an \emph{iteration congruence} if for every pair 
of morphisms $f,g\colon X\to_2 T(Y+X)$, 
\begin{align}\label{eq:icong}
\rho f = \rho g\text{\qquad implies\qquad} \rho f^\istar = \rho g^\istar.
\end{align}
If $\rho$ is moreover a guarded retraction, we call $\rho$ an
\emph{iteration-congruent retraction}.
\end{definition}
\begin{theorem}\label{thm:ext_ax}
  Under the premises of Theorem~\ref{thm:ext_iter}, assume
  moreover that $\rho$ is an iteration-congruent retraction. Then any
  property out of naturality, dinaturality, codiagonal, and uniformity
  that is satisfied by $\BBT$ is also satisfied by~$\BBS$.
\end{theorem}
\begin{proof}
  The crucial observation is that under our
  assumptions,~\eqref{eq:icong} is equivalent to the condition that 
for all $f\colon X\to_2 T(Y+X)$,
  \begin{align}\label{eq:retraction-iteration}
    \rho\comp (\upsilon\rho\comp f)^\istar = \rho f^\istar.
  \end{align}
Indeed,~\eqref{eq:icong} $\implies$ \eqref{eq:retraction-iteration},
for $\rho\upsilon\rho\comp f=\rho\comp f$ and therefore
$\rho (\upsilon\rho\comp f)^\istar = \rho\comp f^\istar$ and
conversely, assuming~\eqref{eq:retraction-iteration} both for $f$ and
for $g$, and $\rho f = \rho g$, we obtain that
$\rho f^\istar = \rho (\upsilon\rho\comp f)^\istar= \rho
(\upsilon\rho\comp g)^\istar=\rho g^\istar$.
The proof of transfer of the respective properties then proceeds as follows.
\begin{citemize}
 \item\emph{Naturality:} 
\begin{flalign*}
&&g^\klstar f^\iistar
=&\;g^\klstar\rho\comp (\upsilon\comp f)^\istar\\
&&=&\;(\rho\upsilon\comp g)^\klstar\comp\rho\comp (\upsilon\comp f)^\istar&\by{$\rho\upsilon=\id$}\\
&&=&\;\rho\comp (\upsilon\comp g)^\klstar\comp(\upsilon\comp f)^\istar&\by{$\rho$ is a monad~morphism}\\
&&=&\;\rho\comp ([(T\inl)\comp \upsilon\comp g,\eta\inr]^\klstar	\upsilon\comp f)^\istar&\by{naturality for $(\argument)^\istar$}\\
&&=&\;\rho\comp (\upsilon\rho\comp[(T\inl)\comp \upsilon\comp g,\eta\inr]^\klstar\upsilon\comp f)^\istar&\by{\eqref{eq:retraction-iteration}}\\
&&=&\;\rho\comp (\upsilon[\rho\comp(T\inl)\comp \upsilon\comp g,\eta\inr]^\klstar\rho\comp\upsilon\comp f)^\istar&\by{$\rho$ is a monad~morphism}\\
&&=&\;\rho\comp (\upsilon[(S\inl)\comp\comp g,\eta\inr]^\klstar\comp f)^\istar&\by{$\rho\upsilon=\id$}\\
&&=&\;([(S\inl)\comp\comp g,\eta\inr]^\klstar f)^\iistar.
\end{flalign*}   
\item \emph{Dinaturality:} First observe that it follows from the fact that $\rho$
is a monad morphism and $\rho\upsilon=\id$ that 
$\rho\upsilon[\eta \inl, h]^{\klstar} \comp g
=\rho[\eta \inl, \upsilon h]^{\klstar}\upsilon g$, and therefore, by~\eqref{eq:icong},
that
\begin{align}\label{eq:dinat_rho_ups}
\rho(\upsilon[\eta \inl, h]^{\klstar} \comp g)^\istar
=\rho([\eta \inl, \upsilon h]^{\klstar} \upsilon g)^\istar.
\end{align}
Then we obtain the goal as follows:
\begin{flalign*}
&&([\eta \inl, h]^{\klstar} \comp g)^{\iistar} 
=&\; \rho(\upsilon[\eta\inl, h]^{\klstar} \comp g)^{\istar}&\\
&&=&\; \rho([\eta\inl, \upsilon h]^{\klstar} \upsilon g)^{\istar}&\by{\eqref{eq:dinat_rho_ups}}\\
&&=&\; \rho [\eta, ([\eta \inl, \upsilon g]^{\klstar} \comp\upsilon  h)^{\istar}]^{\klstar} \comp \upsilon  g&\by{dinaturality for~$(\argument)^\istar$}\\
&&=&\; [\eta, \rho ([\eta \inl, \upsilon g]^{\klstar} \comp\upsilon  h)^{\istar}]^{\klstar} \rho \upsilon  g&\by{$\rho$ is a monad~morphism}\\
&&=&\; [\eta, \rho ([\eta \inl, \upsilon g]^{\klstar} \comp\upsilon  h)^{\istar}]^{\klstar}  g&\by{since $\rho\upsilon=\id$}\\
&&=&\; [\eta, \rho (\upsilon [\eta \inl, g]^{\klstar} h)^{\istar}]^{\klstar}  g&\by{analogous~to~\eqref{eq:dinat_rho_ups}}\\
&&=&\; [\eta, ([\eta \inl, g]^{\klstar} h)^{\iistar}]^{\klstar}  g.
\end{flalign*}
\item \emph{Codiagonal:} 
\begin{flalign*}
&&(f^\iistar)^\iistar =&\; \rho\comp (\upsilon\rho (\upsilon\comp f)^\istar)^\istar\\
&&=&\;\rho\comp ((\upsilon\comp f)^\istar)^\istar&\by{\eqref{eq:retraction-iteration}}\\
&&=&\;\rho\comp (T[\id,\inr]\comp\upsilon\comp f)^\istar&\by{codiagonal for $(\argument)^\istar$}\\
&&=&\;\rho\comp (\upsilon\rho\comp T[\id,\inr]\comp\upsilon\comp f)^\istar&\by{\eqref{eq:retraction-iteration}}\\
&&=&\;\rho\comp (\upsilon\comp S[\id,\inr]\comp\comp f)^\istar&\by{$\rho\upsilon=\id$}\\
&&=&\;(S[\id,\inr]\comp f)^\iistar.
\end{flalign*}
\item\emph{Uniformity:} Suppose that $f\comp h = S(\id+h)\comp g$. Then 
\begin{flalign*}
&&\upsilon\comp f h =&\; 
\upsilon\comp S(\id+h)\comp g&\\
&&=&\;\upsilon\comp S(\id+h)\comp\rho\upsilon\comp g\\
&&=&\;T(\id+h)\comp\upsilon\comp g
\end{flalign*}
and therefore $(\upsilon\comp f)^\istar\comp h=(\upsilon\comp g)^\istar$. This implies $f^\iistar\comp h=g^\iistar$ by definition.\qed
\end{citemize}
\noqed\end{proof}
\noindent Recall from the introduction that a monad $\BBS$ is
\emph{iteratable} if its coinductive resumption transform
\/~$\BBS^\nu$ exists. We make $\BBS^\nu$ into a guarded monad by
applying Theorem~\ref{thm:ext_guard} to $\BBS$ as a vacuously guarded
monad; explicitly: $f\colon X\to S^\nu(Y+X)$ is guarded iff
\begin{equation*}
\out f = S(\inl+\id)\comp g\qquad\text{for some\qquad $g\colon X\to S(Y+S^\nu(Y+X))$}.
\end{equation*}
We are now set to prove our first main result, which states that every
iteratable Elgot monad can be obtained by quotienting a guarded
iterative monad; that is, every choice of solutions that obeys the
iteration laws arises by quotienting a more fine-grained model in
which solutions are uniquely determined:
\begin{theorem}
\label{thm:elgot_is_retract}
A totally guarded iteratable monad $\BBS$ is an Elgot monad iff there
is a guarded iterative monad $\BBT$ and an iteration-congruent
retraction $\rho\colon\BBT\to\BBS$. Specifically, every iteratable Elgot
monad $\BBS$ is an iteration-congruent retract of its coinductive
resumption transform \/~$\BBS^\nu$.
\end{theorem}
\begin{proof}
  `If' is immediate by Theorems~\ref{thm:ext_ax}
  and~\ref{thm:compl_ax}. We prove `only if', i.e.\ that
  $\BBS=(S,\eta,\argument^\klstar,\argument^\istar)$ is an
  iteration-congruent retract of
  $\BBS^\nu=(\nu\gamma.\,S(\argument+\gamma),\eta^\nu,\argument^\klstar,\argument^\iistar)$. We
  define $\upsilon_X=\tuo\eta\inr\tuo (S\inl)$ and
\begin{displaymath}
\rho_X = \bigl(S^\nu X\xrightarrow{\out} S(X+S^\nu X)\bigr)^\istar.
\end{displaymath}
Clearly, $\upsilon f$ is $\sigma$-guarded for every $f\colon X\to SY$ and $\upsilon$ is left inverse to $\rho$, for
\begin{flalign*}
&&\rho\upsilon 
=&\;\out^\istar\tuo\eta\inr\tuo (S\inl)\\ 
&&=&\; [\eta,\rho]^\klstar \out\tuo\eta\inr\tuo (S\inl)&\by{fixpoint for $(\argument)^\istar$}\\ 
&&=&\; \rho\tuo (S\inl)\\
&&=&\; [\eta,\rho]^\klstar \out\tuo (S\inl)&\by{fixpoint for $(\argument)^\istar$}\\ 	
&&=&\; \eta^\klstar\\
&&=&\;\id.
\end{flalign*}
It follows straightforwardly by naturality of $(\argument)^\istar$ that $\rho$
is a natural transformation. Note the following property of $\rho$: for any $h\colon X\to S(Y+X)$, 
$\out (\coit h) = S(\id+\coit h)\comp h$, and hence, by uniformity
\begin{align}\label{eq:rho_coit_is_iter}
\rho(\coit h) = h^\istar. 
\end{align}
Let us verify that $\rho$ is a monad morphism. For one thing 
\begin{align*}
\rho\eta^\nu = [\eta,\rho]^\klstar\out\eta^\nu= [\eta,\rho]^\klstar\eta\inl = \eta.
\end{align*}
Next, we have to check that $\rho f^\kklstar = (\rho f)^\klstar\rho$ for any $f\colon X\to S^\nu Y$.
Note that 
\begin{align*}
f^\kklstar = \coit\bigl([[\eta\inl,S(\id+S^\nu\inl)\out f],\eta\inr]^\klstar\out\bigr)\comp (S^\nu\inr). 
\end{align*}
Therefore
\begin{flalign*}
&&\rho f^\kklstar =&\; \rho \coit\bigl([[\eta\inl,S(\id+S^\nu\inl)\out f],\eta\inr]^\klstar\out\bigr)\comp (S^\nu\inr)\\
&&=&\; \bigl([[\eta\inl,S(\id+S^\nu\inl)\out f],\eta\inr]^\klstar\out\bigr)^\istar\comp (S^\nu\inr)&\by{\eqref{eq:rho_coit_is_iter}}\\
&&=&\; \bigl(S[\id,\inr]\comp [(S\inl)\comp [\eta\inl,S(\id+S^\nu\inl)\out f],\eta\inr]^\klstar\out\bigr)^\istar\comp (S^\nu\inr)\\
&&=&\; \bigl(([(S\inl)\comp [\eta\inl,S(\id+S^\nu\inl)\out f],\eta\inr]^\klstar\out)^\istar\bigr)^\istar\comp (S^\nu\inr)&\by{codiagonal}\\
&&=&\; ([\eta\inl,S(\id+S^\nu\inl)\out f]^\klstar \out^\istar)^\istar\comp (S^\nu\inr)&\by{naturality}\\
&&=&\; ([\eta\inl,S(\id+S^\nu\inl)\out f]^\klstar \rho)^\istar\comp (S^\nu\inr)\\
&&=&\; [\eta,([\eta\inl,\rho]^\klstar S(\id+S^\nu\inl)\out f)^\istar]^\klstar\rho\comp (T\inr)&\by{dinaturality}\\
&&=&\; [\eta,([\eta\inl,\rho]^\klstar S(\id+S^\nu\inl)\out f)^\istar]^\klstar (S\inr)\comp\rho\\
&&=&\; (([\eta\inl,\rho]^\klstar S(\id+S^\nu\inl)\out f)^\istar)^\klstar \comp\rho\\
&&=&\; (((S\inl)\comp [\eta,\rho]^\klstar \out f)^\istar)^\klstar \comp\rho\\
&&=&\; (((S\inl)\comp \rho f)^\istar)^\klstar \comp\rho\\
&&=&\; ([\eta, ((S\inl)\comp \rho f)^\istar]^\klstar (S\inl)\comp \rho f)^\klstar \comp\rho&\by{fixpoint}\\
&&=&\; (\rho f)^\klstar\rho.
\end{flalign*}
Finally, let us check that $\rho$ is an iteration
congruence. Let $f,g\colon X\to_2 S^\nu(Y+X)$, which means that there are $f',g'\colon X\to S(Y+S^\nu(X+Y))$
such that $\out f=S(\inl+\id) f'$ and $\out f=S(\inl+\id) g'$. Suppose that $\rho f=\rho g$, 
which amounts to
\begin{align}\label{eq:cong_simp}
[\eta\inl,\rho]^\klstar f' = [\eta\inl,\rho]^\klstar g',
\end{align}
for
\begin{align*}
\rho f = [\eta,\rho]^\klstar\out f = [\eta,\rho]^\klstar S(\inl+\id) f'= [\eta\inl,\rho]^\klstar f'
\end{align*}
and analogously for $g$. Our goal is to prove that
\begin{align*}
\rho f^\iistar = [\eta,([\eta\inl,\rho]^\klstar f')^\iistar]^\klstar\eta\inr, 
\end{align*}
from which $\rho f^\istar=\rho g^\istar$ will follow by the analogous formula for $\rho g^\istar$ and~\eqref{eq:cong_simp}.
Observe that
\begin{align*}
f^\iistar = (\coit h)\comp\comp\eta^\nu\inr
\end{align*}
where $h=[[\eta\inl,f'],\eta\inr]^\klstar\out$. Now
\begin{flalign*}
&&\rho f^\iistar =&\; \rho\comp (\coit h)\comp\eta^\nu\inr\\
&& =&\; h^\istar\comp\eta^\nu\inr&\by{\eqref{eq:rho_coit_is_iter}}\\
&& =&\; ([[\eta\inl,f'],\eta\inr]^\klstar\out)^\istar\comp\eta^\nu\inr&\\
&& =&\; (S[\id,\inr]\comp [[\eta\inl\inl, (S\inl)f'],\eta\inr]^\klstar\out)^\istar\comp\eta^\nu\inr&\\
&& =&\; (([[\eta\inl\inl, (S\inl)f'],\eta\inr]^\klstar\out)^\istar)^\istar\comp\eta^\nu\inr&\by{codiagonal}\\
&& =&\; (([(S\inl)[\eta\inl, f'],\eta\inr]^\klstar\out)^\istar)^\istar\comp\eta^\nu\inr&\\
&& =&\; ([\eta\inl, f']^\klstar\out^\istar)^\istar\comp\eta^\nu\inr&\by{naturality}\\
&& =&\; ([\eta\inl, f']^\klstar\rho)^\istar\comp\eta^\nu\inr&\\
&& =&\; [\eta,([\eta\inl,\rho]^\klstar f')^\istar]^\klstar\rho\comp\eta^\nu\inr&\by{dinaturality}\\
&& =&\; [\eta,([\eta\inl,\rho]^\klstar f')^\istar]^\klstar\eta^\nu\inr
\end{flalign*}
and we are done.
\end{proof}

\begin{example}[Finite trace semantics]\label{expl:trace}
Let us revisit Example~\ref{ex:transfer-example}~(4), with $A$ assumed to be finite 
throughout. Recall that 
$\mu\gamma.\,(X+1)+A\times\gamma\cong A^\star + A^\star\times X$ is a final $((X+1)+A\times\argument)$-coalgebra in 
the Kleisli category $\Set_{\CSet}$ of $\CSet$. Note that $\nu\gamma.\,\CSet(X+A\times\gamma)$
is a coalgebra of the same type in the same category, with
\begin{displaymath}
	\nu\gamma.\,\CSet(X+A\times\gamma) \xto{~\CSet(\inl+\id)\out\cup\,\sgl\inl\inr\bang~} \CSet((X+1)+A\times\nu\gamma.\,\CSet(X+A\times\gamma))
\end{displaymath}
as the structure morphism, where $\cup$ denotes pointwise union and
$\sgl$ is the map $x\mapsto\{x\}$, i.e.\ the unit of
$\CSet$. Intuitively, we thus add `non-termination', i.e.\ the element of
the right-hand summand $1$ in $X+1$, as a possible result to every
state (in the original view of Hasuo et al.~\cite{HasuoJacobsEtAl07},
this element instead represents acceptance, so the above definition
would correspond to converting a labelled transition system into an
automaton by making every state accepting). This yields a final
coalgebra map
$\xi_X\colon \nu\gamma.\,\CSet(X+A\times\gamma) \to
\CSet(\mu\gamma.\,{(X+1)}+A\times\gamma)$ characterized by the diagram
\begin{equation*}
\begin{tikzcd}[column sep = large,row sep = large]
\nu\gamma.\,\CSet(X+A\times\gamma)
	\ar[r,"(\inl+\id)\out\cup\,\sgl\inl\inr\bang"]
	\ar[d,"\xi_X"'] &[6em] 
(X+1)+A\times\nu\gamma.\,\CSet(X+A\times\gamma) 
	\ar[d,"X+1 + A\times \xi_X"]\\[2ex]
\mu\gamma.\,(X+1)+A\times\gamma
	\ar[ru, phantom, "(\text{in}~\Set_{\CSet})"]	
	\ar[r,"\inm^\mone"] &
(X+1) + A\times \mu\gamma.\,(X+1)+A\times\gamma
\end{tikzcd}
\end{equation*}
which amounts to the following corecursive definition of $\xi_X$:
\begin{displaymath}
	\xi_X(t) = \{\inm\inl\inl x\mid \inl x\in\out t\}\cup \{\inm\inl\inr\star\}\cup \{\inm\inr\brks{a,t''}\mid t''\in \xi_X(t'),\inr\brks{a,t'}\in\out t\}. 
\end{displaymath} 
The result of applying $\xi_X$ to a tree is the set of finite traces in it, which 
are finite sequences from $A^\star$ followed either by an element of $X$ 
(successfully terminating traces) or by the single inhabitant of $1$ 
(divergent traces). 
It is easy to see that~$\xi$ is a natural transformation; we show that
it is in fact a monad morphism. The domain of~$\xi$ is a generalized
coalgebraic resumption monad
${\nu\gamma.\,\CSet(\argument+A\times\gamma)}$ (on $\Set$) as
discussed in Example~\ref{ex:gen-proc}, while the codomain
$\CSet(\mu\gamma.\,(\argument +1)+A\times\gamma)$ is obtained by
sandwiching the monad $\nu\gamma.\,(\argument +1)+A\times\gamma$ (on
$\Set_{\CSet}$) between the adjoint pair
$F\dashv G\colon \Set_{\CSet}\to\Set$ generating the monad $\CSet$, and
therefore is also a monad (cf.\ Section~\ref{sec:prelim}). The
corresponding structure is defined as follows:
{\allowdisplaybreaks[0]
\begin{align*}
	\eta(x) =&\;\{\inm\inl\inl x\} \\
\bigl(f\colon X\to \CSet(\mu\gamma.\,(Y +1)+A\times\gamma)\bigr)^\klstar(p) =&\; \bigcup~\{(\iter \hat f)(t) \mid t\in p\}
\end{align*}
}
where $\iter\hat f\colon \mu\gamma.\,(X +1)+A\times\gamma\to \CSet(\mu\gamma.\,(Y +1)+A\times\gamma)$ is the initial
algebra morphism to the algebra $(\CSet(\mu\gamma.\,(Y +1)+A\times\gamma), \hat f)$ whose structure map
\begin{displaymath}
	\hat f\colon (X + 1) + A\times \CSet(\mu\gamma.\,(Y +1)+A\times\gamma)\to \CSet(\mu\gamma.\,(Y +1)+A\times\gamma)
\end{displaymath}
is as follows: $\hat f(\inl\inl x) = f(x)$, $\hat f(\inl\inr\star) = \{\inm\inl\inr\star\}$, 
$\hat f(\inr\brks{a,t}) = \{\inm\inr\brks{a,t}\}$.
This results in the following inductive definition of $\iter\hat f$:
\begin{align*}
	(\iter\hat f)\inm\inl\inl x =& f(x)\\ 
    (\iter\hat f)\inm\inl\inr\star =& \{\inm\inl\inr\star\}\\
    (\iter\hat f)\inm\inr\brks{a,t} =& \{\inm\inr\brks{a,t'}\mid t'\in (\iter\hat f)(t)\} 
\end{align*}
It is then easy to see that $\xi$ respects 
$\eta$. The fact that $\xi$ respects Kleisli lifting amounts to a rather technical 
verification of the fact that both $\xi f^\klstar$ and $(\xi f)^\klstar\xi$ satisfy 
the same corecursive definition and are thus equal:
\begin{flalign*}
&& \xi\comp f^\klstar(p) 
  = \;&\bigl\{\inm\inl\inl x\mid \inl x\in\out f^\klstar(p)\bigr\}\cup \bigl\{\inm\inl\inr\star\bigr\}&\\
&&  \;&\quad \cup\bigl\{\inr\brks{a,t''}\mid t''\in \xi(t'),\inr\brks{a,t'}\in\out f^\klstar(p)\bigr\}\\
&&  = \;&\bigl\{\inm\inl\inl x\mid \inl x\in \{\out f(x) \mid \inl x\in \out p\} \bigr\}\cup \{\inm\inl\inr\star\}\\
&&  \;&\quad\cup\bigl\{\inm\inl\inl x\mid \inl x\in \{\inr\brks{a,f^\klstar(p')} \mid \inr\brks{a,p'}\in\out p \} \bigr\}\\
&&  \;&\quad\cup \bigl\{\inr\brks{a,t''}\mid t''\in \xi(t'),\inr\brks{a,t'}\in \{\out f(x) \mid \inl x\in \out p\} \bigr\}\\
&&  \;&\quad\cup \bigl\{\inr\brks{a,t''}\mid t''\in \xi(t'),\inr\brks{a,t'}\in \{\inr\brks{a,f^\klstar(p')} \mid \inr\brks{a,p'}\in\out p \} \bigr\}\\
&&  = \;&\{\inm\inl\inl x\mid \inl x\in\out f(x), \inl x\in \out p\} \cup \{\inm\inl\inr\star\}\\
&&  \;&\quad\cup \{\inr\brks{a,t''}\mid t''\in \xi(t'),\inr\brks{a,t'}\in \out f(x), \inl x\in \out p\}\\
&&  \;&\quad\cup  \{\inr\brks{a,t''}\mid t''\in \xi(f^\klstar(p')), \inr\brks{a,p'}\in\out p \}\\[2ex]
&&  (\xi f)^\klstar(\xi(p)) 
    = \;& \bigcup\,\{(\iter \widehat{\xi\comp f})(t) \mid t\in \xi(p)\}\\
&&  = \;& \bigcup\,\bigl\{(\iter \widehat{\xi\comp f})(t) \mid t\in \{\inm\inl\inl x\mid \inl x\in\out p\}\bigr\}\\
&&  \;&\quad\cup \bigcup\,\bigl\{(\iter  \widehat{\xi\comp f})(t) \mid t\in \{\inm\inl\inr\star\}\bigr\}\\
&&  \;&\quad\cup \bigcup\,\bigl\{(\iter  \widehat{\xi\comp f})(t) \mid t\in \{\inm\inr\brks{a,t''}\mid t''\in \xi(t'),\inr\brks{a,t'}\in\out p\} \bigr\}\\
&&  = \;& \bigcup\,\{\xi f(x)\mid \inl x\in\out p\}\cup  \{\inm\inl\inr\star\}\\
&&  \;&\quad\cup \{\inr\brks{a,t} \mid t\in (\iter \widehat{\xi\comp f})(t''),  t''\in \xi(t'),\inr\brks{a,t'}\in\out p\}\\
&&  = \;&\{\inm\inl\inl x\mid \inl x\in\out f(x), \inl x\in \out p\} \cup \{\inm\inl\inr\star\}\\
&&  \;&\quad\cup \{\inr\brks{a,t''}\mid t''\in \xi(t'),\inr\brks{a,t'}\in \out f(x), \inl x\in \out p\}\\
&&  \;&\quad\cup \{\inr\brks{a,t''} \mid  t''\in (\xi f)^\klstar(\xi(t')),\inr\brks{a,t'}\in\out p\}
\end{flalign*}
Now consider the situation where guardedness for $\nu\gamma.\,\CSet(X+A\times\gamma)$
is induced by vacuous guardedness for $\CSet(\argument+A\times\gamma)$ by Theorem~\ref{thm:ext_guard}~(1)
and with guardedness for $\CSet(\mu\gamma.\,{(X+1)}+A\times\gamma)$ defined as follows:
$f\colon X\to \CSet(\mu\gamma.\,(Y+1)+A\times\gamma)$ is $\sigma$-guarded iff as a 
morphism $f\colon X\to\nu\gamma.\,(Y+1)+A\times\gamma$ in $\Set_{\CSet}$ it is 
$\sigma$\dash guarded under the notion of guardedness induced by vacuous guardedness
for $(\argument+1)+A\times\argument$ in $\Set_{\CSet}$, again by Theorem~\ref{thm:ext_guard}~(1).
This turns $\xi$ into a guarded monad morphism, and moreover~$\xi$ is iteration-preserving by 
Lemma~\ref{lem:iter-preserve}, because, as we argued before in Example~\ref{ex:transfer-example}~(4),
its codomain $\CSet(\mu\gamma.\,{(X+1)}+A\times\gamma)$ is guarded iterative (a more abstract argument showing that sandwiching a guarded 
iterative monad between an adjoint pair produces a guarded iterative monad is later given in Theorem~\ref{thm:sandwich_iter}). 

In order to obtain a guarded retraction from $\xi$, let $\rho$ be the
epimorphic part of the image factorization of $\xi$. It is easy to verify
that the codomain of $\rho$ consists precisely of the prefix-closed
subsets of $\CSet(\mu\gamma.\,{(X+1)}+A\times\gamma)$, i.e.\ is the
guarded iterative submonad of
$\CSet(\mu\gamma.\,{(X+1)}+A\times\gamma)$ mentioned in
Example~\ref{ex:transfer-example}~(4).  Under the \emph{axiom of choice},
this is sufficient to turn $\rho$ into a retraction because every epi
splits. However, the requisite section $\upsilon$ can also be
constructed explicitly without choice, for every prefix-closed subset
of $\CSet(\mu\gamma.\,{(X+1)}+A\times\gamma)$ standardly induces an
$A$-branching tree, hence an element of
$\nu\gamma.\,\CSet(X+A\times\gamma)$. In summary,
\begin{align*}
 \rho_X(t)     =&\; \{\inm\inl\inl x\!\mid\! \inl x\in\out t\}\cup \{\inm\inl\inr\star\}\cup \{\inm\inr\brks{a,t''}\!\mid\! t''\in \rho_X(t'),\inr\brks{a,t'}\in\out t\}\\
 \upsilon_X(S) =&\; \tuo\bigl(\bigl\{\inl x\mid \inm\inl\inl x\in S\bigr\}\cup \bigl\{\inr\brks{a,\upsilon_X(\{t\mid \inm\inr\brks{a, t}\in S\})}\mid a\in A\bigr\}\bigr) 
\end{align*}
where $t\in\nu\gamma.\,\CSet(X+A\times\gamma)$ and $S$ is a countable
prefix-closed subset of $\mu\gamma.\,{(X+1)}+{A\times\gamma}$. Note
that the tree constructed by $\upsilon$ has only very special kind of
nondeterminism, not including non-deterministic choice between
processes prefixed by actions. Roughly, we can have $x+y$ and $x+a.t$
in the image of $\upsilon$ with $x,y\in X$ and $a\in A$, but not
$a.t + b.s$ with $a,b\in A$. The composition~$\upsilon\rho$ can
therefore be seen as a determinization procedure, pushing the
non-deterministic choice downwards along the tree. Of course,
non-determinism can not be entirely eliminated, because in the end we
arrive at subsets of $X$, which must remain intact. We conjecture that
this effect is generic, i.e.\ that the same scenario can be run with
$\omega_1$ replaced by any other regular infinite cardinal $\kappa$;
that is, $\upsilon\rho$ pushes $\kappa$-branching non-determinism
downwards in the same sense as above.  We also conjecture that
$\upsilon$ is a monad morphism and hence so is $\upsilon\rho$.

The established retraction $(\rho,\upsilon)$ can thus be reused in two further cases.

\medskip \myparagraph{Guarded iteration for finitely-branching
  processes} We can restrict $\rho$ to the monad
$\nu\gamma.\,\FSet(\argument+A\times X)$ capturing finitely branching
processes with outputs in~$X$. As indicated above, we then essentially
again obtain countable prefix-closed sets $P$ of traces as the image
of $\rho$, which however now additionally satisfy the condition that
for each $w\in A^*$, the set $\{x\in X\mid (w,x)\in P\}$ is finite
(while in the countably branching case, and for infinite $X$, these
sets may be countably infinite). The section $\upsilon$ restricts
accordingly, and we thus obtain a guarded retraction.

\medskip
\myparagraph{Unguarded iteration for countably-branching processes} %
As discussed in Example~\ref{ex:transfer-example}~(2), 
$\nu\gamma.\,\CSet(X+A\times\gamma)$ supports
unguarded iteration, and in fact is an Elgot monad~\cite{GoncharovEA18}.
In the remainder of the example we use the terms ``unguarded'' for total guardedness
and ``guarded'' for the notion of guardedness on $\nu\gamma.\,\CSet(X+A\times\gamma)$ discussed above.  
Now, in order to conclude by Theorem~\ref{thm:elgot_is_retract} that the codomain 
of~$\rho$ as above is an Elgot monad, it suffices to check that~$\rho$ remains iteration preserving if we equip its domain with total guardedness, i.e.\ that~$\rho$ preserves iteration also of unguarded morphisms.\SG{I even more wonder if sandwiching theorems help here.} 
So let $f\colon X\to \nu\gamma.\,\CSet((Y+X)+A\times\gamma)$. The unguarded iterate
$f^\iistar$ is defined as the guarded iterate $\hat f^\iistar$, where $\hat f$ has 
the same profile as $f$ and is defined  as the guarded morphism
\begin{displaymath}
 \hat f = \tuo\CSet(\inl+\id)\comp \bigl(\CSet[\inl+\id,\inl\inr]\out f\bigr)^\istar,	
\end{displaymath}
with iteration $(\argument)^\istar$ on $\CSet$ calculated in the
expected way using least fixpoints~\cite{GoncharovEA18}.  It is easy
to check that
\begin{displaymath}
 \upsilon\rho\hat f = \tuo\CSet(\inl+\id)\comp \bigl(\CSet[\inl+\id,\inl\inr]\out \upsilon\rho f\bigr)^\istar	
\end{displaymath}
and thus $(\upsilon\rho\hat f)^\iistar = (\upsilon\rho f)^\iistar$ by
the above definition of $(\upsilon\rho f)^\iistar$. Therefore,
using~\eqref{eq:retraction-iteration} and the fact that, as we argued
above,~$\rho$ preserves guarded iteration,
$\rho f^\iistar = \rho\hat f^\iistar = \rho(\upsilon\rho \hat
f)^\iistar = \rho(\upsilon\rho f)^\iistar$,
which means that $\rho$ is iteration preserving.

\end{example}

\noindent Recall from Section~\ref{sec:parametrized} that guardedness,
guarded iteration, and the coinductive resumption transform work at
the level of \emph{parametrized monads}, i.e.\ functors from a
parameter category $\BD$ into the category of monads on a
category~$\BC$, typically rearranged into bifunctors
$\hash\colon \BC\times\BD\to\BC$. The notions of guarded retraction and
iteration congruence extend straightforwardly to parametrized monads;
explicitly:
\begin{definition}
  A parametrized guarded monad morphism is a \emph{guarded retraction}
  (an \emph{iteration congruence}) if its components are guarded
  retractions (iteration congruences).
\end{definition}
\noindent We then can take the claims of Theorem~\ref{thm:ext_guard}
further:
\begin{theorem}\label{thm:rho_ext}
  Let $\IB{}{},\IB[\hat\hash]{}{}\colon \BC\times(\BC\times\BD)\to\BC$ be
  guarded parametrized monads, and let
  $\rho\colon\IB{}{}\to\IB[\hat\hash]{}{}$ be an iteration-congruent
  retraction. By Theorem~\ref{thm:ext_guard},
  $\IB[\hash^\nu]{}{}=\nu\gamma.\,\IB{\argument}{(\gamma,\argument)}$ and
  $\IB[\hat\hash^\nu]{}{}=\nu\gamma.\,\IB[\hat\hash]{\argument}{(\gamma,\argument)}$
  are also parametrized guarded monads. Then
  $\rho^\nu\colon\IB[\hash^\nu]{}{}\to\IB[\hat\hash^\nu]{}{}$,
with components 
\begin{align*}
\rho^\nu_{X,Y} = \coit\bigl(\nu\gamma.\, \IB{X}{(\gamma,Y)}\xrightarrow{~\rho\out~}
\IB[\hat\hash]{X}{(\nu\gamma.\, \IB{X}{(\gamma,Y),Y)}} \bigr), 
\end{align*}
is again an iteration-congruent retraction.
\end{theorem}
\begin{proof}
It is already shown in Theorem~\ref{thm:ext_guard} that $\rho^\nu$ is a monad morphism.

We define the associated section by $\upsilon^\nu = \coit(\upsilon\out)$. Indeed it is easy to check 
that $\rho^\nu\upsilon^\nu=\id$: since
\begin{align*}
\out\comp \rho^\nu\upsilon^\nu 
=&\; (\IB[\hat\hash]{\id}{\rho^\nu})\comp\rho\out \upsilon^\nu\\
=&\; (\IB[\hat\hash]{\id}{\rho^\nu})\comp\rho (\IB{\id}{\upsilon^\nu})\comp \upsilon \out\\
=&\; (\IB[\hat\hash]{\id}{\rho^\nu}) (\IB[\hat\hash]{\id}{\upsilon^\nu})\comp \rho\upsilon \out\\
=&\; (\IB[\hat\hash]{\id}{\rho^\nu\upsilon^\nu})\comp  \out.
\end{align*}
and also $\out\id=(\IB[\hat\hash]{\id}{\id})\out$, the claim $\rho^\nu\upsilon^\nu=\id$ follows by uniqueness of final coalgebra morphisms.

Next, suppose that $f\colon X\to \IB[\hat\hash^\nu]{Y}{Z}$ is
$\sigma$-guarded, which according to
Theorem~\ref{thm:ext_guard} means that $\out f$ is
$\sigma$-guarded. We need to show that so is
$\upsilon^\nu f\colon X\to \IB[\hash^\nu]{Y}{Z}$. Now
\begin{align*}
\out\upsilon^\nu f =  (\IB{\id}{\upsilon^\nu})\comp \upsilon\out f 
\end{align*}
is $\sigma$-guarded because $\rho$ is a guarded retraction and hence
$\upsilon\out f $ is $\sigma$-guarded, and $\IB{\id}{\upsilon^\nu}$ is
a parametrized guarded monad morphism. Hence, again, according to
Theorem~\ref{thm:ext_guard}, $\upsilon^\nu f$ is $\sigma$-guarded. We
have thus proved that $\rho^\nu$ is a guarded retraction.

We are left to check that $\rho^\nu$ is an iteration congruence. Suppose that
$\rho^\nu f=\rho^\nu g$ for some $f,g\colon X\to_2 \IB[\hash^\nu]{(Y+X)}{Z}$. Then
\begin{align*}
\rho\comp(\IB{\id}{\rho^\nu})\out f = (\IB[\hat\hash]{\id}{\rho^\nu})\comp \rho\comp\out f = \out\rho^\nu f
\end{align*}
(by naturality of~$\rho$ and the definition of $\rho^\nu$) and
analogously for~$g$ in place of~$f$, so using that $\rho$ is an
iteration congruence, we obtain
\begin{align}\label{eq:rho_nu_simp}
\rho\comp((\IB{\id}{\rho^\nu})\out f)^\istar = \rho\comp((\IB{\id}{\rho^\nu})\out g)^\istar.
\end{align}
Observe that
for suitably typed $h$, 
\begin{align*}
\out\coit(\rho\out)\comp(\coit h)
=&\;(\IB[\hat\hash]{\id}{\coit(\rho\out)})\comp\rho\out\comp (\coit h)\\ 
=&\;(\IB[\hat\hash]{\id}{\coit(\rho\out)})\comp\rho\comp (\IB{\id}{(\coit h)})\comp h\\ 
=&\;(\IB[\hat\hash]{\id}{\coit(\rho\out)\comp(\coit h)})\comp \rho\comp h,
\end{align*} 
and therefore, by finality of $\coit(\rho h)$,
\begin{align}\label{eq:rho_out_coit}
\coit(\rho\out)\comp(\coit h) = \coit(\rho h).
\end{align}
Therefore,
\begin{flalign*}
&&\rho^\nu f^\iistar 
=&\; \rho^\nu\coit([\eta,(\out f)^\istar]^\klstar\out)\comp\eta\inr&\by{Theorem~\ref{thm:ext_guard}}\\
&&=&\; \coit(\rho\out)\coit([\eta,(\out f)^\istar]^\klstar\out)\comp\eta\inr\\ 
&&=&\; \coit(\rho[\eta,(\out f)^\istar]^\klstar\out)\comp\eta\inr&\by{\eqref{eq:rho_out_coit}}\\ 
&&=&\; \coit([\eta,\rho\comp((\IB{\id}{\rho^\nu})\comp\out f)^\istar]^\klstar\out)\comp\rho^\nu\eta\inr.
\intertext{
The last step is due to uniqueness of the final coalgebra morphism $\coit([\eta,\rho\comp(\out f)^\istar]^\klstar\rho\out)$ and 
the following calculation:}
&&\out \coit\bigl(&[\eta,\rho\comp((\IB{\id}{\rho^\nu})\comp\out f)^\istar]^\klstar\out\bigr)\comp\rho^\nu&\\
&&=&\;\bigl(\IB[\hat\hash]{\id}{\coit([\eta,\rho\comp((\IB{\id}{\rho^\nu})\comp\out f)^\istar]^\klstar\out)}\bigr)\\
&&&\;\quad[\eta,\rho\comp((\IB{\id}{\rho^\nu})\comp\out f)^\istar]^\klstar\out\rho^\nu\\
&&=&\;\bigl(\IB[\hat\hash]{\id}{\coit([\eta,\rho\comp((\IB{\id}{\rho^\nu})\comp\out f)^\istar]^\klstar\out)}\bigr)\\
&&&\;\quad\rho\comp [\eta,((\IB{\id}{\rho^\nu})\comp\out f)^\istar]^\klstar(\IB{\id}{\rho^\nu})\out&\by{$\rho$ is a monad morphism}\\
&&=&\;\bigl(\IB[\hat\hash]{\id}{\coit([\eta,\rho\comp((\IB{\id}{\rho^\nu})\comp\out f)^\istar]^\klstar\out)}\bigr)\\
&&&\;\quad\rho\comp(\IB{\id}{\rho^\nu})\comp [\eta,(\out f)^\istar]^\klstar\out&\by{$\IB{\id}{\rho^\nu}$ is a monad morphism}\\
&&=&\;\bigl(\IB[\hat\hash]{\id}{\coit([\eta,\rho\comp((\IB{\id}{\rho^\nu})\comp\out f)^\istar]^\klstar\out)}\comp\rho^\nu\big)\\
&&&\;\quad\rho\comp [\eta,(\out f)^\istar]^\klstar\out.
\end{flalign*}
An analogous calculation applies to $\rho^\nu g^\iistar$, and therefore 
by~\eqref{eq:rho_nu_simp}, $\rho^\nu f^\iistar=\rho^\nu g^\iistar$.
\end{proof}
\noindent Theorems~\ref{thm:elgot_is_retract} and~\ref{thm:rho_ext}
jointly provide a simple and structured way of showing that Elgotness
extends along the parametrized monad transformer
$\hash{}{}\mapsto\IB[\hat\hash]{}{}$: If $\IB{\argument}{X}$ is Elgot,
then by Theorem~\ref{thm:elgot_is_retract} there is an
iteration-congruent retraction
$\rho\colon\nu\gamma.\,\IB{\argument+\gamma}{X}\to\IB{\argument}{X}$. By
Theorem~\ref{thm:rho_ext}, this gives rise to an iteration-congruent
retraction
\begin{displaymath}
\rho^\nu\colon\nu\gamma'.\,\nu\gamma.\,\IB{\argument+\gamma}{(\gamma',X)}\to\nu\gamma'.\,\IB{\argument}{(\gamma',X)}
\end{displaymath}
and by Theorem~\ref{thm:elgot_is_retract}, the right-hand side is 
again Elgot. We have thus proved%
\begin{cor}\label{cor:ext}
Given a parametrized monad $\IB{}{}$ and $X\in |\BC|$, if $\IB{\argument}{X}$ is Elgot then so is $\IB[\hash^\nu]{\argument}{X}=\nu\gamma.\,\IB{\argument}{(\gamma,X)}$. 
\end{cor}
\noindent In particular, we have thus obtained a more structured and
simpler proof of one of the main results
in~\cite{GoncharovEA18}, which states that the coinductive
generalized resumption monad transformer preserves Elgotness. 
Theorem~\ref{thm:elgot_is_retract} characterizes iteratable Elgot
monads as iteration-congruent retracts of their
$(\argument)^\nu$-transforms. We take this perspective further as
follows. 

\begin{definition}
  We extend the notation $F^\nu=\nu\gamma.F(-+\gamma)$ to
  functors~$F$. We say that a functor~$F$ is
\begin{itemize}
\item \emph{$1$-iteratable} if $F^\nu$ exists,
  \item \emph{$(n+1)$-iteratable} if $F^\nu$ is $n$-iteratable, 
  \item \emph{$\omega$-iteratable} if $F$ is $n$-iteratable for every $n$.
\end{itemize}
We apply all these notions mainly to monads $\BBT$, referring to their
underlying functor~$T$.
\end{definition}
\begin{remark}
  Note that for every natural number $n$,
\begin{align*}
\nu\gamma'.\,\nu\gamma.\,T(X+\gamma'+n\times\gamma)
  \;\iso \nu\gamma.\,T(X+\gamma+n\times\gamma)
  \;\iso \nu\gamma.\,T(X+(n+1)\times\gamma),
\end{align*}
where $n\times X$ denotes the $n$-fold sum $X+\ldots+X$. It follows by
induction that $n$-iteratability of $T$ is equivalent to the
assumption that all coalgebras $\nu\gamma.\,T(X+n\times\gamma)$ exist,
a condition that does not appear much stronger than iteratability
of~$T$. Still, the $2$-iteratable functors are properly contained in
the iteratable functors, as the following example shows. Let $\BC$ be
the category of countable sets and $T=\Id$. Then, it is easy to see
that $T^\nu X$ is isomorphic to $X\times \Nat+1$, hence $\Id$ is
iteratable.  However, it is not $2$-iteratable, because
$(T^\nu)^\nu\iobj \cong \nu\gamma.\,2\times\gamma$ can be
characterized as the object of all infinite bit streams, which does
not fit into $\BC$ for cardinality reasons. Showing this formally 
amounts to mimicking Cantor's classical diagonalization argument.

We expect that separating $n$-iteratability from $(n+1)$-iteratability
for $n>1$ would involve much less natural examples, as the previous
cardinality argument typically would not apply.

\end{remark}
\noindent Consider the functor $\BBT\mapsto\BBT^\nu$ on the category
of $\omega$-iteratable monads over
$\BC$. %
This construction is itself a monad: the unit~$\bta$
is the natural transformation with components
$\bta_X=\tuo(T\inl)\colon TX\to T^\nu X$, and the multiplication
$\bmu\colon T^{\nu\nu}\to T^\nu$ has components
\begin{align*}
\bmu_X = \coit\bigl(T[\id,\inr\tuo]\out\out\gray{\colon T^{\nu\nu} X\to T(X+T^{\nu\nu} X)}\bigr). 
\end{align*} 
We record explicitly that the relevant laws are satisfied:
\begin{lem}
  With multiplication $\bmu$ and unit $\bta$ as defined above, the
  construction $(\argument)^\nu$ becomes a monad on the (overlarge)
  category of $\omega$-iteratable monads.
\end{lem}
\begin{proof}
  By coinduction. Using the definitions of $\bmu$ and $\bta$, we have
  \begin{flalign*}
    && \out\bmu\bta
    =\;& \out \coit\bigl(T[\id,\inr\tuo]\out\out\bigr) \tuo(T^\nu\inl)&\\
    &&=\;& T\bigl(\id+\coit(T[\id,\inr\tuo]\out\out)\bigr)\comp T[\id,\inr\tuo]\out\comp (T^\nu\inl)&\\
    &&=\;& T\bigl(\id+\coit(T[\id,\inr\tuo]\out\out)\bigr)\comp T[\id,\inr\tuo]\comp T\bigl(\inl+ (T^\nu\inl)\bigr)\out&\\
    &&=\;& T\bigl(\id+\coit(T[\id,\inr\tuo]\out\out)\tuo\comp (T^\nu\inl)\bigr)\out&\\
    &&=\;& T\bigl(\id+ \bmu\bta\bigr)\out,& \intertext{
      and therefore $\bmu\bta = \id$ by uniqueness of final coalgebra
      morphisms. Analogously, }
    && \out\bmu\bta^\nu
    =\;& \out\coit\bigl(T[\id,\inr\tuo]\out\out\bigr)\coit (\bta\out) &\\
    &&=\;& T\bigl(\id+\coit(T[\id,\inr\tuo]\out\out)\bigr)\comp\\
    &&& T[\id,\inr\tuo]\out T(\id+\coit (\bta\out))\comp \bta\out &\\
    &&=\;& T\bigl(\id+\coit(T[\id,\inr\tuo]\out\out)\bigr)\comp\\
    &&& T[\id,\inr\tuo]\comp (T\inl)\comp T\bigl(\id+\coit (\bta\out)\bigr)\comp \out &\\
    &&=\;& T\bigl(\id+\coit(T[\id,\inr\tuo]\out\out)\bigr)\comp T\bigl(\id+\coit (\bta\out)\bigr)\comp \out &\\
    &&=\;& T(\id+\bmu \bta^\nu) \out &\\
    \intertext{and therefore $\out\bmu \bta^\nu=\id$. The remaining
      law $\bmu\bmu = \bmu\comp\bmu^\nu$ follows by the same argument
      from} && \out\bmu\bmu^\nu
    =\;& \out \coit\bigl(T[\id,\inr\tuo]\out\out\bigr)\coit (\coit\bigl(T[\id,\inr\tuo]\out\out\bigr)\out) &\\
    &&=\;& T\bigl(\id+\coit(T[\id,\inr\tuo]\out\out)\bigr)\comp T[\id,\inr\tuo]\out\\
    &&& T\bigl(\id+\coit
    (\coit\bigl(T[\id,\inr\tuo]\out\out\bigr)\out) \bigr)\comp
    \coit\bigl(T[\id,\inr\tuo]\out\out\bigr)\out &\\
    &&=\;& T(\id+\bmu)\comp T[\id,\inr\tuo]\out T(\id+\bmu^\nu)\comp
    \bmu\out &\\
    &&=\;& T[\id+\bmu\bmu^\nu, \inr\bmu\tuo T(\id+\bmu^\nu)]\out\bmu\out &\\
    &&=\;& T[\id+\bmu\bmu^\nu, \inr\bmu\tuo T(\id+\bmu^\nu)\comp\bmu]\comp T[\id,\inr\tuo]\out\out\out &\\
    &&=\;& T[\id+\bmu\bmu^\nu, \inr\bmu\bmu^\nu\tuo]\comp T[\id,\inr\tuo]\out\out\out &\\
    &&=\;& T(\id+\bmu\bmu^\nu)\comp T[\id,\inr\tuo]\comp T[\id,\inr\tuo]\out\out\out &\\[2ex]
    && \out\bmu\bmu
    =\;& \out \coit\bigl(T[\id,\inr\tuo]\out\out\bigr)\coit\bigl(T^\nu[\id,\inr\tuo]\out\out\bigr) &\\
    &&=\;& T(\id+\bmu)\comp T[\id,\inr\tuo]\out T^\nu(\id+\bmu)\comp T^\nu[\id,\inr\tuo]\out\out\\
    &&=\;& T[\id+\bmu\bmu,\inr\bmu\tuo T^\nu(\id+\bmu)]\comp\out T^\nu[\id,\inr\tuo]\out\out\\
    &&=\;& T[(\id+\bmu\bmu)\comp [\id,\inr\tuo],\inr\bmu\tuo T^\nu(\id+\bmu)\comp T^\nu[\id,\inr\tuo]]\comp\out\out\out\\
    &&=\;& T[(\id+\bmu\bmu)\comp [\id,\inr\tuo],\inr\bmu\bmu\tuo\tuo]\comp\out\out\out\\
    &&=\;& T(\id+\bmu\bmu)\comp T[[\id,\inr\tuo],\inr\tuo\tuo]\comp\out\out\out\\
    &&=\;& T(\id+\bmu\bmu)\comp T[\id,\inr\tuo]\comp T[\id,\inr\tuo]\out\out\out \tag*{\qedhere}
  \end{flalign*}
\end{proof}
\noindent For every $T$ we now define the \emph{delay transformation}
\begin{equation*}
\del =\tuo\eta\inr\colon T^\nu\to T^\nu.
\end{equation*}
This leads to our second main result:
\begin{theorem}\label{thm:nu-algebras}
  The category of $\omega$-iteratable Elgot monads over $\BC$ is
  isomorphic to the full subcategory of the category of
  $(\argument)^\nu$-algebras consisting of the
  $(\argument)^\nu$-algebras $(\BBS,\rho\colon\BBS^\nu\to \BBS)$ (for $\omega$-iteratable $\BBS$) satisfying ${\rho\del=\rho}$.
\end{theorem}
\noindent We refer to the condition ${\rho\del=\rho}$ as \emph{delay
  cancellation}. 
\begin{rem}\label{exmp:delay}
  The point of the above result is to systematize the connection
  between the $(\argument)^\nu$ construction and Elgot monads
  previously indicated by
  Theorem~\ref{thm:elgot_is_retract}. Alternative efforts to show that
  Elgotness is monadic exist (see Section~\ref{sec:related}) but
  necessarily involve quite different monads than $(\argument)^\nu$:
  Any monad~$\mathfrak{M}$ (on a category of monads) whose algebras are
  precisely the Elgot monads would itself have to produce Elgot monads
  $\mathfrak{M}\BBT$, while the point of involving $(\argument)^\nu$ is to obtain
  Elgot monads from guarded iterative ones.

  Formally, the following simple example shows that the delay
  cancellation condition $\rho\del=\rho$ cannot be omitted from
  Theorem~\ref{thm:nu-algebras}. Let $\catname{Mon}(\BC)^\nu$ be the
  category of $(\argument)^\nu$-algebras, and let
  $\catname{Mon}(\BC)^\nu_{\del}$ be the full subcategory of
  $\catname{Mon}(\BC)^\nu$ figuring in
  Theorem~\ref{thm:nu-algebras}. Since the identity functor is the
  initial monad, the initial object of $\catname{Mon}(\BC)^\nu$ is
  Capretta's delay monad~\cite{Capretta05}
  $D=\nu\gamma.\,(\argument+\gamma)$.  On the other hand, the initial
  object of $\catname{Mon}(\BC)^\nu_{\del}$ (if it exists) is the
  \emph{initial Elgot monad}~$\BBL$, which on $\BC=\Set$ is the
  \emph{maybe monad} $(\argument)+1$.

  If $\BC=\Set$, then $DX=(X\times\Nat + 1)$ does turn out to be
  Elgot~\cite{GoncharovMiliusEtAl16} (but applying
  Theorem~\ref{thm:nu-algebras} to $D$ qua Elgot monad yields a
  different $(\argument)^\nu$-algebra structure than the initial one),
  and $\BBL$ is, in this case, a retract of $\BBD$ in
  $\catname{Mon}(\BC)^\nu_{\del}$. %
  The situation is more intricate in categories with a nonclassical
  internal logic, for which $\BBD$ is mainly intended. We believe that
  in such a setting, neither is $\BBD$ Elgot in general, nor is $\BBL$
  the maybe monad.  However, there will still be a unique
  $(\argument)^\nu$-algebra morphism $\BBD\to\BBL$ in
  $\catname{Mon}(\BC)^\nu$. %
\end{rem}

\begin{proof}[Proof (Theorem~\ref{thm:nu-algebras})]
We fix the notation $(\eta,\argument^\klstar,\istar)$ for (potential) Elgot monads over 
$\BC$ and $(\eta^\nu,\argument^\kklstar,\iistar)$ for their $(\argument)^\nu$-transforms. 
We record the following identity, satisfied by any monad morphism~$\rho$ for which
$\rho\bta=\id$ and $\rho\del=\rho$: 
\begin{align}\label{eq:rho_out}
\rho=[\eta,\rho]^\klstar\out.
\end{align}
Indeed, 
\begin{flalign*}
&&[\eta,\rho]^\klstar\out=&\;[\eta,\rho]^\klstar\rho\bta\out&\by{$\rho\bta=\id$}\\
&&=&\;[\eta,\rho\del]^\klstar\rho\bta\out&\by{$\rho\del=\rho$}\\
&&=&\;\rho\comp[\eta,\del]^\kklstar\bta\out&\by{$\rho$ is a monad morphism}\\
&&=&\;\rho\comp[\eta,\del]^\kklstar\tuo(S\inl)\out\\
&&=&\;\rho\tuo\comp[\out[\eta,\id],\eta\inr[\eta,\id]^\kklstar]^\klstar(S\inl)\out\\
&&=&\;\rho\tuo\comp[\eta\inl,\eta\inr]^\klstar\out\\
&&=&\;\rho.
\end{flalign*}
For the inclusion from Elgot monads to $(\argument)^\nu$-algebras, let
$\BBS$ be an Elgot monad. By Theorem~\ref{thm:elgot_is_retract},
$\BBS$ is an iteration-congruent retract of $\BBS^\nu$ with
$S^\nu X = \nu\gamma.\,S(X+\gamma)$; specifically,
$\upsilon = \del\bta\colon S\to S^\nu$ is a left inverse to $\rho=\out^\istar\colon S^\nu\to S$.

First of all, it is easy to see that %
\begin{flalign*}
\rho\del = [\eta,\rho]^\klstar\out\del 
= [\eta,\rho]^\klstar\eta\inr 
= \rho.
\end{flalign*}
Moreover, we need to show the axioms of $(\argument)^\nu$-algebras:
\begin{align}\label{eq:meta_m}
\rho\bta=\id\qquad \text{and}\qquad  \rho\bmu =\rho\rho^\nu
\end{align}
where $\rho^\nu = \coit(\rho\out)\colon S^{\nu\nu}\to S^\nu$. For the left axiom, we 
readily have $\id = \rho\upsilon = \rho\del\bta=\rho\bta$. The right axiom is 
shown as follows:
\begin{align*}
\rho\bmu&\; 
\overset{\scriptscriptstyle(i)}{=\joinrel=}
\rho[\eta,\comp(\del\out)^\iistar]^\kklstar\out\overset{\scriptscriptstyle(ii)}{=\joinrel=}
\rho\bigl(\tuo S(\inl+\eta\inr)\rho\out\bigr)^\iistar
\overset{\scriptscriptstyle(iii)}{=\joinrel=}
\rho\rho^\nu.
\end{align*}
To show step (i), first observe that on the one hand
\begin{align*}
\bmu = \coit\bigl(S[\id+\out,\inr]\out\bigr)\comp\out
\end{align*}
Indeed, let $t=\coit\bigl(S[\id+\out,\inr]\out\bigr)$. Then
\begin{align*}
\out\comp t\comp\out=&\; S(\id+t)\comp S[\id+\out,\inr]\out\out\\
=&\; S[\id+t\out,t\inr]\out\out\\
=&\; S(\id+\comp t\out)\comp S[\id,\inr\tuo]\out\out,
\end{align*}
which means that $\comp t\out$ satisfies the equation uniquely characterizing 
$\bmu$, hence $\bmu=t\out$. On the other hand, $[\eta,\comp(\del\out)^\iistar]^\kklstar$
satisfies the equation characterizing~$t$. In order to see this, note that
\begin{equation}
  \out(\del\out)^\iistar =\;\eta\inr[\eta^\nu,(\del\out)^\iistar]^\kklstar\out, \label{eq:out_del_dag}
\end{equation}
witnessed by the following calculation:
\begin{align*}
 \out(&\del\out)^\iistar \\
=&\;\out [\eta^\nu,(\del\out)^\iistar]^\kklstar\del\out\\ 
=&\;[\out[\eta^\nu,(\del\out)^\iistar],\eta\inr[\eta^\nu,(\del\out)^\iistar]^\kklstar]^\klstar\out\del\out\\
=&\;[\out[\eta^\nu,(\del\out)^\iistar],\eta\inr[\eta^\nu,(\del\out)^\iistar]^\kklstar]^\klstar\eta\inr\out\\ 
=&\;\eta\inr[\eta^\nu,(\del\out)^\iistar]^\kklstar\out
\end{align*}
Therefore,
\begin{flalign*}
&&\out [\eta^\nu,\comp(\del\out)^\iistar]^\kklstar 
=&\; [\out [\eta^\nu,\comp(\del\out)^\iistar],\eta\inr [\eta^\nu,\comp(\del\out)^\iistar]^\kklstar]^\klstar\out\\
&&=&\; [[\out \eta^\nu,\out \comp(\del\out)^\iistar],\eta\inr [\eta^\nu,\comp(\del\out)^\iistar]^\kklstar]^\klstar\out\\ 
&&=&\; [\eta\comp(\id+[\eta^\nu,\comp(\del\out)^\iistar]^\kklstar\comp\out),\eta\inr [\eta^\nu,\comp(\del\out)^\iistar]^\kklstar]^\klstar\out&\by{\eqref{eq:out_del_dag}}\\ 
&&=&\; [\eta\comp(\id+[\eta^\nu,\comp(\del\out)^\iistar]^\kklstar)\comp[\id + \out,\inr]]^\klstar\out\\ 
&&=&\; S(\id+[\eta^\nu,\comp(\del\out)^\iistar]^\kklstar)\comp S[\id+\out,\inr]\out.
\end{flalign*}
In summary we obtain
\begin{align*}
\bmu = \coit\bigl(S[\id+\out,\inr]\out\bigr)\comp\out = [\eta^\nu,(\del\out)^\iistar]^\kklstar\out,
\end{align*}
which justifies~(i). Let us check (iii). Let us denote $\tuo S(\inl+\eta^\nu\inr)\rho\out$ by $t$. Then
\begin{align*}
\out t^\iistar =\;&\out [\eta^\nu,t^\iistar]^\kklstar \tuo S(\inl+\eta^\nu\inr)\rho\out\\
=\;&[\out[\eta^\nu,t^\iistar],\eta\inr[\eta^\nu,t^\iistar]^\kklstar]^\klstar S(\inl+\eta^\nu\inr)\rho\out\\
=\;&[\out\eta^\nu,\eta\inr t^\iistar]^\klstar\rho\out\\
=\;&S(\id+ t^\iistar)\rho\out,
\end{align*}
and therefore $t^\iistar$ satisfies the equation characterizing $\rho^\nu$, hence
$\rho^\nu=t^\iistar$. Finally, we proceed with the proof of~(ii). Using the fact
that $\rho$ is a monad morphism and that it cancels $\del$, we obtain that
\begin{align*}
\rho[\eta,(\del\out)^\iistar]^\kklstar\comp\out =&\;
[\eta,\rho\comp(\del\out)^\iistar]^\kklstar\rho\comp\out\\
=&\;[\eta,\rho\comp(\del\out)^\iistar]^\kklstar\rho\del\comp\out\\
=&\;\rho [\eta,\comp(\del\out)^\iistar]^\kklstar\del\comp\out\\
=&\;\rho(\del\out)^\iistar.&\by{fixpoint}
\end{align*}
In order to finish the proof of~(ii), it suffices to check that
\begin{align}\label{eq:ii}
\rho\del\out = \rho\tuo S(\inl+\eta\inr)\rho\out
\end{align}
and call the assumption that $\rho$ is an iteration congruence. The proof 
of~\eqref{eq:ii} runs as follows:
\begin{align*}
\rho\tuo S(\inl+\eta\inr)\rho\out 
=&\; [\eta,\rho]^\klstar\out\tuo S(\inl+\eta\inr)\rho\out\\
=&\; [\eta,\rho]^\klstar S(\inl+\eta\inr)\rho\out\\
=&\; [\eta\inl,\rho\eta\inr]^\klstar\rho\out\\
=&\; \rho\out\\
=&\; \rho\del\out.
\end{align*}
We have thus proved the claimed inclusion on objects. To extend the claim to morphisms, suppose that $\alpha\colon \BBS\to\BBT$ is an Elgot monad morphism, i.e.\ a monad morphism
such that $\alpha f^\istar = (\alpha f)^\istar$, and let us show that it is also
a morphism of the corresponding $(\argument)^\nu$-algebras, i.e.\ $\alpha\rho=\rho\alpha^\nu$.
Indeed, on the one hand $\alpha\rho=\alpha\out^\istar=(\alpha\out)^\istar$, and also
on the other hand, by uniformity of $(\argument)^\istar$, $\rho\alpha^\nu = 
\out^\istar\alpha^\nu=(\alpha\out)^\istar$, since $\out\alpha^\nu = T(\id+\alpha^\nu)\alpha\out$.

We proceed with the converse inclusion, i.e.\ from
$(\argument)^\nu$-algebras to Elgot monads. So assume that
$(\BBS,\rho)$ is a $(\argument)^\nu$-algebra, i.e.\ the
laws~\eqref{eq:meta_m} are satisfied, and $\rho\del=\rho$. We claim
that $\BBS$ equipped with the iteration operation
$f^\istar = \rho(\coit f)$ is an Elgot monad. The corresponding axioms
are verified as follows.
\begin{citemize}
 \item\emph{Fixpoint.} Let $f\colon X\to S(Y+X)$. Then $f^\iistar = \rho(\coit f)$ and 
hence
\begin{flalign*}
&&f^\iistar =\;&\rho(\coit f)\\
&&=\;&[\eta,\rho]^\klstar\out\comp(\coit f)\\
&&=\;&[\eta,\rho]^\klstar S(\id+\coit f) f&\by{\eqref{eq:rho_out}}\\
&&=\;&[\eta,\rho\coit f]^\klstar f\\
&&=\;&[\eta,f^\iistar]^\klstar f.
\end{flalign*}
 \item\emph{Naturality.} Let $f\colon X\to S(Y+X)$ and $g\colon  Y \to SZ$. Then
\begin{align*}
g^{\klstar} f^{\istar} =&\; g^\klstar \rho \comp(\coit f) = \rho\comp (\bta g)^\kklstar \comp(\coit f),\\
([(S\inl) \comp g, \eta\inr]^{\klstar} \comp f)^{\istar} =&\; \rho \coit([(S\inl) \comp g, \eta\inr]^{\klstar} \comp f).
\end{align*}
We are left to show that $(\bta g)^\kklstar (\coit f)$ satisfies the equation for 
$\coit([(S\inl) \comp g, \eta\inr]^{\klstar} \comp f)$. This runs as follows:
\begin{align*}
\out (\bta g)^\kklstar (\coit f) 
=&\; [\out\bta\comp g,\eta\inr(\bta g)^\kklstar]^\klstar\out \comp (\coit f)\\
=&\; [(S\inl) g,\eta\inr(\bta g)^\kklstar]^\klstar S(\id+\coit f) f\\
=&\; [(S\inl) g,\eta\inr(\bta g)^\kklstar (\coit f)]^\klstar  f\\
=&\; S(\id+(\bta g)^\kklstar (\coit f))\comp [(S\inl) \comp g, \eta\inr]^{\klstar} \comp f.
\end{align*}
\item\emph{Codiagonal.} Let $f\colon X\to S((Y+X)+X)$. Observe that since
  \begin{align*}
    &\out \comp (\coit(\rho \comp \out)) \comp (\coit(\coit f)) \\
    =&~S(\id+\coit(\rho \comp \out)) \comp \rho \comp \out \comp (\coit(\coit f)) \\
    =&~S(\id+\coit(\rho \comp \out)) \comp \rho \comp S^{\nu}(\id + \coit(\coit f)) \comp (\coit f) \\
    =&~S(\id+\coit(\rho \comp \out)) \comp S(\id + \coit(\coit f)) \comp \rho \comp (\coit f) &\by{naturality of $\rho$} \\
    =&~S(\id+(\coit(\rho \comp \out))\comp (\coit(\coit f))) \comp \rho \comp (\coit f),
  \end{align*}
  we have that
  \begin{equation}\label{eq:rho_nu_coit_f}
    \coit(\rho \comp (\coit f)) = (\coit(\rho \comp \out)) \comp (\coit (\coit f))
  \end{equation}
  by uniqueness of final morphisms. Thus,
\begin{flalign*}
&&f^{\istar\istar} =\;&\rho\comp (\coit(\rho(\coit f)))&\\
&&=\;&\rho\comp (\coit(\rho \comp \out))(\coit (\coit f))&\by{\eqref{eq:rho_nu_coit_f}}\\
&&=\;&\rho\comp\rho^\nu\comp(\coit(\coit f))&\by{definition of~$\rho^\nu$}\\
&&=\;&\rho\comp\bmu\comp(\coit(\coit f)). &\by{\eqref{eq:meta_m}}
\end{flalign*}
Since by definition, $(S[\id,\inr]\comp f)^\istar = \rho\coit(S[\id,\inr]\comp f)$, we are only
left to check that $\coit(S[\id,\inr]\comp f)=\bmu \comp (\coit(\coit f))$. This is easy to
establish directly by showing that the right-hand side satisfies the equation
characterizing the left-hand side:
\begin{align*}
  \out\bmu&\coit(\coit f)\\
  =&\;S(\id+\bmu) \comp S[\id,\inr\tuo]\out\out\coit(\coit f)\\
  =&\;S[\id+\bmu,\inr\bmu\tuo]\out\out\comp(\coit(\coit f))\\
  =&\;S[\id+\bmu,\inr\bmu\tuo]\comp \out \comp S^\nu(\id+\coit(\coit f))\comp \coit f\\
  =&\;S[\id+\bmu,\inr\bmu\tuo]\comp S((\id+\coit(\coit f))+S^\nu(\id+\coit(\coit f)))\comp \out \comp \coit f\\
  =&\;S[\id+\bmu,\inr\bmu\tuo]\comp S((\id+\coit(\coit f))+S^\nu(\id+\coit(\coit f)))\comp S(\id+\coit f) f\\
  =&\;S[\id+\bmu\comp (\coit(\coit f)),\inr\comp\bmu\comp\tuo S^\nu(\id+\coit(\coit f))(\coit f)]\comp f\\
  =&\;S[\id+\bmu\comp (\coit(\coit f)),\inr\comp\bmu\comp(\coit(\coit f))]\comp f\\
  =&\;S(\id+\bmu\comp (\coit(\coit f)))\comp S[\id,\inr]\comp f.
\end{align*} 
 \item\emph{Uniformity.} Let $f\colon  X \to T(Y + X)$, $g\colon  Z \to S(Y + Z)$, $h\colon  Z \to X$
and suppose that $f \comp h = T(\id+ h) \comp g$. If follows standardly by uniqueness of final coalgebra morphisms that $(\coit f)\comp h = \coit g$ and therefore
\begin{align*}
f^\iistar \comp h 
= \rho\comp (\coit f)\comp h
= \rho\comp (\coit g)
= g^\iistar.
\end{align*}
\end{citemize}
Finally, let us check that every $(\argument)^\nu$-algebra morphism $\alpha\colon \BBT\to\BBS$
is an Elgot monad morphism. By assumption we have that $\alpha\rho=\rho\alpha^\nu$,
and therefore, for every $f\colon X\to S(Y+X)$, $\alpha f^\istar = \alpha\rho(\coit f)=
\rho\alpha^\nu(\coit f)=\rho\coit(\alpha\out)(\coit f)$. It is then 
straightforward to verify that $\rho\coit(\alpha\out)(\coit f) = \rho\coit(\alpha f)=(\alpha f)^\istar$.
\end{proof}

\section{A Sandwich Theorem for Elgot Monads}\label{sec:sandwich}

As an application of Theorem~\ref{thm:elgot_is_retract}, we show that
sandwiching an Elgot monad between a pair of adjoint functors again
yields an Elgot monad. A similar result has previously been shown for
completely iterative monads~\cite{PirogGibbons15}; this result
generalizes straightforwardly to guarded iterative monads:
\begin{theorem}\label{thm:sandwich_iter}
  Let $F \colon\BC \to \BD$ and $U \colon\BD \to \BC$ be a pair of adjoint
  functors with associated natural isomorphism
  $\Phi \colon\BD(FX,Y) \to \BC(X,UY)$, and let $\BBT$ be a guarded
  iterative monad on\/ $\BD$. Then the monad induced on the composite
  functor $UTF$ is guarded iterative, with the guardedness relation
  defined by taking $f\colon  X \to_\sigma UTFY$ if and only if
  $\Phi^{\mone}f\colon  FX \to_\sigma TFY$, and unique solutions given by
  $f \mapsto \Phi((\Phi^{\mone}f)^\istar)$.
\end{theorem}
\begin{proof}
  First, we need to verify that the guardedness relation defined in
  the claim satisfies the rules from Definition~\ref{def:g-cat}. Note
  that since left adjoints preserve coproducts (LAPC), we can assume
  w.l.o.g.\ that $F(X+Y) = FX + FY$.
\begin{citemize}
 \item\textbf{(trv)} Let $f\colon  X \to UTFY$ be a morphism. By \textbf{(trv)} for $\BBT$, we have $(T\inl)(\Phi^{\mone}f)\colon FX \to_\sigma T(FY+FX)$. Then, the following holds:
\begin{flalign*}
&&(T\inl)(\Phi^{\mone}f)
=&\; (TF\inl)(\Phi^{\mone}f)&\by{LAPC}\\
&&=&\; \Phi^{\mone}((UTF\inl)f)&\by{$\Phi$ is a natural isomorphism}
\end{flalign*}
Thus, $\Phi^{\mone}((UTF\inl)f)\colon FX \to_\sigma T(FY+FX)$, so
$(UTF\inl)f\colon  X \to_\sigma UTF(Y+X)$.

\item\textbf{(par)} Let $f\colon  X \to_\sigma UTFZ$ and $g\colon  Y \to_\sigma UTFZ$. This means that $\Phi^{\mone}f\colon  FX \to_\sigma TFZ$ and $\Phi^{\mone}g\colon  FY \to_\sigma TFZ$, hence, by \textbf{(par)} for $\BBT$, $[\Phi^{\mone}f,\Phi^{\mone}g]\colon  FX + FY \to_\sigma TFZ$. By LAPC, we have $[\Phi^{\mone}f,\Phi^{\mone}g] = \Phi^{\mone}[f,g]\colon  F(X + Y) \to_\sigma TFZ$, so $[f,g]\colon  X+Y \to_\sigma UTFZ$.

 \item\textbf{(cmp)} Let $f\colon  X \to_2 UTF(Y+Z)$, $g\colon  Y \to_\sigma UTFV$, and $h\colon  Z \to UTFV$ be morphisms. Then, by \textbf{(cmp)} for $\BBT$, we obtain that $[\Phi^{\mone}g, \Phi^{\mone}h]^{\klstar}(\Phi^{\mone} f)\colon  FX \to_\sigma TFV$. Then, the following holds:
\begin{flalign*}
&& &\;[\Phi^{\mone}g, \Phi^{\mone}h]^{\klstar}(\Phi^{\mone} f)&\\
&&=&\; \mu^\BBT(T[\Phi^{\mone}g, \Phi^{\mone}h])(\Phi^{\mone} f)&\\
&&=&\; \mu^\BBT(\Phi^{\mone}(UT[\Phi^{\mone}g, \Phi^{\mone}h])f)&\by{$\Phi$ is a nat. iso.}\\
&&=&\; \Phi^{\mone}((U\mu^\BBT)(UT[\Phi^{\mone}g, \Phi^{\mone}h])f)&\by{$\Phi$ is a nat. iso.}\\
&&=&\; \Phi^{\mone}((U\mu^\BBT)(UT\Phi^{\mone}[g, h])f)&\by{LAPC}\\
&&=&\; \Phi^{\mone}((U\mu^\BBT)(UT\Phi^{\mone}\id)(UTF[g, h])f)&\by{$\Phi$ is a nat. iso.}\\
&&=&\; \Phi^{\mone}(\mu^{UTF}(UTF[g, h])f)&\\
&&=&\; \Phi^{\mone}([g, h]^{\klstar}f).
\end{flalign*}
Thus, $\Phi^{\mone}([g, h]^{\klstar}f)\colon  X \to_\sigma UTF(Y+X)$ in
$\BBT$, so $[g, h]^{\klstar}f\colon  FX \to_\sigma TF(Y+X)$ in the monad on $UTF$.
\end{citemize}
This means that if $f\colon  X \to_2 UTF(Y+X)$, then
$\Phi^{\mone}f\colon  FX \to_2 T(FY+FX)$, so $\Phi^{\mone}f$ has a unique
solution due to the fact that $\BBT$ is guarded iterative. The rest of
the proof is the same as for Theorem~3.1 in~\cite{PirogGibbons15}.
\end{proof}
\noindent Now, to obtain a similar result for Elgot monads, we can
easily combine Theorems~\ref{thm:elgot_is_retract}
and~\ref{thm:sandwich_iter} without having to verify the equational
properties by hand.

\begin{theorem}\label{thm:sandwich_elgot}
  With an adjunction as in Theorem~\ref{thm:sandwich_iter}, let $\BBS$
  be an Elgot monad on $\BD$. Then, the monad induced on the
  composite $USF$ is an Elgot monad.
\end{theorem}
\begin{proof}
By Theorem~\ref{thm:elgot_is_retract}, there exist a guarded
iterative monad $\BBT$ and an iteration-congruent retraction
$\rho \colon\BBT \to \BBS$. By Theorem~\ref{thm:sandwich_iter}, the monad
induced on $UTF$ is guarded iterative. Thus, it is enough to show that
$U \rho F : UTF \to USF$ is an iteration-congruent retraction.
\begin{citemize}
\item It is a retraction, since retractions are preserved by all functors.
\item To see that it is guarded, let $f\colon  X \to USFY$ be
  $\sigma$-guarded. By definition, this means that
  $\Phi^{\mone}f\colon  FX \to SFY$ is $\sigma$-guarded in $\BBS$. Since
  $\rho$ is a guarded retraction, it follows that
  $\upsilon(\Phi^{\mone}f)$, for $\rho$'s family of sections $\upsilon$,
  is also $\sigma$-guarded. By the fact that $\Phi$ is a natural
  isomorphism, we obtain
  $\upsilon(\Phi^{\mone}f) = \Phi^{\mone}((U\upsilon)f)$, hence, by
  definition, $(U\upsilon)f$ is also $\sigma$-guarded.
\item To see that $U \rho F$ is an iteration congruence, let us denote
  by $(\argument)^\istar$ the solution in~$\BBT$, and by
  $(\argument)^\iistar$ the solution in the monad on $UTF$. Let
  $f,g\colon  X \to_2 UTF(X+Y)$ be morphisms such that
  $(U\rho)f = (U\rho)g$. First, using this and the fact that $\Phi$ is
  a natural isomorphism, we obtain the following: 
\begin{displaymath}
\rho(\Phi^{\mone}f)
=
\Phi^{\mone}((U\rho)f)
=
\Phi^{\mone}((U\rho)g)
=
\rho(\Phi^{\mone}g)
\end{displaymath}
Thus, by the fact that $\rho$ is an iteration congruence, we obtain that $\rho(\Phi^{\mone}f)^\istar = \rho(\Phi^{\mone}g)^\istar$. Now, we check that $U\rho F$ is an iteration congruence:
\begin{flalign*}
&&(U\rho)f^\iistar
=&\; (U\rho)(\Phi(\Phi^{\mone}f)^\istar)\\
&&=&\; \Phi(\rho(\Phi^{\mone}f)^\istar)&\by{$\Phi$ is a natural isomorphism}\\
&&=&\; \Phi(\rho(\Phi^{\mone}g)^\istar)&\by{the above}\\
&&=&\; (U\rho)(\Phi(\Phi^{\mone}g)^\istar)&\by{$\Phi$ is a natural isomorphism}\\
&&=&\; (U\rho)g^\iistar &\tag*{\qedhere}
\end{flalign*}
\end{citemize}
\end{proof}

\newcommand{\CBUS}{\mathbf{CbUMet}}
\newcommand{\FD}{F_{\mathrm D}}
\newcommand{\UD}{U_{\mathrm D}}
\newcommand{\CPO}{\mathbf{Cpo}_\bot}
\newcommand{\FL}{F_{\mathrm L}}
\newcommand{\UL}{U_{\mathrm L}}

\begin{example}[From Metric to CPO-based Iteration]\label{expl:met-lifting}
  As an example exhibiting sandwiching as well as the setting of
  Theorem~\ref{thm:elgot_is_retract}, we compare two iteration
  operators on $\mathbf{Set}$ that arise from different fixed point
  theorems: Banach's, for complete metric spaces, and Kleene's, for
  complete partial orders, respectively. We obtain the first operator
  by sandwiching Escardo's \emph{metric lifting
    monad}~$\BBS$~\cite{Escardo99} in the adjunction between sets and
  bounded complete ultrametric spaces (which forgets the metric in one
  direction and takes discrete spaces in the other), obtaining a monad
  $\bar\BBS$ on $\Set$. Given a bounded complete metric space $(X,d)$,
  $S(X,d)$ is a metric on the set $(X\times\mathbb N)\cup\{\bot\}$. As
  we show in the appendix, $\BBS$ is guarded iterative if we define
  $f\colon (X,d)\to S(Y,d')$ to be $\sigma$-guarded if $k>0$ whenever
  $f(x)=(\sigma(y),k)$.  By Theorem~\ref{thm:sandwich_iter},
  $\bar\BBS$ is also guarded iterative (of course, this can also be
  shown directly). The second monad arises by sandwiching the identity
  monad on cpos with bottom in the adjunction between sets and cpos
  with bottom that forgets the ordering in one direction and adjoins
  bottom in the other, obtaining an Elgot monad $\BBL$ on Set
  according to Theorem~\ref{thm:sandwich_elgot}. The latter is
  unsurprising, of course, as $\BBL$ is just the maybe monad $LX=X+1$.

  The monad $\bar\BBS$ keeps track of the number of steps needed to
  obtain the final result. We have an evident extensional collapse map
  $\rho\colon\bar\BBS\to\BBL$, which just forgets the number of steps. One
  can show that $\rho$ is in fact an iteration-congruent retraction,
  so we obtain precisely the situation of
  Theorem~\ref{thm:elgot_is_retract}. Technical details are in the
  appendix.

\end{example}

\section{Related Work}\label{sec:related}
\newcommand{\brhd}{\text{\raisebox{-1.8pt}{\LARGE$\blacktriangleright$}}}

Alternatively to our guardedness relation on Kleisli morphisms,
guardedness can be formalized using type constructors~\cite{Nakano00}
or, categorically, functors, as in \emph{guarded fixpoint
  categories}~\cite{MiliusLitak17}. Roughly speaking, in such settings
a morphism $X\to Y+Z$ is guarded in~$Z$ if it factors through a
morphism $X\to Y+\blacktriangleright Z$ where $\blacktriangleright$ is
a functor or type constructor to be thought of as isolating the
guarded inhabitants of a type. The functorial approach, giving rise to
\emph{guarded fixpoint categories}, covers also total guardedness,
like we do. Our approach is slightly more fine-grained, and in
particular natively supports the two variants of the dinaturality
axiom (Figure~\ref{fig:ax}), which, e.g., in guarded fixpoint
categories require additional assumptions~\cite[Proposition
3.15]{MiliusLitak17} akin to the one we discuss in
Remark~\ref{rem:multi_guard}. In our own subsequent work, we have
generalized the notion of abstract guardedness from co-Cartesian to
symmetric monoidal categories~\cite{GoncharovSchroder18}, where
guardedness becomes a more symmetric concept: among morphisms
$X\otimes Y\to Z\otimes W$, where $\otimes$ is the monoidal structure,
one distinguishes morphisms that are (simultaneously) \emph{unguarded}
in the input~$A$ and \emph{guarded} in the output~$D$.

A result that resembles our Theorem~\ref{thm:nu-algebras}, due to
Ad\'amek et~al.~\cite{AdamekMiliusEtAl11}, states roughly that if
$\BC$ is locally finitely presentable and hyperextensive (a property
imposing certain compatibility constraints between pullbacks and
countable coproducts, satisfied, e.g., over sets and over complete
partial orders), then the finitary Elgot monads are the algebras for a
monad on the category of endofunctors given by
$H\mapsto L_H=\rho\gamma.\,(\argument+1+H\gamma)$ where $\rho$ takes
\emph{rational fixpoints} (i.e.\ final coalgebras among those where
every point generates a finite subcoalgebra); that is, in the
mentioned setting, finitary Elgot monads are monadic over
endofunctors. Besides Theorem~\ref{thm:nu-algebras} making fewer
assumptions on $\BC$, the key difference (indicated already in
Remark~\ref{exmp:delay}) is that, precisely by dint of the mentioned
result, $L_H$ is already a finitary Elgot monad (namely, the free
finitary Elgot monad over~$H$); contrastingly, we characterize Elgot
monads as quotients of \emph{guarded iterative} monads, i.e.\ of
monads where guarded recursive definitions have \emph{unique}
fixpoints.

\section{Conclusions and Further Work}\label{sec:concl}
We have given a unified account of monad-based guarded and unguarded
iteration by axiomatizing the notion of guardedness to cover standard
definitions of guardedness, and additionally, as a corner case, what
we call \emph{total guardedness}, i.e.\ the situation when all
morphisms are declared to be guarded. We thus obtain a common umbrella
for \emph{guarded iterative monads}, i.e.\ monads with unique iterates
of guarded morphisms, and Elgot monads, i.e.\ totally guarded monads
satisfying Elgot's classical laws of iteration. We reinforce the view
that the latter constitute a canonical model for monad-based unguarded
iteration by establishing the following equivalent characterizations:
Provided requisite final coalgebras exist, a monad $\BBT$ is Elgot iff
it satisfies one of the following equivalent conditions:
\begin{itemize}
 \item it satisfies the quasi-equational theory of iteration~\cite{AdamekMiliusEtAl10,GoncharovEA18} (definition);
 \item it is an iteration-congruent retract of a guarded iterative monad (Theorem~\ref{thm:elgot_is_retract});
 \item it is an algebra $(\BBT,\rho)$ of the monad $T\mapsto\nu\gamma.\,T(X+\gamma)$ in
the category of monads satisfying a natural delay cancellation condition (Theorem~\ref{thm:nu-algebras}).    
\end{itemize}
In future work, we aim to investigate further applications of this
machinery, in particular to examples which did not fit previous
formalizations. One prospective target is suggested by work of Nakata
and Uustalu~\cite{NakataUustalu15}, who give a coinductive big-step
trace semantics for a while-language. We conjecture that this work has
an implicit guarded iterative
monad~$\mathbb{T}$\kern.5pt\textscale{.86}{$\mathbb{R}$} under the
hood, for which guardedness cannot be defined using the standard
argument based on a final coalgebra structure of the monad because
$\mathbb{T}$\kern.5pt\textscale{.86}{$\mathbb{R}$} is not a final
coalgebra. Moreover, we aim to extend the treatment of iteration in
finite trace semantics via iteration-congruent retractions
(Example~\ref{expl:trace}) to infinite traces, possibly taking
orientation from recent work on coalgebraic infinite trace
semantics~\cite{UrabeHasuo15}.

In type theory, there is growing interest in forming an extensional
quotient of the delay
monad~\cite{ChapmanUustaluEtAl15,AltenkirchDanielssonEtAl17}. It is shown
in~\cite{ChapmanUustaluEtAl15} that under certain reasonable
conditions, a suitable collapse of the delay monad by removing delays
is again a monad; however, the proof is already quite complex, and
proving directly that the collapse is in fact an Elgot monad, as one
would be inclined to expect, seems daunting.  We expect that
Theorem~\ref{thm:nu-algebras} may shed light on this issue. A natural
question that arises in this regard is whether the subcategory of
$(\argument)^\nu$-algebras figuring in the theorem is reflexive. A
positive answer would provide a means of constructing canonical
quotients of $(\argument)^\nu$-algebras (such as the delay monad) with
the results automatically being Elgot monads.

\subsubsection*{Acknowledgements} We would like to thank the
anonymous referees for their thorough attention to the text and their useful comments on improving the 
presentation.

\bibliographystyle{myabbrv}
\bibliography{monads}

\clearpage
\appendix

\section{Details of Example~\ref{expl:met-lifting}}
As an example of the setting described in
Theorem~\ref{thm:elgot_is_retract}, we compare two iteration operators
on $\mathbf{Set}$ that arise from different fixed point theorems:
Banach's for complete metric spaces and Kleene's for complete partial
orders, respectively. The first operator keeps track of the number of
steps needed to obtain the final result. We show that its extensional
collapse, defined as a morphism that forgets the number of steps, is
an iteration-congruent retraction with the second operator as the
retract.

First, we consider the category $\CBUS$ of complete $1$-bounded
ultrametric spaces and nonexpansive maps. Note that $\CBUS$ has
coproducts and is Cartesian closed~\cite{Smyth92}. Following
Escard{\'o}~\cite{Escardo99}, one has a monad $T$ on $\CBUS$
given on objects as
$T(A,d) = ((A \times \mathbb N) \cup \{ \infty \}, d')$, with $d'$
given by
\begin{align*}
& d'(\infty, \infty) = 0
&&
d'((x,k), \infty) = d'(\infty, (x,k)) = (1/2)^{k}
\\
& d'((x,k),(y,k)) = (1/2)^{k} d(x,y)
&&
d'((x,k),(y,t)) = (1/2)^{\min(k,t)} \text{ if $k \neq t$}
\end{align*}
The monad structure on~$T$ is defined as expected, by
\begin{math}
\eta(a) = (a,0)
\end{math}
and
\begin{math}
f^\klstar(a,k) = (b, k+t) \text{ where $f (a) = (b, t)$}
\end{math}

\begin{theorem}\label{thm:met-delay-iterative}
The monad $T$ is guarded iterative with $f\colon  X \to TY$ being $\sigma$-guarded if for all $x$ and $y$,  $f(x) = (\sigma(y), k)$ implies $k > 0$.
\end{theorem}
\begin{proof}
  First, note that the product of $\tuple{A,d_A}$ and $\tuple{B, d_B}$ in $\mathbf{CbUMet}$ is given by $\tuple{A \times B, d_{A \times B}}$, where
\begin{equation*}
d_{A \times B}(\tuple{x_1,y_1}, \tuple{x_2,y_2}) = \max \{ d_A(x_1, x_2) , d_B(y_1, y_2) \}.
\end{equation*}
The exponential object is equal to $\tuple{B^A, d_{A \Rightarrow B}}$ where
\begin{equation*}
d_{A \Rightarrow B}(f,g) = \sup \{ d_B(f(x), g(x)) \ |\ x \in A \}.
\end{equation*}
The coproduct is given by $\tuple{A+B, d_{A+B}}$, where
\begin{equation*}
d_{A+B}(p,q) = 
\begin{cases}
d_A(x_1, x_2) & \text{if } p = \inl\, x_1 \text{ and } q =\inl\, x_2
\\
d_B(y_1, y_2) & \text{if } p = \inr\, y_1 \text{ and } q =\inr\, y_2
\\
1 & \text{otherwise}
\end{cases}
\end{equation*}

\noindent Now, we show that the monad $\BBT$ is guarded. The only nontrivial case is \textbf{(cmp)}. So, assume $([g,h]^\klstar f)(x) = (\sigma(y), k)$. We consider two cases:
\begin{itemize}
\item $f(x) = (\inr z, k')$. Then, since $f$ is $\inr$-guarded, $k' > 0$, so, by definition of~$(\argument)^\klstar$,  $k > 0$.
\item $f(x) = (\inl z, k')$. Then, $([g,h]^\klstar f)(x) = [g,h]^\klstar (\inl z, k') = (\sigma(y), k' + k'')$, where $g(z) = (\sigma(y), k'')$. Since $g$ is $\sigma$-guarded, $k'' > 0$, so $k = k' + k'' > 0$.
\end{itemize}

\noindent Given a guarded morphism $f\colon  X \to_1 T(Y+X)$, we define the morphism $f^\istar\colon  X \rightarrow TY$ as the unique fixed point of the following map $\psi : (X \rightarrow TY) \rightarrow (X \rightarrow TY)$:
\begin{align*}
\psi (g) &= [\eta,g]^\klstar f
\end{align*}

\noindent One can easily see that any fixed point of $\psi$ satisfies the fixed point identity, and that the uniqueness of such a fixed point gives us that $f$ has a unique solution. We use Banach's theorem to achieve both.

By Banach's theorem, it is enough to show that $\psi$ is contractive, that is, there exists a non-negative real $c < 1$ such that for maps $g, g'\colon  X \rightarrow TY$, the following holds:
\begin{equation}\label{eq:contractiveCimL}
d_{X \Rightarrow TY}(\psi(g), \psi(g')) \leq c \cdot d_{X \Rightarrow TY}(g, g')
\end{equation}
The left-hand side of the equation~\eqref{eq:contractiveCimL} is equal to:
\begin{align*}
d_{X \Rightarrow TY}(\psi(g), \psi(g')) = \sup \{ d_{TY}(\psi (g)(x), \psi(g')(x)) \ |\ x \in X \}
\end{align*}
In turn, the right-hand side is as follows:
\begin{align*}
c \cdot d_{X \Rightarrow TY}(g, g') & = c \cdot \sup \{ d_{TY}(g(x), g'(x)) \ |\ x \in X \}
\\
& = \sup \{ c \cdot d_{TY}(g(x), g'(x)) \ |\ x \in X \}
\end{align*}
Thus, it is enough to show that for all $x \in X$, there exists $y \in X$ such that:
\begin{equation*}
d_{TY}(\psi (g)(x), \psi(g')(x)) \leq c \cdot d_{TY}(g(y), g'(y))
\end{equation*}
We show this for $c = 1/2$. We consider two cases:
\begin{citemize}
\item $f(x) = (\inl y, k)$ for some $y \in Y$ and $k \in \mathbb N$. Then, for all $g\colon  X \to TY$, the following holds:
\begin{flalign*}
&&\psi(g)(x)
=&\; ([\eta, g]^\klstar f)(x) &\\
&&=&\; [\eta, g]^\klstar(\inl y, k) &\\
&&=&\; (y, k) 
\end{flalign*}
So, the following holds:
\begin{flalign*}
&&d_{TY}(\psi(g)(x),\psi(g')(x))
=&\; d_{TY}((y, k),(y, k)) &\\
&&=&\; 0 &\\
&&\leq&\; (1/2) \cdot d_{TY}(g(x), g'(x)) 
\end{flalign*}
\item $f(x) = (\inr y, k + 1)$ for some $y \in X$ and $k \in \mathbb N$ (the `$+1$' part follows from the fact that $f$ is guarded). Assume that $g(y) = (z, t)$ for some $z$ and $t$. Then
\begin{flalign*}
&&\psi(g)(x)
=&\; ([\eta, g]^\klstar f)(x) &\\
&&=&\; [\eta, g]^\klstar(\inr y, k + 1) &\\
&&=&\; (z, t + k + 1). 
\end{flalign*}
Similarly, let $g'(y) = (z', t')$, and so $\psi(g')(x) = (z', t' + k + 1)$. Then, it follows that:
\begin{flalign*}
&&d_{TY}(\psi(g)(x),\psi(g')(x))
=&\; d_{TY}((z, t+k+1),(z', t'+k+1)) &\\
&&=&\; (1/2)^{k+1} \cdot d_{TY}((z, t),(z', t')) &\\
&&=&\; (1/2)^{k} \cdot (1/2) \cdot d_{TY}((z, t),(z', t')) &\\
&&\leq&\; (1/2) \cdot d_{TY}((z, t),(z', t')) &\\
&&=&\; (1/2) \cdot d_{TY}(g(y),g'(y))  &\tag*{\qedhere}
\end{flalign*}
\end{citemize}
\end{proof}

\noindent We thus obtain a monad $\UD T \FD$ on $\mathbf{Set}$ by
sandwiching $T$ in the adjunction $\FD \dashv \UD$ where $\UD$ is the
forgetful functor $\CBUS\to\Set$ and $\FD$ takes discrete metrics. By
Theorem~\ref{thm:sandwich_iter}, $\UD T \FD$ is guarded iterative.

For the second operator, let $\CPO$ be the category of complete
partial orders and continuous bottom-preserving functions. The
identity on $\CPO$ is an Elgot monad, hence, by
Theorem~\ref{thm:sandwich_elgot}, we obtain an Elgot monad $\UL \FL$
on $\Set$ by sandwiching in the adjunction $\FL\dashv\UL$ where $\UL$
is the forgetful functor $\CPO\to\Set$ and $\FL$ adjoins bottom. The
relation between the two monads on $\Set$ is an instance of our notion
of iteration-congruent retraction:
\begin{theorem}
  Define $\rho \colon\UD T \FD \to UL \FL$ by $\rho(a, k) = a$ and
  $\rho(\infty) = \bot$.  Then~$\rho$ is an iteration-congruent
  retraction with the section given by $\upsilon(a) = (a,1)$ and
  $\upsilon(\bot) = \infty$.  Moreover, the respective iteration
  operators induced by $\rho$ and the sandwich theorem coincide.
\end{theorem}
\begin{proof}
It is trivial that $\rho$ is a guarded retraction. To see that it is a iteration
congruence, we first define an auxiliary relation: given two functions $f,
h\colon  X \to (B \times \mathbb N) \cup \{ \infty \}$, we write $f \sim h$ if $f
(x) = (a,k)$ for some $a \in B$, $k \in \mathbb N$ if and only if $h (x) =
(a,k')$ for some $k' \in \mathbb N$ and $f (x) = \infty$ if and only if $h(x) =
\infty$ (i.e.\ the two functions differ only in the number of steps needed to obtain the value). We also write $\psi_f(g) = [\eta, g]^\klstar f$ for the function $\psi$ from the proof of Theorem~\ref{thm:met-delay-iterative}.

Given a $2$-guarded function $f\colon  X \to ((Y + X) \times \mathbb N) \cup \{ \infty \}$, the function~$f^\istar$ can be defined as the unique fixed point of $\psi_f$ (see the proof of Theorem~\ref{thm:met-delay-iterative}), which by Banach's fixed-point theorem is given by the limit of the sequence $W^f_0 = c$ and $W^f_{(n+1)} = \psi_f(W^f_n)$, where $c$ is the constant function $c(x) = \infty$. It is easy to see that for each $x$ the sequence $W^f_n(x)$ stabilizes. Given a function $h$ such that $f \sim h$, it is easy to show by induction that for every $x$, the sequence $W^f_n(x)$ stabilizes with $(a,k)$ for some $k$ if and only if $W^h_n(x)$ stabilizes with $(a,k')$ for some $k'$ at the same index $n$. Then, for all $x$, $f^\istar(x) = (a,k)$ and $h^\istar(x) = (a, k')$, so $\rho(f^\istar(x)) = \rho (h^\istar(x))$. Then, the result is obtained by noticing that for all $f$ and $h$, $\rho f = \rho h$ implies $f \sim h$.

It is left to see that the solution operator that follows from the sandwich theorem and the one that follows from the iteration-congruent retraction coincide. Given $f\colon  X \to (X + Y) \cup \{ \bot \}$, its solution in the Elgot monad $\UL \FL$ is given by the fixed point of the equation $\phi(g) = [\eta,g]^\klstar f$, that is, by Kleene's theorem, by the limit of the sequence $W'_0 = c'$ and $W'_{(n+1)} = \phi_f(W'_n)$, where $c'$ is the constant function $c'(x) = \bot$. It is easy to see that $W'_n = \rho W^{(\upsilon f)}_n$, so the solutions coincide.
\end{proof}
\noindent Forgetting the provenance of the above-mentioned monads on
$\Set$ via sandwiching, we obtain that the maybe monad $(\argument)+\{\bot\}$
on $\Set$ is an iteration-congruent retract of the delay monad
$(\argument)\times\Nat+\{\bot\}$, which is, of course, not surprising. In
categories beyond sets (where the delay monad, or partiality monad, is
more generally defined as $\nu\gamma.(-+\gamma)$~\cite{Capretta05}),
the situation is more complex, see Remark~\ref{exmp:delay}.

\end{document}

